\documentclass[10pt,conference]{IEEEtran}

\usepackage[dvipsnames]{xcolor}
\usepackage{totcount}
\usepackage{refcount}
\usepackage{amsmath, amssymb, amsthm}
\usepackage{mathtools}
\usepackage{complexity}
\usepackage{todonotes}
\usepackage{bm}
\usepackage{hyperref}
\usepackage{newfile}
\usepackage{thm-restate}
\usepackage[capitalize]{cleveref}
\usepackage[inline,shortlabels]{enumitem}
\usepackage{tikz}
\usepackage{pgffor}
\usetikzlibrary{snakes, positioning, arrows.meta, patterns.meta}

\newtheorem{theorem}{Theorem}[section]

\newtheorem{lemma}[theorem]{Lemma}

\newtheorem{claim}[theorem]{Claim}
\newtheorem{proposition}[theorem]{Proposition}
\newtheorem{definition}[theorem]{Definition}

\theoremstyle{remark}
\newtheorem{example}[theorem]{Example}
\newtheorem{remark}[theorem]{Remark}

\crefname{appsec}{Appendix}{Appendices}

\usepackage[
        bibstyle=numeric-comp,  %
        firstinits=true,        %
        maxnames=9,             %
        backend=bibtex,          %
        useprefix=true,         %
        sortcites=true,         %
        url=false               %
]{biblatex}

\bibliography{references.bib}

\newcommand{\tuple}[1]{\langle #1 \rangle}
\newcommand{\set}[1]{\{#1\}}

\newcommand{\sset}{\subseteq}

\newcommand{\N}{\mathbb{N}}
	\newcommand{\NN}{\N}
\newcommand{\Q}{\mathbb{Q}}
	\newcommand{\QQ}{\Q}
\newcommand{\Z}{\mathbb{Z}}
	\newcommand{\ZZ}{\Z}
\newcommand{\compP}{\ComplexityFont{P}}
\newcommand{\B}{\mathbb{B}}
	\newcommand{\BB}{\B}
\newcommand{\MM}{\mathbb{M}}

\newcommand{\calS}{\mathcal{S}}
\newcommand{\calF}{\mathcal{F}}

\newcommand{\Mat}{\mathrm{Mat}}
\newcommand{\downclose}[1]{\mathord{\downarrow}#1}

\renewcommand{\vec}[1]{\bm{#1}}

\newcommand{\VASS}{\mathrm{VASS}}
\newcommand{\cA}{\mathcal{A}}
\newcommand{\cV}{\mathcal{A}}
\newcommand{\Aa}{\mathcal{A}}
\newcommand{\Vv}{\mathcal{V}}

\DeclareMathOperator{\Adj}{Adj}

\DeclareMathOperator{\Reach}{Reach}
\DeclareMathOperator{\AcReach}{AcReach}
\newcommand{\vasreach}{\Reach_{\VASS}}
\newcommand{\intreach}{\Reach_{\ZZ}} 
\newcommand{\natreach}{\Reach_{\NN}}
\DeclareMathOperator{\acintreach}{\AcReach_\ZZ}

\DeclareMathOperator{\Eff}{Eff}

\newcommand{\RUN}[5]{#1\xrightarrow[#3]{#2}_{#4}#5}
\newcommand{\Run}[4]{\RUN{#1}{#2}{}{#3}{#4}}

\DeclareMathOperator{\Cycles}{Cyc}
\DeclareMathOperator{\Runs}{Runs}
\newcommand{\config}[2]{\tuple{#1, #2}}

\newcommand{\generalstep}[2]{\xrightarrow{#1}_{#2}}
\newcommand{\stepVASS}{\generalstep{}{\mathrm{VASS}}}
\newcommand{\stepint}{\generalstep{}{\ZZ}}
\newcommand{\stepnat}{\generalstep{}{\NN}}
\newcommand{\stepsVASS}{\generalstep{*}{\mathrm{VASS}}}
\newcommand{\stepsint}{\generalstep{*}{\ZZ}}
\newcommand{\stepsnat}{\generalstep{*}{\NN}}

\newcommand{\covertable}[2]{T(#1,#2)}
\newcommand{\impl}[2]{``\labelcref{#1}$\Rightarrow$\labelcref{#2}''}

\newcommand{\density}{D}
\newcommand{\densityplus}{D^+}

\newcommand{\mysubsubsection}[1]{\noindent\textbf{#1}}

\newcommand{\compL}{\ComplexityFont{L}}

\begin{document}

\title{A Complexity Dichotomy for Semilinear Target Sets in Automata with One Counter}

\author{
\IEEEauthorblockN{Yousef Shakiba}
\IEEEauthorblockA{\textit{Max Planck Institute for Software} \\
        \textit{Systems (MPI-SWS)}\\
        Kaiserslautern, Germany \\
        \textit{yshakiba@mpi-sws.org}}
\and
\IEEEauthorblockN{Henry Sinclair-Banks}
\IEEEauthorblockA{\textit{University of Warsaw} \\
        Warsaw, Poland \\
        \href{http://henry.sinclair-banks.com}{\texttt{henry.sinclair-banks.com}}\\
        \textit{hsb@mimuw.edu.pl}
}
\and
\IEEEauthorblockN{Georg Zetzsche}
\IEEEauthorblockA{\textit{Max Planck Institute for Software} \\
        \textit{Systems (MPI-SWS)}\\
        Kaiserslautern, Germany \\
        \textit{georg@mpi-sws.org}}
}

\maketitle

\begin{abstract}
	In many kinds of infinite-state systems, the coverability problem has
significantly lower complexity than the reachability problem. In order to
delineate the border of computational hardness between coverability and
reachability, we propose to place these problems in a more general context,
which makes it possible to prove complexity dichotomies.

The more general setting arises as follows. We note that for coverability, 
we are given a vector $\vec{t}$ and are asked if there is a reachable
vector $\vec{x}$ satisfying the relation $\vec{x}\ge\vec{t}$. For
reachability, we want to satisfy the relation $\vec{x}=\vec{t}$. In the
more general setting, there is a Presburger formula $\varphi(\vec{t},\vec{x})$,
and we are given $\vec{t}$ and are asked if there is a reachable $\vec{x}$ with
$\varphi(\vec{t},\vec{x})$.

We study this setting for systems with one counter and binary updates: (i)~integer VASS, (ii)~Parikh automata, and (i)~standard (non-negative) VASS.
In each of these cases, reachability is \linebreak$\NP$-complete, but coverability is 
known to be in polynomial time. Our main results are three dichotomy theorems, 
one for each of the cases (i)--(iii). In each case, we show that for every 
$\varphi$, the problem is either $\NP$-complete or belongs to $\AC^1$, a 
circuit complexity class within polynomial time. We also show that it is 
decidable on which side of the dichotomy a given formula falls.

For (i) and (ii), we introduce novel density measures for sets of integer
vectors, and show an $\AC^1$ upper bound if the respective density of the set
defined by $\varphi$ is positive; and $\NP$-completeness otherwise. 
For (iii), the complexity border is characterized by a new notion of \emph{uniform quasi-upward closedness}. 
In particular, we improve the best known upper bound for coverability in (binary encoded) $1$-VASS from $\NC^2$ (as shown by Almagor, Cohen, P\'{e}rez, Shirmohammadi, and Worrell in 2020) to $\AC^1$.

\end{abstract}

\section{Introduction}
\label{sec:introduction}
\mysubsubsection{Reachability and coverability problems.}
Deciding reachability is one of the most fundamental tasks in the verification
of infinite-state systems. Here, we are given a configuration $\vec{t}$ of some
infinite-state system and want to decide whether there is a run from an
initial configuration that arrives in $\vec{t}$. While reachability algorithms
are powerful tools, the complexity of reachability is often exceedingly high.
For example, in \emph{vector addition systems with states} (VASS), which
feature a set of non-negative counters that can be incremented and decremented, the reachability problem is
Ackermann-complete~\cite{DBLP:conf/focs/Leroux21,DBLP:conf/focs/CzerwinskiO21,DBLP:conf/lics/LerouxS19}.
Moreover, even in VASS with one counter (and binary-encoded updates),
reachability is already $\NP$-complete~\cite{DBLP:conf/concur/HaaseKOW09}.

In contrast to reachability, deciding \emph{coverability} is usually much more
efficient.  In this setting, there is an ordering $\le$ on the set of
configurations, and we want to know whether for the given configuration
$\vec{t}$, there is a reachable configuration $\vec{x}$ with
$\vec{x}\ge\vec{t}$. For example, in VASS, coverability is
$\EXPSPACE$-complete (thus considerably lower than the Ackermann complexity of
reachability)~\cite{lipton1976reachability,Rackoff78}. In well-structured transition systems, reachability is often
undecidable (e.g.\ in VASS with resets~\cite{DBLP:conf/icalp/DufourdFS98}),
whereas coverability is decidable under mild
assumptions~\cite{DBLP:journals/iandc/AbdullaCJT00,finkel2001well}. In
fixed-dimensional VASS---where the number $d$ of counters is fixed---the complexity of reachability is still far from understood, but known to
be non-elementary in dimension $8$~\cite{DBLP:conf/lics/CzerwinskiO22}.
Coverability, however, is $\PSPACE$-complete in every dimension $d\ge
2$~\cite{rosier1986multiparameter,DBLP:conf/lics/BlondinFGHM15,DBLP:conf/icalp/FearnleyJ13}. In $1$-VASS (i.e.\ VASS with one counter), coverability is known to
be in $\NC^2$~\cite{AlmagorCPSW20}, a circuit complexity class within polynomial
time.  In systems with one integer counter (i.e.\ that is
allowed to go below zero), a classic result is that coverability is even in
$\AC^1$ (see \cref{reference-coverability-one-z-vass}), a subclass of $\NC^2$~\cite{Vollmer1999}; this is in stark contrast to the \hbox{$\NP$-completeness} of reachability for both integer \hbox{$1$-VASS} and standard \hbox{$1$-VASS}.
Moreover, deciding
coverablity is usually sufficient for practical verification tasks. Thus,
coverability is both pratically relevant and admits significantly more
efficient algorithms.

\vspace{2mm}
\mysubsubsection{What makes coverability easier?} 
Reachability and coverability are similar decision problems: In each
setting, we are given a vector $\vec{t}\in\ZZ^d$ and are asked whether
we can reach a configuration $\vec{x}$ that relates to $\vec{t}$ in a
particular way.
For reachability, we have the relation ``$\vec{x}=\vec{t}$'', whereas for
coverability, we have ``$\vec{x}\ge\vec{t}$''. This viewpoint raises the question: What
makes the relation ``$\vec{x}\ge\vec{t}$'' so much more amenable to analysis
than ``$\vec{x}=\vec{t}$''? Therefore, we propose to study reachability
problems as above, where the relation between $\vec{t}$ and $\vec{x}$ is
specified by an arbitrary---but fixed!---formula in Presburger arithmetic, the
first-order theory of the structure $\tuple{\ZZ;+,\le,0,1}$.  Presburger arithmetic
can define the relations ``$\vec{x}=\vec{t}$'' and ``$\vec{x}\ge\vec{t}$'', and
its expressiveness is well-understood.  This makes it a promising
candidate to explore the space around reachability and coverability.

\vspace{2mm}
\mysubsubsection{Understanding the complexity of Presburger target sets.} More
precisely (but still informally), we want to understand the following
type of reachability problems.  For each fixed formula $\varphi$ of Presburger
arithmetic and a class $\mathcal{C}$ of counter systems, we define the
\emph{$\varphi$-reachability problem} as follows. We are given (i)~a system
from $\mathcal{C}$ with $d$ counters and (ii)~a vector $\vec{t}\in\ZZ^p$. And
we are asked whether there is a reachable configuration with
counter values $\vec{x}\in\ZZ^d$ such that $\varphi(\vec{t},\vec{x})$ holds.
Note that we allow $p$ to differ from $d$, which provides even more flexibility
in specifying the set of allowed $\vec{x}$.

In this setting, our goal is to understand how the complexity of
$\varphi$-reachability depends on the formula $\varphi$. 

\vspace{2mm}
\mysubsubsection{Our contributions.}
We initiate this line of research by studying $\varphi$-reachability in systems
with one counter (with binary encodings). As mentioned above, in this context,
reachability is usually $\NP$-complete, whereas for coverability,
polynomial-time algorithms are available.  Our main results are \emph{three
complexity dichotomies} for Presburger target sets, each regarding one model of
one-dimensional counter systems.  Specifically, we consider the following
models of computation\footnote{Note that we do not consider \emph{one-counter
automata}, where zero tests are available, since in this setting, there is a
trivial dichotomy (see \cref{app:one-counter-automata})}:
\begin{itemize}
\item 1-dim.\;$\Z$-VASS: Updates and configurations are integers.
\item 1-dim.\;Parikh automata: All updates (and hence also configurations) are natural numbers.
\item 1-dim.\;VASS: Updates are integers, but only natural numbers are valid in configurations.
\end{itemize}
Each of our three dichotomies shows that for every $\varphi$,
$\varphi$-reachability for the respective model is either in $\AC^1$ or is \hbox{$\NP$-complete}, 
and provides a characterization for the respective formulas. Since we
always have $d=1$, the formula $\varphi$ always defines a subset of
$S\subseteq\ZZ^{p+1}$; the first $p$ components refer to parameter vectors $\vec{t} \in \ZZ^p$, and the right-most component carries the counter values $x\in\ZZ$.

\vspace{2mm}
\mysubsubsection{One-dimensional $\Z$-VASS and Parikh automata.}
For \emph{$1$-dim.\ $\ZZ$-VASS} and for \emph{$1$-dim.\ Parikh automata}, the
resulting decision problems are denoted $\intreach(S)$ and $\natreach(S)$,
respectively. 
We obtain dichotomies based on \emph{density measures}. We
introduce two density measures for subsets $S\subseteq\ZZ^{p+1}$: the
\emph{local density of $S$}, denoted $\density(S)$, and the \emph{local${}^+$
density}, denoted $\densityplus(S)$. Here, $\density(S)$ and $\densityplus(S)$
are real numbers in $[0,1]$. We show that if $\density(S)>0$, then
$\intreach(S)$ is in $\AC^1$, and otherwise, $\intreach(S)$ is $\NP$-complete.
Similarly, if $\densityplus(S)>0$, then $\natreach(S)$ is in $\AC^1$, and
otherwise, $\natreach(S)$ is $\NP$-complete.

Our local density measures are similar in spirit to (but different from) the
\emph{Shnirel'man density} from Additive Number Theory~\cite{Nathanson1996} (see
\cref{shnirelman-density}). Intuitively, $\density(S)$ measures the probability
of hitting $S$ when randomly picking a point in an interval (or
arithmetic progression) centered at a point in $S$.

\vspace{2mm}
\mysubsubsection{One-dimensional VASS.}
For \emph{$1$-VASS} (we also refer to this setting as ``VASS semantics''), we
denote $\varphi$-reachability as $\vasreach(S)$. We show that if $S$ is
``uniformly quasi-upward closed'' (defined precisely in \cref{sec:main-results}),
then $\vasreach(S)$ is in $\AC^1$. Moreover, in all other cases, $\vasreach(S)$
is $\NP$-complete. In particular, we show that coverability in $1$-VASS is in
$\AC^1$ (with a surprisingly simple algorithm!), improving upon the
previously best-known upper bound of $\NC^2$. 

Here, the border
between tractability and $\NP$-hardness is more subtle than expected. For
example, for $S_1=\{(t,x)\in\ZZ^2 \mid x=0 \vee x\ge t\}$ (i.e.\ we have $p=1$
and $\vec{t}$ consists of a single number $t$), and for $S_2=\{(t,x) \mid
\text{$x=0$ or $x$ is odd}\}$, we have $\NP$-completeness, whereas for $S_3=\{(t,x)\in\ZZ^2 \mid \text{$x$ even or ($x$ odd and $x\ge t$)}\}$ and $S_4=\{(t,x)\in\ZZ^2 \mid x=t-5 \vee x\ge t\}$, the problem is in $\AC^1$.

Finally, we show \emph{decidability} of the dichotomies. 
For each dichotomy, we show that given a Presburger formula for
$S\subseteq\ZZ^{p+1}$, it is decidable on which side of the dichotomy $S$
falls.

\vspace{2mm}
\mysubsubsection{Key ingredients.} 
Our results rely on several novel concepts and new algorithmic techniques.  
Our contributions for \hbox{$\Z$-VASS} and Parikh automata are substantially more challenging than for $1$-VASS. 
The case of Parikh automata can be
deduced from the one for $\Z$-VASS, making the $\Z$-VASS setting the most
involved result of this work. The first difficulty lies within concrete sets
$S\subseteq\ZZ^{p+1}$ for which it is not clear whether it is solvable in
polynomial time. For example, consider
\begin{align*}
	S_1&=\{(t,x)\in\ZZ^{1+1} \mid t\ge 0,~x\in[t,2t] \}.
\end{align*}
Then the complexity of $\natreach(S)$ is not clear: The interval $[t,2t]$ is
finite, and thus coverability techniques are not applicable. The interval
is also ``too wide'' (compared to its distance from $0$) to reduce, for example, to subset
sum. One of our new algorithmic approaches is to construct a semiring (defined by a set of equations) for
$S$ such that reachability in $\natreach(S)$ can again be decided in polynomial time
(and even in $\AC^1$) via repeated matrix squaring. The semiring construction also applies to
more complicated settings where $S$ contains multiple intervals, such as
$S'\subseteq\ZZ^{2+1}$, the set of all $\tuple{s,t,x}$ where 
\[ s,t\ge 0,~x\in[s+t,s+2t]\cup[s+3t,2s+4t]. \]
More generally, it applies whenever $x$ is constrained to belong to a
``$\rho$-chain of intervals'', where each interval's distance to the next (or
to zero) is at most $\rho$ times its own length. In $S'$, the interval
$[s+t,s+2t]$ has length $t$, and distance $t$ to the next interval
$[s+3t,2s+4t]$. The latter, in turn, has length $s+t$, and has distance
$s+3t\le 3(s+t)$ from zero.

Second, using ideal decompositions from WQO theory, we show that whenever
$\density(S)>0$, then $S$ decomposes into \hbox{$\rho$-chains} of intervals.
Conversely, we show that if $\density(S)=0$, then subset sum reduces to
$\intreach(S)$.

Third, we show that for every semilinear $S\subseteq\ZZ^{p+d}$, both
$\intreach(S)$ and $\natreach(S)$ reduce in logspace to the case of acyclic
automata. This even applies to higher dimensions $d$, and is crucial for
applying matrix squaring logarithmically many times. The reduction relies on
Eisenbrand \& Shmonin's Carath\'{e}odory theorem for integer
cones~\cite{EisenbrandS06}.

For the VASS semantics, the $\AC^1$ upper bounds rely on a new $\AC^1$ procedure for coverablity in $1$-VASS, improving upon the previously best-known upper bound of $\NC^2$.
The procedure is surprisingly simple and although we claim that our contributions for $\Z$-VASS are the most technically challenging, the correctness of our dichotomy for 1-VASS is not obvious. 
It relies on viewing coverability as computing weights for a weighted automaton over a newly introduced semiring.  
Deciding coverability in one-dim.\ $\Z$-VASS via weighted automata over the tropical semiring $(\ZZ\cup\{-\infty\},\max,+,0,1)$ is a classical technique\footnote{\label{reference-coverability-one-z-vass}For example, for the closely related shortest path problem for graphs with positive (binary-encoded) weights, this was mentioned in~\cite[remarks after Corollary 4.9]{cook1985taxonomy}, relying on~\cite[Lemma~5.11]{bookAHU74}. Coverability in integer $1$-VASS can be done in essentially the same way.}: 
In this approach, the automaton is turned into a matrix, for which the $n$-th power is computed, where $n$ is the number of states. 
Via repeated squaring, this leads to circuits of logarithmic depth, i.e.\ $\AC^1$.  
In our setting, we introduce a semiring $\calF$ that permits the same for $1$-VASS. However, matrix squaring in $\calF$ seems to require $\TC^0$, meaning repeated squaring would only yield a $\TC^1$ upper bound (see~\cref{remark-vass-matrix-product}; recall that
$\AC^1\subseteq\TC^1\subseteq\NC^2$). 
Instead, we include an additional approximation step to achieve the $\AC^1$ upper bound.

\section{Preliminaries}
\label{sec:preliminaries}
We use boldface for vectors, and square brackets to index vectors. 
For example, if $\vec{v} \in \ZZ^k$, then $\vec{v} = \tuple{\vec{v}[1], \ldots, \vec{v}[k]}$.
Given $\vec{u} \in \ZZ^{k_1}$ and $\vec{v} \in \ZZ^{k_2}$, we use the shorthand $\tuple{\vec{u}, \vec{v}}$ to denote the vector $\tuple{\vec{u}[1], \ldots, \vec{u}[k_1], \vec{v}[1], \ldots, \vec{v}[k_2]} \in \ZZ^{k_1+k_2}$.

A \emph{$\ZZ$-weighted automaton} $\cA = \tuple{Q, T, q_0, q_1}$ is an automaton where $Q$ is a finite set of \emph{states}, $T \subseteq Q \times \ZZ \times Q$ is a finite set of \emph{transitions}, and $q_0, q_1 \in Q$ are the \emph{initial} and \emph{final state}, respectively. 
Here, for a transition $\tuple{q, w, r} \in T$, the number $w$ is the \emph{weight} of the transition. 
Transition weights are always encoded in binary.

We define three semantics. A \emph{configuration} is a pair $\tuple{q, x}$ for some $q \in Q$ and $x\in\ZZ$. 
For configurations $\tuple{q,x},\tuple{r,y}$, we write $\tuple{q,x} \stepint \tuple{r,y}$ if there a transition $\tuple{q,w,r} \in T$ with $y=x+w$. 
Moreover, $\tuple{q,x} \stepVASS \tuple{r,y}$ means $\tuple{q,x} \stepint \tuple{r,y}$ and also $x,y\ge 0$. 
Finally, $\tuple{q,x}\stepnat\tuple{r,y}$ means $\tuple{q,x}\stepVASS\tuple{r,y}$ and also $x\le y$.
These step relations are called \emph{integer semantics} ($\stepint$), \emph{natural semantics} ($\stepnat$), and \emph{VASS semantics} ($\stepVASS$). By $\stepsint$, we denote the reflexive, transitive closure of $\stepint$ (analogously for $\stepsVASS$ and $\stepsnat$).

The decision problems we study are parameterized by a set $S \in \ZZ^{p+1}$ of
vectors. The input will consist of a $\ZZ$-weighted automaton and a vector
$\vec{t} \in \ZZ^p$, and we want to know whether we can reach a number in
the set
\begin{equation*}
	S[\vec{t}] \coloneqq \set{x \in \ZZ : \tuple{ \vec{t}, x } \in S }.
\end{equation*}
With this, we can define the decision problems. Let $S\subseteq\ZZ^{p+1}$.
Then $\intreach(S)$ is the following problem.
\begin{description}
	\item[Given.] A $\Z$-weighted automaton $\cA = \tuple{Q, T, q_0, q_1}$ and a vector $\vec{t} \in \ZZ^p$ in binary.
	\item[Question.] Is there an $x \in S[\vec{t}]$ with $\tuple{q_0, 0} \stepsint \tuple{q_1, x}$?
\end{description}
The problems $\natreach(S)$ and $\vasreach(S)$ are defined in the same way, except
that $\stepsint$ is replaced with $\stepsnat$ and $\stepsVASS$, respectively.
Note that by definition of $\natreach(S)$, we can only use transitions with non-negative weights.
In addition, for $\stepsVASS$, this means for $\vasreach(S)$, we only consider runs where the counter value remains non-negative. 

\emph{Presburger arithmetic} is the first-order theory of the structure
$\tuple{\ZZ; +, <, 0, 1}$. Sets definable in Presburger arithmetic are
precisely the semilinear sets~\cite{ginsburg1966semigroups}, but we will not
need the definition of semilinearity (except when we will define it
explicitly), so we will use the terms ``Presburger-definable'' and
``semilinear'' interchangeably. Note that almost all algorithms in this work
will assume a \emph{fixed} Presburger-definable set $S$, the algorithmic
complexity of changing between representations of semilinear sets will not be
relevant.

We assume basic familiarity with circuit complexity (however, no circuit
complexity is needed to see that our $\AC^1$ procedures work in polynomial
time). We will work with $\AC^i$, the class of decision problems solvable by
polynomial-sized \emph{unbounded} fan-in Boolean circuits of depth $(\log n)^i$
(for input size $n$), and sometimes $\NC^i$, the class of problems decidable by
polynomial-sized \emph{bounded} fan-in circuits of depth $(\log n)^i$ (for
input size $n$). Finally, $\TC^i$ is the class of problems solvable by polynomial-sized unbounded fan-in circuits with Boolean gates \emph{and majority gates}, of depth $(\log n)^i$ (for input size $n$). Recall the inclusions~\cite[Fig 4.6]{Vollmer1999}:
\[ \AC^0\subseteq\TC^0\subseteq\NC^1\subseteq\compL\subseteq\NL\subseteq\AC^1\subseteq\TC^1\subseteq\NC^2. \]
We will not mention details explicitly, but our $\AC^1$ upper bounds refer to (deterministic) logspace uniform $\AC^1$.
In our algorithms, the only knowledge we use is that additions and comparisons of pairs of binary-encoded integers can be computed in $\AC^0$~\cite[Thm~1.15]{Vollmer1999}.

\section{Main results}
\label{sec:main-results}
\subsection{Integer semantics}
For a subset $A\subseteq\ZZ$ and $k\in\NN$, we define the \emph{$k$-density of $A$ at point $x\in\ZZ$} as
\begin{align*}
	\density_k(A,x) \coloneqq \inf_{n\ge 1} \frac{| A \cap (x+k\cdot[-n,n]) |}{2n+1}
\end{align*}
and the fraction inside the infimum is also denoted $\density_{k,n}(A,x)$.
Thus, $\density_1(A,x)$ is the infimum of the probabilities of hitting $A$ when choosing points in intervals centered at $x$. And $\density_k(A,x)$ is the variant where instead of intervals, we consider arithmetic progressions with period $k$, radiating outward from $x$.

\begin{remark}\label{shnirelman-density}
	For $A\subseteq\NN$, $d(A)\coloneqq\inf_{n\ge 1} |A\cap[1,n]|/n$ is
	the \emph{Shnirel'man density} of $A$, a tool in Additive
	Number Theory for showing that every natural can be written as a sum of
	boundedly many naturals of a particular set~\cite[Section 7.4]{Nathanson1996}. 
\end{remark}

For $A \neq \emptyset$ as a whole, we define its \emph{local density} $\density(A)$ as
\[ \density(A)\coloneqq\inf_{k\in\NN} \inf_{x\in A} \density_k(A,x). \]
Moreover, we set $\density(\emptyset)\coloneqq 1$.
More generally, for a subset $S\subseteq\ZZ^{p+1}$, we define its \emph{local density} $\density(S)$ as
\[ \density(S) \coloneqq \inf_{\vec{t}\in\ZZ^p} \density(S[\vec{t}]). \]
We say that $S$ \emph{has positive local density} if $\density(S)>0$.

\pagebreak

\begin{theorem}\label{main-result-integer}
	Let $S\subseteq\Z^{p+1}$ be semilinear. Then:
\begin{enumerate}[(1)]
	\item If $\density(S)>0$, then $\Reach_\Z(S)$ is in $\AC^1$.
	\item Otherwise, $\Reach_\Z(S)$ is $\NP$-complete.
\end{enumerate}
\end{theorem}

Let us see some examples.

\begin{example}\label{exa:z1}
	Let $S_1 \coloneqq -2\cdot\NN \cup (2\cdot\NN+1)$. 
	In other words, $S_1$ contains the non-positive even integers and the positive odd integers. 
	It is not hard to see that $\intreach(S_1)$ belongs to $\AC^1$, because we can check whether a given automaton can reach a negative even counter value by slightly modifying the $\AC^1$ procedure for coverability. 
	Likewise, we can check if it can reach a positive odd counter value. 
	In~\cref{app:z-density-properties}, we show how this follows from~\cref{main-result-integer} (see~\cref{clm:z1}).
\end{example}

\begin{example}\label{exa:z2}
	Let $S_2 \coloneqq 2\cdot\ZZ \cup \set{1}$.
	In other words $S_2$ contains the even integers and $1$.
	Then $\density_2(S_2, 1) = 0$, and thus $\density(S_2)=0$.
	Here, $\NP$-hardness is not difficult to see.
\end{example}

\begin{example}\label{exa:z3}
	Let $S_3[t] = (-\infty,0] \cup [t,2t]$. 
	The density of $S_3$ is at least $\frac{1}{4}$ (see~\cref{fig:example-z3}); accordingly, $\intreach(S_3)$ is in $\AC^1$.
\end{example}

\begin{example}\label{exa:z4}
	Consider $S_4 \subseteq \ZZ^{2+1}$, where for $t,s\ge 0$, 
	\begin{equation*}
		S_4[t,s] \coloneqq (-\infty,0]~\cup~[t+s,2t+2s]~\cup~[3t+2s,4t+2s]
	\end{equation*}
	and $S_4[t,s] = \emptyset$ otherwise. 
	Thus, for $t,s\ge 0$, $S_4[t,s]$ contains all negative numbers, then there is a gap of size $t+s$, before an interval of size $t+s$, then a gap of size $t$, and then an interval of size $t$.
	The density of $S_4$ is at least $\frac{1}{4}$ (see~\cref{fig:example-z4}), which implies that $\intreach(S_4)$ is in $\AC^1$.
\end{example}

\begin{figure}
	\centering
	\begin{tikzpicture}
 	\foreach \x in {0,-1}
    {
    	\draw[gray, {Stealth[width=1.5mm, length=2mm, black]}-{Stealth[width=1.5mm, length=2mm, black]}, dashed, dash pattern={on 5pt off 1pt}, dash phase = 1.4, line width = 0.2mm] (-1.8, \x) -- (6.8, \x);
		\node at (-1.6, \x-0.3) {\scriptsize$-\infty$};
		\node at (6.7, \x-0.3) {\scriptsize$\infty$};

		\draw[line width = 0.4mm] (-1.7,\x) -- (0.8, \x);

		\draw[line width = 0.2mm] (0.8,\x+0.1) -- (0.8, \x-0.1);
		\node at (0.8, \x-0.3) {\scriptsize$0$};

		\draw[line width = 0.2mm] (2.25,\x+0.1) -- (2.25, \x-0.1);
		\node at (2.25, \x-0.3) {\scriptsize$t$};

		\draw[line width = 0.4mm] (2.25,\x) -- (3.7, \x);

		\draw[line width = 0.2mm] (3.7,\x+0.1) -- (3.7, \x-0.1);
		\node at (3.7, \x-0.3) {\scriptsize$2t$};

    }

	\node[red, fill, circle, inner sep = 1.5pt] at (0.7, 0) {};
	\node[red] at (0.7, 0.25) {\footnotesize$x_1$}; 
	\node[red] at (-0.6, 0.25) {\footnotesize$x_1\!-\!kn$}; 
	\node[red] at (2, 0.25) {\footnotesize$x_1\!+\!kn$}; 
	\fill[red, opacity = 0.3] (-0.7, -0.05) rectangle (2.1, 0.05);
	\draw[red, line width = 0.2mm] (-0.7, -0.07) -- (-0.7, 0.07);
	\draw[red, line width = 0.2mm] (-0.7, 0.06) -- (-0.67, 0.06);
	\draw[red, line width = 0.2mm] (-0.7, -0.06) -- (-0.67, -0.06);
	\draw[red, line width = 0.2mm] (2.1, -0.07) -- (2.1, 0.07);
	\draw[red, line width = 0.2mm] (2.1, 0.06) -- (2.07, 0.06);
	\draw[red, line width = 0.2mm] (2.1, -0.06) -- (2.07, -0.06);

	\node[blue, fill, circle, inner sep = 1.5pt] at (3.6, -1) {};
	\node[blue] at (3.6, -0.75) {\footnotesize$x_2$};
	\node[blue] at (1, -0.75) {\footnotesize$x_2\!-\!kn$};
	\node[blue] at (6.2, -0.75) {\footnotesize$x_2\!+\!kn$};
	\fill[blue, opacity = 0.3] (0.9, -0.95) rectangle (6.3, -1.05);
	\draw[blue, line width = 0.2mm] (0.9, -1.07) -- (0.9, -0.93);
	\draw[blue, line width = 0.2mm] (0.9, -0.94) -- (0.93, -0.94);
	\draw[blue, line width = 0.2mm] (0.9, -1.06) -- (0.93, -1.06);
	\draw[blue, line width = 0.2mm] (6.3, -1.07) -- (6.3, -0.93);
	\draw[blue, line width = 0.2mm] (6.3, -0.94) -- (6.27, -0.94);
	\draw[blue, line width = 0.2mm] (6.3, -1.06) -- (6.27, -1.06);
\end{tikzpicture}
	\vspace{-0.2in}
	\caption{
		An illustration for~\cref{exa:z3}.
		The solid black and dashed grey horizonal line segments represent the set $S_3[t]$ and its complement, respectively.
		The top picture shows that $\density_k(S_3[t], \textcolor{red}{x_1}) \geq \frac{1}{2}$ for any (\textcolor{red}{red}) point $\textcolor{red}{x_1} \in (-\infty, 0]$.
		The bottom picture shows that $\density_k(S_3[t], \textcolor{blue}{x_2}) \geq \frac{1}{4}$ for any (\textcolor{blue}{blue}) point $\textcolor{blue}{x_2} \in [t,2t]$.
	}
	\label{fig:example-z3}
\end{figure}
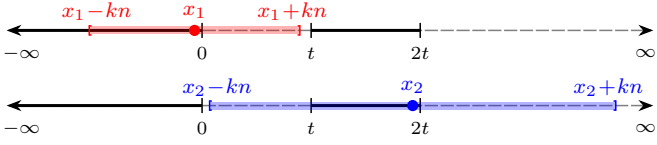

\begin{figure}
	\centering
	\begin{tikzpicture}
 	\foreach \x in {0,-1,-2}
    {
    	\draw[gray, {Stealth[width=1.5mm, length=2mm, black]}-{Stealth[width=1.5mm, length=2mm, black]}, dashed, dash pattern={on 5pt off 1pt}, dash phase = 1.4, line width = 0.2mm] (-1.8, \x) -- (6.8, \x);
		\node at (-1.6, \x-0.3) {\scriptsize$-\infty$};
		\node at (6.7, \x-0.3) {\scriptsize$\infty$};

		\draw[line width = 0.4mm] (-1.7,\x) -- (0, \x);

		\draw[line width = 0.2mm] (0,\x+0.1) -- (0, \x-0.1);
		\node at (0, \x-0.3) {\scriptsize$0$};

		\draw[line width = 0.2mm] (1.4,\x+0.1) -- (1.4, \x-0.1);
		\node at (1.4, \x-0.3) {\scriptsize$t\!+\!s$};

		\draw[line width = 0.4mm] (1.4,\x) -- (2.8, \x);

		\draw[line width = 0.2mm] (2.8,\x+0.1) -- (2.8, \x-0.1);
		\node at (2.8, \x-0.3) {\scriptsize$2t\!+\!2s$};

		\draw[line width = 0.2mm] (3.8,\x+0.1) -- (3.8, \x-0.1);
		\node at (3.8, \x-0.3) {\scriptsize$3t\!+\!2s$};

		\draw[line width = 0.4mm] (3.8,\x) -- (4.8, \x);

		\draw[line width = 0.2mm] (4.8,\x+0.1) -- (4.8, \x-0.1);
		\node at (4.8, \x-0.3) {\scriptsize$4t\!+\!2s$};

    }

	\node[red, fill, circle, inner sep = 1.5pt] at (-0.1, 0) {};
	\node[red] at (-0.1, 0.25) {\footnotesize$x_1$}; 
	\node[red] at (-1.4, 0.25) {\footnotesize$x_1\!-\!kn$}; 
	\node[red] at (1.2, 0.25) {\footnotesize$x_1\!+\!kn$}; 
	\fill[red, opacity = 0.3] (-1.5, -0.05) rectangle (1.3, 0.05);
	\draw[red, line width = 0.2mm] (-1.5, -0.07) -- (-1.5, 0.07);
	\draw[red, line width = 0.2mm] (-1.5, 0.06) -- (-1.47, 0.06);
	\draw[red, line width = 0.2mm] (-1.5, -0.06) -- (-1.47, -0.06);
	\draw[red, line width = 0.2mm] (1.3, -0.07) -- (1.3, 0.07);
	\draw[red, line width = 0.2mm] (1.3, 0.06) -- (1.27, 0.06);
	\draw[red, line width = 0.2mm] (1.3, -0.06) -- (1.27, -0.06);

	\node[ForestGreen, fill, circle, inner sep = 1.5pt] at (1.9, -1) {};
	\node[ForestGreen] at (1.9, -0.75) {\footnotesize$x_2$};
	\node[ForestGreen] at (0.2, -0.75) {\footnotesize$x_2\!-\!kn$};
	\node[ForestGreen] at (3.6, -0.75) {\footnotesize$x_2\!+\!kn$};
	\fill[ForestGreen, opacity = 0.3] (0.1, -0.95) rectangle (3.7, -1.05);
	\draw[ForestGreen, line width = 0.2mm] (0.1, -1.07) -- (0.1, -0.93);
	\draw[ForestGreen, line width = 0.2mm] (0.1, -0.94) -- (0.13, -0.94);
	\draw[ForestGreen, line width = 0.2mm] (0.1, -1.06) -- (0.13, -1.06);
	\draw[ForestGreen, line width = 0.2mm] (3.7, -1.07) -- (3.7, -0.93);
	\draw[ForestGreen, line width = 0.2mm] (3.7, -0.94) -- (3.67, -0.94);
	\draw[ForestGreen, line width = 0.2mm] (3.7, -1.06) -- (3.67, -1.06);

	\node[blue, fill, circle, inner sep = 1.5pt] at (4.7, -2) {};
	\node[blue] at (4.7, -1.75) {\footnotesize$x_3$};
	\node[blue] at (3.0, -1.7) {\footnotesize$x_3\!-\!kn$};
	\node[blue] at (6.4, -1.7) {\footnotesize$x_3\!+\!kn$};
	\fill[blue, opacity = 0.3] (2.9, -1.95) rectangle (6.5, -2.05);
	\draw[blue, line width = 0.2mm] (2.9, -2.07) -- (2.9, -1.93);
	\draw[blue, line width = 0.2mm] (2.9, -1.94) -- (2.93, -1.94);
	\draw[blue, line width = 0.2mm] (2.9, -2.06) -- (2.93, -2.06);
	\draw[blue, line width = 0.2mm] (6.5, -2.07) -- (6.5, -1.93);
	\draw[blue, line width = 0.2mm] (6.5, -1.94) -- (6.47, -1.94);
	\draw[blue, line width = 0.2mm] (6.5, -2.06) -- (6.47, -2.06);

\end{tikzpicture}
	\vspace{-0.2in}
	\caption{
		An illustration for~\cref{exa:z4}.
		The solid black and dashed grey horizonal line segments represent the set $S_4[t]$ and its complement, respectively.
		The top picture shows that $\density_k(S_4[t], \textcolor{red}{x_1}) \geq \frac{1}{2}$ for any (\textcolor{red}{red}) point $\textcolor{red}{x_1} \in (-\infty, 0]$.
		The middle picture shows that $\density_k(S_4[t], \textcolor{Green}{x_2}) \geq \frac{1}{3}$ for any (\textcolor{Green}{green}) point $\textcolor{Green}{x_2} \in [t+s, 2t+2s]$.
		The bottom picture shows that $\density_k(S_4[t], \textcolor{blue}{x_3}) \geq \frac{1}{4}$ for any (\textcolor{blue}{blue}) point $\textcolor{blue}{x_3} \in [3t+2s, 4t+2s]$.
	}
	\label{fig:example-z4}
\end{figure}
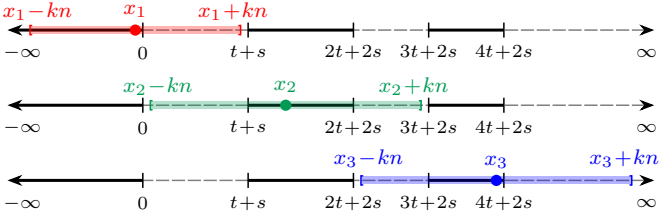

\begin{example}\label{exa:z5}
	If we remove the interval of length $(t+s)$ from $S_4[t, s]$ in~\cref{exa:z4}, we get $S_5 \subseteq \Z^{2+1}$, where for $t,s\ge 0, S_5[t,s] \coloneqq (-\infty,0]~\cup~[3t+2s,4t+2s]$	and $S_5[t,s] = \emptyset$ otherwise. 
	However, unlike $S_4$, the density of $S_5$ is $0$; since $\density_1(S_5[t,s], 3t+2s)\le \tfrac{t}{4t+2s} \le \tfrac{1}{4+2s/t}$. 
	Thus, if $s$ is large compared to $t$, the density of $S_5[t,s]$ becomes arbitrarily small. 
	Hence, $\intreach(S_5)$ is $\NP$-complete.
\end{example}

There are some additional noteworthy density properties.
The first is that positivity of the density of a set is not preserved by containment.
In other words, for two sets $S, S' \sset \ZZ^{p+1}$, if $\density(S) > 0$ and $S' \supseteq S$, then it is \emph{not} necessarily true that $\density(S') > 0$.
See~\cref{exa:z6} in~\cref{app:z-density-properties}.
Furthermore, intersection does not preserve positive density; even if $\density(S), \density(S') > 0$, it is \emph{not} necessarily true that $\density(S \cap S') > 0$.
However, union does preserve positive density.
See~\cref{exa:z7} in~\cref{app:z-density-properties}.

\begin{theorem}\label{main-result-integer-decidable}
	Given a semilinear $S\subseteq\ZZ^{p+1}$, it is decidable whether $\density(S)>0$.
\end{theorem}

\subsection{Natural semantics}
For a subset $A\subseteq\NN$ and $k\ge 1$, we define the \emph{local${}^+$ \hbox{$k$-density}} of $A$ at point $x\in\NN$ as
\begin{equation*}
	\densityplus_k(A,x) \coloneqq 
	\inf_{n\in\NN,~x-kn\ge 0} \frac{| A \cap (x + k\cdot[-n,n]) |}{2n+1}
\end{equation*}
and the fraction inside the infimum is also denoted $\densityplus_{k,n}(A,x)$, if $x-kn\ge 0$.
Thus, in contrast to $\density_k(A,x)$, the infimum only considers bounded arithmetic progressions $x + k\cdot[-n,n]$ that stay within $\NN$.
For $A\subseteq\NN$, $A\ne\emptyset$, as a whole, we define its \emph{local${}^+$ density} $\densityplus(A)$ as
\[ \densityplus(A) \coloneqq \inf_{k\in\NN} \inf_{x\in A} \densityplus_k(A,x). \]
Moreover, we set $\densityplus(\emptyset)\coloneqq 1$.
More generally, for a subset $S\subseteq\ZZ^p\times\NN$, we define its \emph{local${}^+$ density} $\densityplus(S)$ as
\[ \densityplus(S) \coloneqq \inf_{\vec{t}\in\ZZ^p} \densityplus(S[\vec{t}]). \]
We say that $S$ \emph{has positive local${}^+$ density} if $\densityplus(S)>0$.

\begin{theorem}\label{main-result-natural}
	Let $S\subseteq\ZZ^{p}\times\N$ be semilinear. Then:
	\begin{enumerate}[(1)]
		\item If $\densityplus(S)>0$, then $\natreach(S)$ is in $\AC^1$.
		\item Otherwise, $\natreach(S)$ is $\NP$-complete.
	\end{enumerate}
\end{theorem}

\begin{example} \label{exa:n}
	Consider $S\subseteq\NN^{p+1}$ with $p=1$ where \linebreak $S[t] \coloneqq [t,2t]$, then $\natreach(S)$ is in $\AC^1$ because, for every $x\in S[t], \densityplus_k(S[t],x)\ge\tfrac{1}{4}$.
\end{example}

In fact, when $S$ has no parameters, $\natreach(S)$ is always in $\AC^1$; see~\cref{clm:no-parameters-nat} which is proved in~\cref{app:n-density-properties}.
This is in contrast to the integer semantics: 
For \hbox{$S\subseteq\ZZ$} with $S=\set{0}$, the problem $\intreach(S)$ is \hbox{$\NP$-complete}.

\begin{restatable}{claim}{noParametersNat}\label{clm:no-parameters-nat}
	Let $S\subseteq\ZZ^{0}\times\N$ (i.e. when $p = 0$); then $\natreach(S)$ is in $\AC^1$.
\end{restatable}

\begin{theorem}\label{main-result-natural-decidable}
	Given a semilinear set $S\subseteq\ZZ^{p+1}$, it is decidable whether
	$\densityplus(S)>0$.
\end{theorem}

\subsection{VASS Semantics}

To state our dichotomy for when VASS semantics are used, we need some terminology.  
For $\delta\ge 1$, we say that $A\sset\NN$ is \emph{$\delta$-upward closed} if
$A+\delta\sset A$, where $A+\delta\coloneqq \{x+\delta \mid x\in A\}$.  
For $M\ge 0$ and $\delta\ge 1$, we say that $A\sset\NN$ is
\emph{$\tuple{\delta,M}$-upward closed} if there is a set $F\sset\NN$ with $|F|\le M$
such that $A\cup F$ is $\delta$-upward closed.  Further, $A$ is
\emph{quasi-upward closed} if there are $\delta\ge 1$ and $M\ge 0$ such that
$A$ is $\tuple{\delta,M}$-upward closed.  A set $S\subseteq\ZZ^{p}\times\N$ is
\emph{uniformly quasi-upward closed} if there are $\delta\ge 1$ and $M\ge 0$
such that for every $\vec{t}\in\ZZ^p$, the set $S[\vec{t}]$ is
$\tuple{\delta,M}$-upward closed.
\begin{theorem} \label{main-result-vass}
	Let $S\subseteq\ZZ^{p}\times\N$ be semilinear.
	\begin{enumerate}[(1)]
		\item \label{itm:vas-p-new}
			If $S$ is uniformly quasi-upward closed, then $\vasreach(S)$ is in $\AC^1$.
		\item \label{itm:vas-np-new}
		Otherwise, $\vasreach(S)$ is \NP-complete.
	\end{enumerate}
\end{theorem}

It may not seem surprising that a variant of upward closedness describes the polynomial-time cases. 
Therefore, the following examples illustrate some
subtleties in delineating the exact complexity border.

\begin{example} \label{exa:v1}
	Consider the set $S_1\subseteq\N^{p+1}$ with $p=1$ and $S_1[t]:=\{0\}\cup\{x\in\NN \mid x\ge t\}$. 
	Then $\vasreach(S_1)$ is $\NP$-complete (see~\cref{ex:vas-np-hard} on~\cpageref{ex:vas-np-hard}).
\end{example}

Note that $S_1[t]$ is $t$-upward closed for each $t\in\NN$. 
In particular, to achieve polynomial time, it does not suffice to require that for every $\vec{t} \in \ZZ^p$, the set $S[\vec{t}]$ is $\delta$-upward closed for some $\delta$. This is why we require one $\delta$ for all $\vec{t}$.

\Cref{exa:v1} also shows that it does not suffice to require a single $\delta\ge 1$ such that for every $\vec{t} \in \ZZ^p$, there is a finite set $F$ so that $S[\vec{t}]\cup F$ is $\delta$-upward closed.
Indeed, for every $t \in \ZZ$, the set $S_1[t]\cup [0,t]$ is upward closed, yet $\vasreach(S_1)$ is \hbox{$\NP$-complete}. 
This is why we need the global cardinality bound $M$.

\begin{example}\label{exa:v2}
	Consider the two sets $S_2, S_2' \sset \ZZ^{2}\times\NN$ where
	\begin{align*} 
		S_2[t_1,t_2] & \coloneqq \{x\in\NN \mid x\ge t_1\}\setminus \{t_2\}, \text{ and }\\
		S_2'[t_1,t_2] & \coloneqq \{x\in\NN\mid x\ge t_1\}\cup \{t_2\}. 
	\end{align*}
	Then, $\vasreach(S_2)$ is still in $\AC^1$, but $\vasreach(S_2')$ is $\NP$-complete (since it slightly generalizes~\cref{exa:v1}).
\end{example}

\begin{example}\label{exa:v3}
	The previous example sets $S_1$ and $S_2'$ have shown that hardness ensues when there are points with unbounded gaps above them. 
	However, such gaps can also ``hide within residue classes''.
	Consider the set $S_3 \sset \ZZ^{1+1}$ where
	\begin{equation*}
		S_3[t] \coloneqq \{x\in\NN \mid \text{$x$ is even}\} \cup \{1\}.
	\end{equation*}
	On first inspection, $S_4[t]$ only has small gaps:
	Every point has another within distance $2$. 
	However, because of the odd number $1\in S_3[t]$, it turns out that \hbox{$\tuple{\delta,M}$-upward} closedness can hold neither with an even $\delta$ nor with an odd $\delta$.
	Accordingly, $\vasreach(S_3)$ is $\NP$-complete.
\end{example}

\begin{theorem}\label{main-result-vass-decidable}
	Given a semilinear $S\subseteq\ZZ^p\times\N$, it is decidable whether $S$ is uniformly quasi-upward closed.
\end{theorem}

\section{General techniques}
\label{sec:general-techniques}
In this short section, we collect some facts that are used across the later sections.
First, the $\NP$ upper bounds are a simple consequence of small Presburger formulas for reachability relations of $\Z$-VASS~\cite{DBLP:conf/rp/HaaseH14} and $1$-VASS~\cite{LiCWX20} (see~\cref{app:np-upper-bound}).
\begin{restatable}{theorem}{npUpperBound}\label{np-upper-bound}
	Let $S \sset \ZZ^{p+1}$ be semilinear. 
	$\intreach(S)$, $\natreach(S)$, and $\vasreach(S)$ are in \NP.
\end{restatable}

\noindent\textbf{Residue classes.}
We will decompose semilinear sets \hbox{$S\sset\ZZ^{p+1}$} into residue classes modulo
some common multiple $B$ of all moduli in a formula for $\varphi$. With such a
$B$, each residue class modulo $B$ is a \emph{modulo-free set}, meaning it can
be defined by a Presburger formula that is \emph{modulo-free}.
This means, the formula contains no quantifiers and no modulo constraints. 
Let us make this precise. 
For any set $S\subseteq\ZZ^k$ and any modulus $B\ge 0$, we define the
residue classes of $S$ as
\begin{equation*}
	[S]_{B, \vec{b}} \coloneqq \set{ \vec{u} \in \ZZ^k : B\vec{u} + \vec{b} \in S }, 
\end{equation*}
where $\vec{b}$ ranges over $[0,B-1]^{p+1}$, the set of $(p+1)$-dimensional
vectors with entries in $[0,B-1]$.  
The following lemma follows from the argument above (see~\cref{app:modulo-free} for details).
\begin{restatable}{lemma}{moduloFree}\label{lem:modulo-free}
	Let $S \sset \ZZ^{p+1}$ be a semilinear set.
	One can compute $B \geq 1$ so that, for each $\vec{b} \in [0, B-1]^{p+1}$, $[S]_{B, \vec{b}}$ is effectively definable by a modulo-free Presburger formula.
\end{restatable}

\mysubsubsection{Computing interval endpoints.}
If $S\subseteq\ZZ^{p+1}$ is modulo-free (definable by a quantifier-free
modulo-free formula), then each set $S[\vec{t}]$ is a finite union of
intervals. It is straightforward to show that the functions computing (i)~the
number of intervals in $S[\vec{t}]$ (which is bounded by a constant only
depending on $S$), and (ii)~the endpoints of those intervals are also
Presburger computable. The exact statement is somewhat technical and can be
found in~\cref{app:intervals}.

\mysubsubsection{Computing Presburger-definable functions.} We will often need
to compute points in $S[\vec{t}]$ that have certain properties (e.g.\ have a
large gap above them). Here, the following (shown by Chan and
Ibarra~\cite[Theorem 2.1]{DBLP:journals/siamcomp/ChanI83}) will be useful:
It allows us to compute such points if  we can define them
uniquely.
\begin{lemma}\label{compute-presburger}
	Every Presburger-definable function $\NN^k\to\N$ can be computed in deterministic logspace.
\end{lemma}

\section{Integer Semantics}
\label{sec:integer-semantics}
\subsection{Removing modulo constraints}
Our first step is to show that it suffices to show \cref{main-result-integer} in the case where $S$ is given by a modulo-free formula. Roughly speaking, this is because neither positivity of local density, nor the complexity depend on whether we take the entire set, or only residue classes. First, the former:
\begin{restatable}{lemma}{denseResidueClasses}\label{dense-residue-classes}
	Let $S\subseteq\ZZ^{p+1},~B\ge 1$. Then $\density(S)>0$ if and
	only if for every $\vec{b}\in[0,B-1]^{p+1}$, we have
	$\density([S]_{B,\vec{b}})>0$.
\end{restatable}

Next we show that the complexity (i.e.\ whether $\intreach(\cdot)$ belongs to
$\AC^1$ or is $\NP$-hard) does not depend on whether we take the entire set or
residue classes, either.
\begin{restatable}{lemma}{complexityTransfer}\label{complexity-transfer}
	Let $S\subseteq\ZZ^{p+1}$ be semilinear and let $B\ge 1$. 
			If $\intreach([S]_{B,\vec{b}})$ is in $\AC^1$ for
			each $\vec{b}\in[0,B-1]^{p+1}$, then $\intreach(S)$
			belongs to $\AC^1$. %
			Moreover, if $\intreach([S]_{B,\vec{b}})$ is \hbox{$\NP$-hard} for some
			$\vec{b}\in[0,B-1]^{p+1}$, then $\intreach(S)$ is
			\hbox{$\NP$-hard}. %
	Analogous statements hold for $\natreach(\cdot)$.
\end{restatable}

\subsection{Making automata acyclic}
We now show that for $\intreach(S)$, we may assume that the input automaton is
acyclic. Let $\acintreach(S)$ be the problem $\intreach(S)$ where the input
must be an acyclic automaton.
\begin{restatable}{theorem}{makeAcyclic}\label{make-acyclic}
	For each semilinear $S\subseteq\ZZ^{p+k}$, the problem
	$\intreach(S)$ reduces in logspace to $\acintreach(S)$. 
\end{restatable}

\begin{proof}[Proof sketch]
	First, one can construct a polynomial-sized existential Presburger
	formula for the set of vectors over the transitions that correspond to
	a run reaching $S$ (because the Parikh image of an NFA admits a
	polynomial-sized existential Presburger
	formula~\cite[Thm.~1]{DBLP:conf/icalp/SeidlSMH04}). By solution size
	bounds for integer programming, this yields an exponential bound $\ell$
	on the length of a shortest run. This means, it suffices to construct
	an acyclic automaton $\cA'$ that simulates at least those runs of the
	input automaton $\cA$ that are of length $\le\ell$. By standard
	arguments about rearranging cycles, it suffices to do this for the set
	of runs that are concatenations of short $q$-cycles for some state $q$.
	Here, a \emph{short $q$-cycle} is a cycle of length at most $n$ (the
	number of states of $\cA$) that start and end in $q$. 

	A Carath\'{e}odory-type result about integer cones by Eisenbrand and
	Shmonin~\cite[Thm.~1]{EisenbrandS06} implies that the counter effect
	of such a concatenation can be produced by a concatenation that
	uses $\le r$ \emph{distinct} short $q$-cycles, for some polynomial
	$r$. Again by integer programming solution sizes (and our bound
	$\ell$), we may assume this modified concatenation to contain each of
	the $q$-cycles at most $m$ times, for some exponential $m$.

	To simulate such a concatenation, we build an acyclic gadget
	that, for each $j\in[0,\lceil \log m\rceil]$ in sequence, guesses
	a short $q$-cycle on-the-fly, and multiplies its effect by $2^j$.
	Picking different $q$-cycles for different $j$ is possible, but also
	yields runs consistent with $\cA$. Using the gadget $r$ times in a row
	simulates all concatenations as above. See \cref{app:make-acyclic}.
\end{proof}

\subsection{Unbounded isolation}
On the way to prove \cref{main-result-integer}, we will characterize the
modulo-free sets of local density zero using the notion of ``unbounded
isolation''. This is also used in 
\cref{main-result-integer-decidable}.

We say that $S\subseteq\ZZ^{p+1}$ \emph{has unbounded isolation} if there exists an $\ell\ge 1$ such that for every $\delta\ge 1$, there exists a $\vec{t}\in\ZZ^{p}$ and an $x\in\ZZ$ such that:
\begin{enumerate}[(a)]
	\item $S[\vec{t}]\cap [x,x+\ell-1]\ne\emptyset$,
	\item $S[\vec{t}]\cap [x-\delta,x-1]=\emptyset$, and
	\item $S[\vec{t}]\cap [x+\ell,x+\ell+\delta]=\emptyset$.
\end{enumerate}
In other words, we can find an interval $I$ of size $\ell$ (namely,
$[x,x+\ell-1]$) such that $I$ contains a point in $S[\vec{t}]$ and
on either side of $I$, there is an interval of size
$\delta$ that does not intersect $S[\vec{t}]$. Intuitively, $I$ is
isolated in the sense that it is neighbored by $\delta$-sized empty intervals.
A proof (and a sketch) of the following lemma can be found in~\cref{app:intreach-hardness}

\begin{restatable}{lemma}{intreachHardness}\label{intreach-hardness}
	Suppose $S\subseteq\ZZ^{p+1}$ has unbounded isolation. Then
	$\intreach(S)$ is $\NP$-hard.
\end{restatable}

The conditions of unbounded isolation are directly expressible in Presburger
arithmetic. Thus, we obtain:
\begin{lemma}\label{unbounded-isolation-decidable}
	Given a modulo-free Presburger formula for \hbox{$S\subseteq\ZZ^{p+1}$}, it is
	decidable whether $S$ has unbounded isolation.
\end{lemma}

\subsection{Algorithmic building block}
We now define a collection of subsets
$S^{(\rho,m)}\subseteq\ZZ^{2m+1}\times\ZZ$ that will serve as the basic
building blocks in our $\AC^1$ algorithm for $\intreach(S)$ if $S$ has positive local density. Specifically, we will show that for each set $S^{(\rho,m)}$ in our
collection, we have $\intreach(S^{(\rho,m)})$ in $\AC^1$, and that for every
$S$ with positive local density, the problem $\intreach(S)$ reduces to finitely many
instances of $\intreach(S^{(\rho,m)})$ for some $\rho,m$. 

\subsubsection*{Interval notation} It will be convenient to use notations for
distance and length of intervals.  For two disjoint intervals
$I,J\subseteq\ZZ$, let $d(I,J)$ be the minimum distance between a point in $I$
from a point in $J$. By $|I|$, we denote the length of $I$ (i.e.~the difference
between the largest and smallest point). 
For disjoint intervals $I$ and $J$, we write  $I<J$ 
if the numbers in $I$ are smaller than those in $J$.
We observe a \emph{triangle inequality for intervals}: If $I,J,K$ are pairwise
disjoint intervals, then $d(I,J)\le d(I,K)+|K|+d(K,J)$.
This is trivial if $K$ is infinite. If $K$ is finite and wlog $I<J$, then one can check this by considering all possible arrangements: $K,I,J$ and
$I,K,J$, and $I,J,K$. 
A central notion in our proof is the notion of a $\rho$-chain of intervals---a sufficient condition for the existence of a polynomial-time (even $\AC^1$) algorithm for reachability.
\begin{definition}[$\rho$-chain of intervals]\label{def:chain}
Let $\rho\ge 1$ be a constant. A sequence $I_1,\ldots,I_{m+1}$ of intervals is
called a \emph{$\rho$-chain of intervals} if $I_{m+1}$ is one-sided infinite and
	\begin{enumerate}[label=\textup{(A\arabic*)}, ref={(A\arabic*)}, leftmargin=0.5in]
		\item\label{itm:disjoint} $I_1, \ldots, I_m, I_{m+1}$ are pairwise disjoint, 
		\item\label{itm:sizes} $|I_1| \leq |I_2| \leq \cdots \leq |I_{m}|$,
		\item\label{itm:distances} for every $i \in [1,m], d(I_i, I_{i+1}) \leq \rho \cdot |I_i|$, and
		\item\label{itm:december} for every $i \in [2, m+1]$, either
		\begin{enumerate}[(i)]
			\item $I_i > I_j$ for all $j\in[1,i-1]$, or 
			\item $I_i < I_j$ for all $j\in[1,i-1]$.
		\end{enumerate} 
	\end{enumerate}
\end{definition}

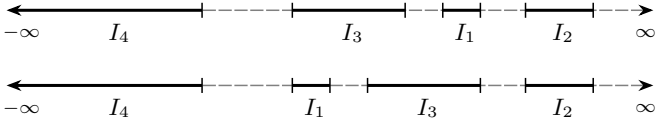
\begin{figure}
	\centering
	\begin{tikzpicture}
	\draw[gray, {Stealth[width=1.5mm, length=2mm, black]}-{Stealth[width=1.5mm, length=2mm, black]}, dashed, dash pattern={on 5pt off 1pt}, dash phase = 1.4, line width = 0.2mm] (-1.8, -1) -- (6.8, -1);
	\node at (-1.6, -1.3) {\scriptsize$-\infty$};
	\node at (6.7, -1.3) {\scriptsize$\infty$};

	\draw[line width = 0.2mm] (0.8, -0.9) -- (0.8, -1.1);
	\draw[line width = 0.4mm] (-1.7, -1) -- (0.8, -1);
	\node at (-0.3, -1.3) {\small$I_4$};

	\draw[line width = 0.2mm] (2, -0.9) -- (2, -1.1);
	\draw[line width = 0.2mm] (2.5, -0.9) -- (2.5, -1.1);
	\draw[line width = 0.4mm] (2, -1) -- (2.5, -1);
	\node at (2.3, -1.3) {\small$I_1$};

	\draw[line width = 0.2mm] (3, -0.9) -- (3, -1.1);
	\draw[line width = 0.2mm] (4.5, -0.9) -- (4.5, -1.1);
	\draw[line width = 0.4mm] (3, -1) -- (4.5, -1);
	\node at (3.8, -1.3) {\small$I_3$};

	\draw[line width = 0.2mm] (5.1, -0.9) -- (5.1, -1.1);
	\draw[line width = 0.2mm] (6, -0.9) -- (6, -1.1);
	\draw[line width = 0.4mm] (5.1, -1) -- (6, -1);
	\node at (5.6, -1.3) {\small$I_2$};

    \draw[gray, {Stealth[width=1.5mm, length=2mm, black]}-{Stealth[width=1.5mm, length=2mm, black]}, dashed, dash pattern={on 5pt off 1pt}, dash phase = 1.4, line width = 0.2mm] (-1.8, 0) -- (6.8, 0);
	\node at (-1.6, -0.3) {\scriptsize$-\infty$};
	\node at (6.7, -0.3) {\scriptsize$\infty$};

	\draw[line width = 0.2mm] (0.8, 0.1) -- (0.8, -0.1);
	\draw[line width = 0.4mm] (-1.7, 0) -- (0.8, 0);
	\node at (-0.3, -0.3) {\small$I_4$};

	\draw[line width = 0.2mm] (2, 0.1) -- (2, -0.1);
	\draw[line width = 0.2mm] (3.5, 0.1) -- (3.5, -0.1);
	\draw[line width = 0.4mm] (2, 0) -- (3.5, 0);
	\node at (2.8, -0.3) {\small$I_3$};

	\draw[line width = 0.2mm] (4, 0.1) -- (4, -0.1);
	\draw[line width = 0.2mm] (4.5, 0.1) -- (4.5, -0.1);
	\draw[line width = 0.4mm] (4, 0) -- (4.5, 0);
	\node at (4.3, -0.3) {\small$I_1$};

	\draw[line width = 0.2mm] (5.1, 0.1) -- (5.1, -0.1);
	\draw[line width = 0.2mm] (6, 0.1) -- (6, -0.1);
	\draw[line width = 0.4mm] (5.1, 0) -- (6, 0);
	\node at (5.6, -0.3) {\small$I_2$};
\end{tikzpicture}
	\vspace{-0.2in}
	\caption{
		An illustration for Condition~\ref{itm:december} of~\cref{def:chain} (when $m = 3$).
		The solid black horizonal line segments represent a sequence of pairwise disjoint intervals $I_1, I_2, I_3, I_4$  where  $|I_1| \leq \ldots \leq |I_4| = \infty$. The dashed grey horizonal line segments represent the complement of $I_1 \cup I_2 \cup I_3 \cup I_4$.
		The top picture shows a scenario when Condition~\ref{itm:december} is satisfied; notice that for $i = 3$, $I_3 < I_1$ and $I_3 < I_2$ (and for $i = 2$, $I_2 > I_1$).
		The bottom picture shows a scenario when Condition~\ref{itm:december} is falsified; observe for $i = 3$: $I_3 > I_1$ but $I_3 < I_2$.
	}
	\label{fig:a4}
\end{figure}

Roughly speaking, Condition~\ref{itm:december} requires that for every interval $I_i$, all smaller intervals $I_1, \ldots, I_{i-1}$ are either to the left or to the right; see~\cref{fig:a4} for a pair of examples.

The target sets of our algorithmic building blocks rely on the notion of $\rho$-chains of intervals.

\begin{definition}[Target sets and $\rho$-admissibility]\label{def:rho-admissible}
	Suppose $\rho>1$ is a rational number and $m\ge 1$. 
	We define the set $S^{(\rho,m)}\subseteq\ZZ^{2m+1}\times\ZZ$ (i.e.\ it has $2m+1$ parameters) as follows.
	We write the parameter vectors $\vec{t} \in \ZZ^{2m+1}$ as $\vec{t} = \langle s_1, t_1, \ldots, s_m, t_m, s_{m+1}\rangle$. 
	Let $I_i = [s_i, t_i]$ for $i \in [1,m]$, and $I_{m+1} = [s_{m+1}, \infty)$.
	If $I_1,\ldots,I_{m+1}$ is a $\rho$-chain of intervals, then $\vec{t}$ is said to be \emph{$\rho$-admissible}, and we set
	\begin{multline*}
		S^{(\rho,m)}[s_1,t_1,\ldots,s_m,t_m, s_{m+1}] \coloneqq I_1 \cup \cdots \cup I_m \cup I_{m+1}.
	\end{multline*}
	Otherwise, we set %
$
		S^{(\rho,m)}[s_1,t_1,\ldots,s_m,t_m, s_{m+1}] \coloneqq \emptyset.
$
\end{definition}

\subsection{Relating density, isolation, and building blocks}
\begin{definition}[Transformed building blocks]\label{def:building-block}
	We say that a set $S\subseteq\ZZ^{p+1}$ is a \emph{transformed building block} if there is a $\rho > 0$, an $m \ge 1$, and Presburger-definable functions \hbox{$\tau\colon \ZZ^p \to \ZZ^{2m+1}$} and \hbox{$\sigma \colon \ZZ^p \to \set{0, 1}$} such that for every $\vec{t} \in \ZZ^p$:
	\begin{equation*}
		S[\vec{t}] = (-1)^{\sigma(\vec{t})} \cdot S^{(\rho,m)}[\tau(\vec{t})]. 
	\end{equation*}
\end{definition}

In other words, $S$ is essentially a set of the form $S^{(\rho,m)}$, but the parameters for $S$ are not directly the interval ends---they have to be computed by a function $\tau$. Moreover, the function $\sigma$ allows us to flip the sign.

\begin{proposition}\label{integer-dichotomy-equivalence}
	Suppose $S\subseteq\ZZ^{p+1}$ is a modulo-free set. Then the following
	are equivalent:
	\begin{enumerate}[(1)]
		\item\label{equivalence-positive-local-density} $S$ has positive local density.
		\item\label{equivalence-no-unbounded-isolation} $S$ does \emph{not} have unbounded isolation.
		\item\label{equivalence-finite-union} $S$ is a finite union of transformed building blocks.
	\end{enumerate}
\end{proposition}
The implications
\impl{equivalence-positive-local-density}{equivalence-no-unbounded-isolation}
and \impl{equivalence-finite-union}{equivalence-positive-local-density} are
mostly straightforward.  For the latter, one can observe that transformed
building blocks have positive local density. Intuitively, this is because when
starting from a point in some interval $I$ in a $\rho$-chain of intervals, the
distance to the next interval $J$ in the chain is always at most $\rho$ times
the length of $I$. Therefore, we ``collect enough points'' during $I$ to get us
over the gap to $J$, while maintaining a probability of roughly $\ge
1/(\rho+1)$ of hitting $I$. Once we reach $J$, we then collect more points from
$J$ that will then get us over the next gap, etc. Furthermore, for
$S,S'\subseteq\ZZ^{p+1}$, we have $\density(S\cup S')\ge
\min(\density(S),\density(S'))$, so that a positive density of the constituents
yields a positive density of the entire union. See
Claims~\ref{clm:equivalence-1implies2} and~\ref{clm:equivalence-3implies1}
in~\cref{app:equivalence-finite-union-implications}.

Accordingly, the remainder of this section concentrates on the implication \impl{equivalence-no-unbounded-isolation}{equivalence-finite-union}.
We first introduce some basic definitions and then some concepts from well-quasi ordering (WQO) theory.

\subsubsection*{Interval-uniform sets} For subsets $S,S_1,\ldots,S_n\subseteq\ZZ^{p+1}$, we write
$S=S_1\oplus\cdots\oplus S_n$ if $S=S_1\cup\cdots\cup S_n$ and for each $\vec{t}\in\ZZ^p$, there is at most one $i\in[1,n]$ such that $S_i[\vec{t}]\ne\emptyset$. Observe that in this case, if some $S_i$ has unbounded isolation, then so does $S$. Moreover, if each $S_i$ is a finite union of transformed building blocks, then so is $S$. This means, if we decompose $S$ in this way, then proving the implication \impl{equivalence-no-unbounded-isolation}{equivalence-finite-union} for each $S_i$ will also establish it for $S$. This will allow us to assume $S$ to be ``interval-uniform'', which we define now.

If a set $A\subseteq\ZZ$ is a finite union of intervals (equivalently, $A$ is definable by a modulo-free Presburger formula), then $A$ has a \emph{canonical decomposition}, meaning $A$ can be written uniquely as $A=I_1\cup\cdots\cup I_m$ such that 
\begin{enumerate}[(a)]
	\item $I_1, \ldots, I_m$ are pairwise disjoint, 
	\item $I_1 < \cdots < I_m$, and 
	\item for all $i \in [2,m]$, $d(I_{i-1}, I_{i}) \ge 2$.
\end{enumerate}
In this case, the \emph{type} of $A$ is the tuple $(c_1, \ldots, c_m) \in
\set{1, -\infty, \infty, 2\infty}^m$, where
\begin{equation*}
	c_i =
	\begin{cases}
		\infty & 	\text{if } I_i = [s,\infty) \text{ for some } s \in \ZZ, \\
		-\infty & 	\text{if } I_i = (-\infty,s] \text{ for some } s \in \ZZ, \\
		2\infty & 	\text{if } I_i = \ZZ, \text{ and }\\
		1 &			\text{if } I_i \text{ is finite.}
	\end{cases}
\end{equation*}
Thus, the type encodes (i)~how many intervals 
there are, and (ii)~which of them are finite, one-sided infinite (and on which
side), or two-sided infinite.  We say that $S\subseteq\ZZ^{p+1}$ is
\emph{interval-uniform} if it is definable by a modulo-free Presburger formula,
and also all sets $S[\vec{t}]$ are of the same type, except for those with
$S[\vec{t}]=\emptyset$. The following is a consequence of the fact that the set
of $\vec{t}$ for which $S[\vec{t}]$ has a particular type is
Presburger-definable (see
\cref{app:interval-uniform-decomposition}).
\begin{restatable}[Inverval-uniform decomposition]{proposition}{intervalUniformDecomposition} \label{pro:interval-uniform-decomposition}
	Let \hbox{$S \sset \ZZ^{p+1}$} be a set given by a modulo-free
	Presburger formula.  Then $S$ can be written as
	$S=S_1\oplus\cdots\oplus S_n$, where $n\in\NN$ and $S_1,
	\ldots, S_n \sset \ZZ^{p+1}$ are interval-uniform.  
\end{restatable}

\subsubsection*{Ratio matrices and ideals} By our observation above---that it suffices to prove the implication
\impl{equivalence-no-unbounded-isolation}{equivalence-finite-union} for
constituents of a $\oplus$-decomposition, we now know that it suffices to prove
\impl{equivalence-no-unbounded-isolation}{equivalence-finite-union} for
interval-uniform sets $S$. 

This allows us to define a key ingredient in our proof (``ratio ideals''), a combinatorial object that encodes which gaps can grow unboundedly compared to interval sizes.
This requires some definitions. 
For a subset $A\subseteq\ZZ$ that is a finite union of intervals, let $A = I_1 \cup \cdots \cup I_m$ be its canonical decomposition. 
Its complement $\ZZ \setminus A$ is also a finite union of intervals, and
thus has a canonical decomposition $\ZZ \setminus A = K_1 \cup \cdots \cup K_r$.  
Here, $r$ can be $m-1$, $m$, or $m+1$; the exact value of $r$ depends on the number of unbounded invervals in $A$.
In fact, the type of $\ZZ\setminus A$, only depends on the type of $A$. 
The \emph{ratio matrix of $A$} is an $r\times m$ matrix $M_A \in \Mat(r \times m, \NN_\omega)$ with entries in $\NN_\omega\coloneqq \NN\cup\{\omega\}$, where
\begin{equation*} 
	(M_A)_{i,j}\coloneqq
	\begin{cases} 
		\left\lceil\frac{|K_i|}{|I_j|}\right\rceil 
			& \text{if $|K_i| \in \NN$,} \\ 
		\omega 
			& \text{if $|K_i|=\infty$.}
	\end{cases}
\end{equation*}
Note that since all $I_j$ are non-empty, the denominator is always non-zero. 
Put in other words, the entry $(M_A)_{i,j}$ is the quotient of the size of the $i$-th gap and the size of the $j$-th interval.

If $S \sset \ZZ^{p+1}$ is interval-uniform, then for every $\vec{t}\in\ZZ^p$, the set $\ZZ \setminus S[\vec{t}]$ has the same number $r \in \NN$ of intervals.
Hence, we obtain a set $\MM \sset \Mat(r \times m, \NN_{\omega})$ by defining 
\begin{equation*}
	\MM \coloneqq \set{M_{S[\vec{t}]} \mid \vec{t} \in \ZZ^p \text{ and } S[\vec{t}] \neq \emptyset}.
\end{equation*}

Consider a subset $U\subseteq \Mat(r \times m, \NN_{\omega})$. 
The \emph{downwards-closure} of $U$, denoted $\downclose{U}$ is the set of all matrices in \hbox{$M\in \Mat(r \times m, \NN)$} such that there is a $M'\in U$ with $M_{i,j} \leq M'_{i,j}$ for all $i,j$. If $U=\{D\}$ is a singleton, we also write $\downclose{D}$ for $\downclose{U}$.
We will now take the ideal decomposition of $\MM$.
For this, we interpret $\MM$ as a set of vectors in $(\NN_{\omega})^{r \cdot m}$ by interpreting each matrix $M_{S[\vec{t}]} \in \MM$ as a vector in $(\NN_{\omega})^{r \cdot m}$.
It does not matter exactly how this is achieved, for example, one can append the rows of each matrix together.
The result is that there exists a finite collection of matrices $D_1, \ldots, D_z \in \Mat(r \times m, \NN_{\omega})$ that are incomparable w.r.t.\ $\le$ and such that
\begin{equation*}	
	\downclose{\MM} =  \downclose{D_1} \cup \cdots \cup \downclose{D_z}.
\end{equation*}
The sets $\downclose{D_i}$ are ideals in the set $(\Mat(r\times m,\NN_\omega),\le)$, viewed as a WQO. Therefore, we call the sets $\downclose{D_i}$ the \emph{ratio ideals} of $S$. In particular,
the existence and uniqueness (up to ordering) of $D_1,\ldots,D_z$ are well-known facts in WQO theory~\cite{thesisHalfon}.

\subsubsection*{$\omega$-constellations} The notion of ratio ideals now allows
us to connect the lack of unbounded isolation with transformed building blocks:
We will show that if an interval-uniform $S$ has no unbounded isolation, then
this rules out a particular pattern called ``$\omega$-constellation'' in the
ratio ideals.

Let $S\sset\ZZ^{p+1}$ be interval-uniform.  We say that $S$ has
an \emph{$\omega$-constellation} if for some of the ratio ideals $D$ for $S$, there are 
$i_1, i_2 \in [1, r]$ such that $i_1 < i_2$ and, for all $j \in [1, m]$ such
that $K_{i_1} < I_j < K_{i_2}$, the entries $D_{i_1, j} = \omega$ and $D_{i_2,
j} = \omega$.  In other words, there exist two gaps $K_{i_1}$ and $K_{i_2}$
whose length has unbounded ratio compared to the lengths of all intervals $I_j$
in between these two gaps. With this notion in hand, we can phrase the following
two propositions, which together imply
\impl{equivalence-no-unbounded-isolation}{equivalence-finite-union} of
\cref{integer-dichotomy-equivalence} (due to \cref{pro:interval-uniform-decomposition}).

\begin{restatable}{proposition}{noIsolationNoConstellation}\label{no-isolation-no-constellation}
	If an interval-uniform $S\sset\ZZ^{p+1}$ has an \hbox{$\omega$-constellation},
	then $S$ has unbounded isolation.
\end{restatable}

\begin{proposition}\label{no-constellation-building-blocks}
	If an interval-uniform $S\subseteq\ZZ^{p+1}$ has no
	\hbox{$\omega$-constellation}, then $S$ is a finite union of transformed building
	blocks.
\end{proposition}

For \cref{no-isolation-no-constellation}, one observes that the two gaps
$K_{i_1}$ and $K_{i_2}$ whose sizes have unbounded ratio compared to all
intervals between them, even yield something stronger: A bound $\ell$ so that
for each size $\delta$, we can find $\vec{t}$ so that all intervals and all
gaps between $K_{i_1}$ and $K_{i_2}$ together span at most $\ell$ points,
whereas $K_{i_1}$ and $K_{i_2}$ have size at least $\delta$. To this end, we
consider the set $A\subseteq\N^2$ of all pairs
\[ \left\langle \min\{|K_{i_1}[\vec{t}]|,|K_{i_2}[\vec{t}]|,~\sum_{j\in J}|I_j[\vec{t}]|+\sum_{i=i_1+1}^{i_2-1} |K_i[\vec{t}]|\right\rangle \]
where $\vec{t}$ ranges over all $\ZZ^p$ for which $S[\vec{t}] \neq\emptyset$. 
Then $A$ is Presburger-definable. 
Moreover, by picking an $\omega$-constellation in
$i_1$ and $i_2$ such that $i_2-i_1$ is minimized, we can make sure that even
every \emph{gap} $K_i$ (i.e.\ not only the invervals $I_k$) between $K_{i_1}$
and $K_{i_2}$ has bounded ratio w.r.t.\ to some interval $I_j$ with
$K_{i_1}<I_j<K_{i_2}$. With this, the
$\omega$-constellation yields for every $k\in\N$ a vector
$\vec{u}\in A$ with $\vec{u}[1]\ge k\cdot\vec{u}[2]$. Then, considering
a semilinear representation for $A$ reveals that there must be a constant
$\ell\ge 0$ so that for every $\delta\ge 1$, we can find a $\vec{u}\in A$ with
$\vec{u}[2]\le\ell$ and $\vec{u}[1]\ge\delta$. This yields the unbounded
isolation of $S$. See \cref{app:no-isolation-no-constellation} for a detailed
proof.

The remainder of this subsection is devoted to proving
\cref{no-constellation-building-blocks}. Thus, we now fix $S\subseteq\ZZ^{p+1}$
as given by a modulo-free Presburger formula, and where $S$ is interval-uniform. 
By interval-uniformity, there is a particular type $\tau=(c_1,\ldots,c_m)$ such
that all $S[\vec{t}]$ have type $\tau$ (except the empty ones). Moreover, let
$S[\vec{t}]=I_1[\vec{t}]\cup\cdots\cup I_m[\vec{t}]$ be the canonical
decomposition of $S[\vec{t}]$. By interval-uniformity, it does not depend on
$\vec{t}$ whether $I_i[\vec{t}]$ is finite (or one-sided infinite, or two-sided
infinite).

\Cref{no-constellation-building-blocks} requires us to show that for each $\vec{t}\in\ZZ^p$, we can decompose $S[\vec{t}]$ into finitely many $\rho$-chains, for some suitable number $\rho$. Let us begin by defining $\rho$. Let $R$ be the largest natural number occurring in the ideal decomposition of $\MM$. 
With this, let  $\rho \coloneqq 2R(m+1)$.

Our first lemma encapsulates the main idea of \cref{no-constellation-building-blocks}.
It tells us that for every $I_i[\vec{t}]$, we can find $I_j[\vec{t}]$ that is the next step within the $\rho$-chain. 
Such a next step will be called a ``harbor'', the idea being that from $I_i[\vec{t}]$, we can ``jump to the harbor'' without having to cross too large of a gap. 
More precisely, we say that $I_j[\vec{t}]$ is a \emph{harbor} for $I_i[\vec{t}]$ if 
\begin{enumerate}[(a)]
	\item \label{itm:harbor-size} $|I_j[\vec{t}]| > |I_i[\vec{t}]|$,
	\item \label{itm:harbor-close} $d(I_i[\vec{t}], I_j[\vec{t}])\le \rho\cdot |I_i[\vec{t}]|$, and 
	\item \label{itm:harbor-nearest} for all $k \in (i,j)$, we have $|I_k[\vec{t}]|\le |I_i[\vec{t}]|$.
\end{enumerate}

\begin{restatable}{lemma}{harborExistence}\label{harbor-existence}
	For every $\vec{t}\in\ZZ^p$ and every $i\in[1,m]$, if the interval
	$I_i[\vec{t}]$ is finite, then it has a harbor.
\end{restatable}
\begin{proof}[Proof sketch] 
	Assume that there are
	intervals $I_{j_1}[\vec{t}]$ and $I_{j_2}[\vec{t}]$ such that
	$I_{j_1}[\vec{t}]<I_i<I_{j_2}[\vec{t}]$ and also both $I_{j_1}[\vec{t}]$ and
	$I_{j_2}[\vec{t}]$ are larger than $I_i[\vec{t}]$ (the other cases are
	similar). Here, we pick $j_1$ and $j_2$ as close to $i$ as possible.  Then
	all the intervals $I_j$ between $I_i[\vec{t}]$ and $I_{j_1}[\vec{t}]$, or
	between $I_i[\vec{t}]$ and $I_{j_2}[\vec{t}]$, must have size at most
	$|I_i[\vec{t}]|$.  We claim that either $I_{j_1}[\vec{t}]$ or $I_{j_2}[\vec{t}]$ must be
	a harbor: If the distance from $I_i[\vec{t}]$ to $I_{j_1}[\vec{t}]$ 
	is more than $\rho=2R(m+1)$ times the length of $I_i[\vec{t}]$, then this must
	be due to a ``big gap'' $K_j[\vec{t}]$, meaning one with ratio $>R$ to
	$I_i[\vec{t}]$, and in particular to every other $I_j[\vec{t}]$ between
	$I_i[\vec{t}]$ and $I_{j_1}[\vec{t}]$. Now if there is a \emph{big gap in both
	directions} (i.e.\ towards $I_{j_1}[\vec{t}]$ and towards $I_{j_2}[\vec{t}]$),
	then these two big gaps would induce an $\omega$-constellation in some ratio ideal.
	See \cref{app:harbor-existence}.
\end{proof}

More generally, a \emph{harbor chain for $\vec{t}\in\ZZ^p$} is a sequence $\tuple{i_1,\ldots,i_\ell}$ of indices $i_1, \ldots, i_\ell \in [1,m]$ such that 
\begin{enumerate}[(a)]
	\item for each $j\in[1,\ell-1]$, the interval $I_{i_j}[\vec{t}]$ is finite and $I_{i_{j+1}}[\vec{t}]$ is a harbor for $I_{i_j}[\vec{t}]$; and
	\item $I_{i_\ell}[\vec{t}]$ is one-sided infinite.
\end{enumerate}

Since
$I_{i_{j+1}}[\vec{t}]$ must be larger than $I_{i_j}[\vec{t}]$, we know that
such a harbor chain must be repetition free, and thus have length at most $m$.
It is not hard to see that if $(i_1,\ldots,i_\ell)$ is a harbor chain, then
intervals $I_{i_1}[\vec{t}], \ldots, I_{i_\ell}[\vec{t}]$ form a $\rho$-chain
of intervals.
By \cref{harbor-existence}, we know that for every $\vec{t}\in\ZZ^p$, and every
finite interval $I_i[\vec{t}]$, there is a harbor chain starting in $i$ (and, trivially, the same is true for one-sided infinite intervals).
This allows us to decompose all of
$S[\vec{t}]$ into $\rho$-chains. 

For each sequence $\chi=\langle i_1,\ldots,i_\ell\rangle$ of numbers in $[1,m]$
with $\ell\le m$, we define the set $T_\chi\subseteq\ZZ^{p+1}$, where
$T_\chi[\vec{t}] = I_{i_1}[\vec{t}]\cup\cdots\cup I_{i_\ell}[\vec{t}]$ if
$\chi$ is a harbor sequence for $\vec{t}$, and $T_\chi[\vec{t}]=\emptyset$
otherwise.  Moreover, we have $T_\infty[\vec{t}] =[0,\infty)$ if
$S[\vec{t}]=\ZZ$ and $T_\infty[\vec{t}]=\emptyset$ otherwise.  Finally, let
$T_{-\infty}[\vec{t}]=(-\infty,0]$ if $S[\vec{t}]=\ZZ$ and
$T_{-\infty}[\vec{t}]=\emptyset$ otherwise. Then it follows from
\cref{harbor-existence} that $S=T_{\infty}\cup T_{-\infty}\cup
\bigcup_{\chi}T_{\chi}$\label{claim-decomposition-transformed-building-blocks},
where $\chi$ ranges over the sequences over $[1,m]$ of length $\le m$.
Moreover, each $T_{\infty},T_{-\infty},T_{\chi}$ is indeed a transformed
building block. See \cref{app:decomposition-transformed-building-blocks}.  This
completes the proof of \cref{no-constellation-building-blocks}.

\subsection{Complexity of building blocks}
We will now complete the first statement of~\cref{main-result-integer}.
Since we have seen in~\cref{integer-dichotomy-equivalence} that $\density(S)>0$
implies that $S$ is a finite union of transformed building blocks, and for each
transformed buildling blocks $T$, the problem $\intreach(T)$ reduces to
$\intreach(S^{(\rho,m)})$ for some constants $\rho,m\ge 1$ in logspace (by
\cref{compute-presburger}), it remains to show:
\begin{restatable}{proposition}{buildingBlockACOne}\label{building-block-ac1}
	For each rational $\rho>1$ and $m\ge 1$, the problem $\intreach(S^{(\rho,m)})$ belongs to $\AC^1$.
\end{restatable}
In the following, it is important to remember that $\rho$ and $m$ are
\emph{constants}, whereas the entries of $\vec{t}$ are part of the input.

\subsubsection*{Growing instances}
As a simple preparation, we show that we may assume that, in some sense, the
interval lengths in the input set are growing quickly.  We say that $\vec{t}$
is \emph{growing} if for every $i\in[1,m+1]$ and every $j\in[1,i]$, we have
$|I_i|\ge 2\cdot d(I_j,I_i)$. 
It is not difficult to see that
for \cref{building-block-ac1}, we only need to consider growing instances (see
\cref{appendix-make-growing}).

\subsubsection*{Reachability problems via semirings}
By \cref{make-acyclic}, it suffices to show that $\acintreach(S^{(\rho,m)})$
belongs to $\AC^1$.  The known $\AC^1$ algorithm for coverability in
$1$-$\ZZ$-VASS (see \cref{reference-coverability-one-z-vass}) uses the classical idea that the
weight of a word in a weighted automaton can be computed by multiplying
matrices. In the case of coverability, one actually only needs to compute the
power $M^n$ of a matrix $M$, which can be done by repeated squaring: If
$n=2^k$, then one can compute $M^n$ by computing $M^1=M$, and then
$M^{2^{i+1}}=(M^{2^i})^2$ for $i=2,\ldots,k$. In order to decide coverability
this way, one views $\ZZ$-weighted automata as weighted automata over the
tropical semiring $\tuple{\ZZ \cup \set{-\infty}, \max, +, -\infty, 0}$. For this semiring,
since each squaring can be performed in $\AC^0$, doing this logarithmically
many times leads to an $\AC^1$ procedure.

Our algorithm for $\AcReach(S^{(\rho,m)})$ is structured similarly. However, in
order to capture the more general sets $S^{(\rho,m)}$, we need to construct
semirings so that (i)~matrix products can still be computed in $\AC^0$ and
(ii)~the semiring elements reflect whether we can reach an element in
$S^{(\rho,m)}$. We construct such semirings as quotients of $\B[X,X^{-1}]$, the semiring of
Laurent polynomials with Boolean coefficients.  Intuitively, while the
tropical semiring captures only ``coverability information'' (i.e.\ the
maximal number we can reach on a set of paths), the semiring $\B[X,X^{-1}]$
captures the entire reachability set.

\newcommand{\sr}[2]{\calS_{#1,#2}}

\subsubsection*{Laurent polynomials with Boolean coefficients}
Let \hbox{$\B = \set{0,1}$} be the Boolean semiring, where we multiply by conjunction and add by disjunction. For a semiring $\calS$, let $\calS[X,X^{-1}]$ be the semiring of \emph{Laurent polynomials with coefficients in $\calS$}, which consists of formal sums $\sum_{i\in\Z} s_iX^i$,
where $s_i\in \calS$ for $i\in\Z$, and all but finitely many $s_i$ are zero. Addition and multiplication are defined as in polynomial rings:
\[ \left(\sum_{i\in\Z} s_iX^i\right)\left(\sum_{i\in\Z} s'_iX^i\right)=\sum_{i\in\Z} \left(\sum_{j,k\in\Z,\,j+k=i} s_js'_k\right)X^i,\]
\[ \sum_{i\in\Z} s_iX^i+\sum_{i\in\Z} s'_iX^i=\sum_{i\in\Z} (s_i+s'_i)X^i.\]
In particular, $\B[X,X^{-1}]$ is the \emph{semiring of Laurent polynomials with
Boolean coefficients}. Then $\B[X,X^{-1}]$ is isomorphic to the semiring
of finite subsets of $\ZZ$, where we multiply by Minkowski sum, and add by
union. We adopt the perspective of Laurent polynomials (rather than finite sets), since here, addition and multiplication are naturally defined the way we require them.
Nevertheless, we will occasionally borrow terminology from finite sets: We will identify a term $X^i$ with the number $i$ and say, e.g.\ that $f\in\B[X,X^{-1}]$ ``contains a point'' or ``intersects an interval''.

\subsubsection*{From automata to matrices}
Let $\cA$ be a $\Z$-weighted automaton with $n$ states $Q = \set{1, \ldots, n}$
and let $1$ be the initial state, and $n$ the final state. 
We may assume that between any $p$ and $q$, there is at most one transition. 
Let $A \in \Mat(n \times n, \B[X, X^{-1}])$ be the matrix defined as follows. For $q,r\in[1,n]$, if there is an transition from $q$ to $r$ with weight $w$, then $A_{q,r}=X^{w}$. 
If there is no transition from $q$ to $r$, then let $A_{q,r} \coloneqq 0$.
Then it is a standard fact (easily shown by induction) that for any $q,r$:

\begin{equation*}
	(A^k)_{q,r} = \sum_{\text{paths $\pi$ from $q$ to $r$ of length $k$}} X^{\text{weight of $\pi$}}.
\end{equation*}
If $B=A+I_n$, where $I_n\in\Mat(n\times n,\BB[X,X^{-1}])$ is the $n\times n$ identity matrix, then $B^k$ is the sum of all $A^\ell$ where $\ell\in[0,k]$.
In particular,
since all paths are of length $\le n$:
\begin{equation} \label{matrix-of-polynomials}
	(B^n)_{q,r} = \sum_{\text{paths $\pi$ from $q$ to $r$}} X^{\text{weight of $\pi$}}. 
\end{equation}
Thus, if we had $B^n$ available, we could easily decide if there is a path with
weight in $I_1\cup\cdots\cup I_{m+1}$. However, there might be exponentially
many different weights on paths from $q$ to $r$. Therefore, instead of
computing in the semiring $\B[X,X^{-1}]$, we work in the quotient
$\sr{\rho}{\vec{t}}=\B[X,X^{-1}]/\mathord{\equiv_{\rho,\vec{t}}}$, where $\equiv_{\rho,\vec{t}}$ is a congruence defined by equations.

A \emph{congruence} on a
semiring $\calS$ is an equivalence relation $\equiv$ such that $s_1\equiv s_1'$
and $s_2\equiv s'_2$ implies $s_1+s_2\equiv s'_1+s'_2$ and also $s_1s_2\equiv
s'_1s'_2$. The equivalence class containing \hbox{$s\in\calS$} is denoted by $[s]$. 
A congruence defines a new semiring $\calS/\mathord{\equiv}$, which consists of the equivalence classes modulo $\equiv$. Addition and multiplication are defined via representatives, i.e.\ $[s_1][s_2]\coloneqq[s_1s_2]$ and $[s_1]+[s_2]\coloneqq[s_1+s_2]$.
Since $\equiv$ is a congruence, these operations are well-defined
(i.e.\ do not depend on the chosen representatives $s_1,s_2$.

\subsubsection*{Equations}
Our congurence $\equiv_{\rho,\vec{t}}$ is defined by equations, which involve further quantities. 
For $i \in [1, m+1]$, define
\begin{align*}
	u_i &= \max_{j\in[1,i]} d(I_j, I_i) & v_i&=|I_i|.
\end{align*}
In particular, $v_m=|I_{m+1}|=\infty$ since $I_{m+1}$ is infinite.  Note that
the $u_i$ and $v_i$ depend on the input $\vec{t}\in\ZZ^{2m+1}$, but can be
computed in $\AC^0$ ($\rho$ and $m$, on the other hand, are constants!).

Let us define the congruence $\equiv_{\rho,\vec{t}}$ on $\B[X,X^{-1}]$.
We use the notation $X^{[i,j]}\coloneqq\sum_{k\in[i,j]} X^k$ and
define $\equiv_{\rho,\vec{t}}$ as the smallest congruence relation such that 
\begin{align}
	X^i+X^j+X^k~&\equiv_{\rho,\vec{t}}~X^{[i,j]}+X^k~\text{if $j-i\in[0,v_\ell]$,}  \label{eq-left}\\
	&\text{and $k-j\in[u_\ell,v_\ell]$ and ($I_1<I_\ell$ or $\ell=1$)}\notag \\
	X^h+X^i+X^j~&\equiv_{\rho,\vec{t}}~X^h+X^{[i,j]}~\text{if $j-i\in[0,v_\ell]$,} \label{eq-right}\\
	&\text{and $i-h\in[u_\ell,v_\ell]$ and $I_\ell<I_1$}    \notag
\end{align}
for all $\ell\in[1,m+1]$.  Let
$\sr{\rho}{\vec{t}}\coloneqq\B[X,X^{-1}]/\mathord{\equiv_{\rho,\vec{t}}}$.  We
shall prove two key properties of the equations: (i)~applying them will not
affect whether a polynomial contains a term in $S^{(\rho,m)}[\vec{t}]$ and
(ii)~every polynomial in $\B[X,X^{-1}]$ has an
$\equiv_{\rho,\vec{t}}$-equivalent that is a sum of at most a constant number
of terms $X^{[r,s]}$. The latter will be exploited algorithmically because we
can represent $X^{[r,s]}$ by the pair $\tuple{r,s}$.  We first
show (i):
\begin{restatable}{lemma}{correctnessEquations}\label{correctness-equations}
	Let $f\equiv_{\rho,\vec{t}} f'$. Then $f'$ contains a term in
	$S^{(\rho,m)}[\vec{t}]$ if and only if $f$ does.
\end{restatable}
\begin{proof}[Proof sketch]
We have to argue that applying an equation does not affect whether
$S^{(\rho,m)}[\vec{t}]$ is reached.  Let us first consider \eqref{eq-left} with
$\ell=1$. It is equivalent to $X^i+X^j\equiv_{\rho,\vec{t}} X^{[i,j]}$, where
$j-i\in[0,v_1]$, since $k-j$ is allowed to be zero, as $u_1=0$.  The equation
says: If we have terms $X^i$ and $X^j$ that are at most $v_1=|I_1|$ apart, then
we can introduce all terms $X^s$ with $s\in[i,j]$. This does not affect
intersection with $S[\vec{t}]$. Indeed, since all our intervals have length $\ge
|I_1|=v_1$: If $j-i\le |I_1|$, then $[i,j]$ intersects $I_1\cup\cdots\cup
I_{m+1}$ iff $\{i,j\}$ does.

Now consider \eqref{eq-left} with $\ell>1$. The equation says, if we have
terms $X^i$ and $X^j$ that are at most $v_\ell=|I_\ell|$ apart, and we can see
an $X^k$ with distance in $[u_\ell,v_\ell]$ from $X^j$, then we may
introduce all $X^s$ with $s\in[i,j]$. This time, it may happen that
$[i,j]$ intersects some $I_r$ for $r<\ell$ although $\{i,j\}$ did not: e.g.\ if
$r<\ell$ and $|I_r|<|I_\ell|$. However, since all intervals $I_r$ with $r<\ell$
are to the left of $I_\ell$, we know that $j$ must lie between $I_r$ and
$I_\ell$. Crucially, since $k\in[j+u_\ell,j+v_\ell]$ with $u_\ell\ge d(I_r,I_\ell)$,
the point $k$ must belong to $I_\ell$. To summarize, if this equation happens
to introduce a point in some $I_r$ with $r<\ell$, then there already was one in
$I_\ell$.  See \cref{appendix-correctness-equations}).
\end{proof}

The second key property of the equations is that every polynomial has an
equivalent with a small representation:
\begin{restatable}{lemma}{smallRepresentation}\label{small-representation}
	Every polynomial in
	$\B[X,X^{-1}]$ is $\equiv_{\rho,\vec{t}}$-equivalent to a sum of at
	most $(m\rho+2m)^{m}$ intervals.
\end{restatable}
\begin{proof}[Proof sketch] 
Two terms $X^i$ and $X^j$ are called \emph{$\ell$-close}, in symbols
$X^i\leftrightarrow_\ell X^j$, if $|i-j|\le v_\ell$. A polynomial
$f\in\B[X,X^{-1}]$ is an \emph{$\ell$-cluster} if its set of terms is strongly
connected w.r.t.\ the relation $\leftrightarrow_\ell$. We claim that every
$\ell$-cluster can be written as a sum of at most $m(\rho+2)$ many
$(\ell-1)$-clusters. This is sufficient, since (i)~by \eqref{eq-left}, any
$1$-cluster is equivalent to a single interval and (ii)~$v_{m+1}=\infty$
implies that every polynomial is an $(m+1)$-cluster. 

For the claim, note that an $\ell$-cluster for $\ell>1$ can be written as a sum
of disjoint $(\ell-1)$-clusters, which consequently have distance $>v_{\ell-1}$
from each other, but also at most $v_\ell$.  The fact that $I_1,\ldots,I_{m+1}$
is a $\rho$-chain implies that $u_{\ell}/v_{\ell-1}\le \rho(m+1)$. This means,
if our $\ell$-cluster contains more than $m(\rho+1)$ many $(\ell-1)$-clusters,
then there must be two points of distance at least $u_\ell$. One can then argue
(here we use that our instance is growing) that there must even be such a pair
with distance in $[u_\ell,v_{\ell}]$. But such a pair within the first
$m(\rho+1)$ many $(\ell-1)$-clusters allows us to apply \eqref{eq-right} and
use it to collapse all remaining $(\ell-1)$-clusters into a single interval,
resulting in $\le m(\rho+2)$ many $(\ell-1)$-clusters. 
See \cref{app:small-representation}.
\end{proof}

\begin{proof}[Proof sketch of \cref{building-block-ac1}]
Consider the matrix $C\in\Mat(n\times n,\sr{\rho}{\vec{t}})$ with
$C_{p,q}=[B_{p,q}]$, i.e.\ $C$ contains the equivalence classes of the
polynomials in $B$. We compute $C^n$, which, by \eqref{matrix-of-polynomials}
and \cref{correctness-equations}, tells us whether any of the paths in $\cA$
reaches $S^{(\rho,m)}[\vec{t}]$. In our algorithm, we represent a polynomial as
a \emph{$\langle\rho,\vec{t}\rangle$-representation}, which is a list $\langle i_1,j_1\rangle,\ldots,\langle i_r,j_r\rangle$ of pairs of
binary-encoded numbers with $r\le (m\rho+2m)^m$, which is a constant. It represents the polynomial
$X^{[i_1,j_1]}+\cdots+X^{[i_r,j_r]}$. By \cref{small-representation}, every
polynomial in $\B[X,X^{-1}]$ has a $\langle\rho,\vec{t}\rangle$-representation. We compute $C^n$ by repeated squaring, i.e.\
we compute $C^1=C$ and then $C^{2^{k+1}}=(C^{2^k})^2$, and we claim that each
step can be performed in $\AC^0$. Since we only need to square $\lceil\log
n\rceil$ many times, we obtain an $\AC^1$ circuit overall.

To compute a matrix square, we first compute in $\AC^0$ the square of the
matrix of representing polynomials.  This yields polynomials
$X^{[i_1,j_1]}+\cdots+X^{[i_r,j_r]}$, where it is possible that
$r>(m\rho+2m)^m$. Therefore, along the lines of \cref{small-representation}, we
then compute in $\AC^0$ an equivalent
$\langle\rho,\vec{t}\rangle$-representation. Here, decomposing the polynomial into
$1$-clusters may seem to require $\NL$, since connectedness of graphs is
$\NL$-hard. However, two terms $X^i,X^j$ belonging to the same $1$-cluster is
equivalent to all gaps between $X^i$ and $X^j$ having size $\le v_\ell$, which
can be checked in $\AC^0$.  See \cref{app:building-block-ac1}.
\end{proof}

\subsection{Decidability}
\cref{main-result-integer-decidable} follows from the material in this section:
Given $S\subseteq\ZZ^{p+1}$, \cref{lem:modulo-free} lets us
compute modulo-free formulas for $S$'s residue classes. By
\cref{dense-residue-classes}, it suffices to decide whether all residue classes
have positive local density, or equivalently, whether none of them has
unbounded isolation (\cref{integer-dichotomy-equivalence}). This is
decidable by \cref{unbounded-isolation-decidable}.

\section{Natural Semantics}
\label{sec:natural-semantics}
The results about natural semantics (\cref{main-result-natural,main-result-natural-decidable}) can largely be deduced from the material in~\cref{sec:integer-semantics}. 
To this end, for a given semilinear $S\subseteq\Z^p\times\N$, we consider $S' \coloneqq S\cup (\ZZ\times\ZZ_{<0})$. 
Then, we show that $\densityplus(S) > 0$ if and only if $\density(S')>0$. 
The results about sets with positive local density then apply to $S'$. 
For example, we conclude that $\intreach(S')$ is in $\AC^1$, which implies that $\natreach(S)$ is in $\AC^1$. 
For space reasons, the remainder of the proofs are in \cref{app:natural-semantics} (with a short overview in~\cref{app:natural-semantics-overview}).

\section{VASS Semantics}
\label{sec:vass-semantics}
\subsection{$\AC^1$ upper bound in dichotomy}
In the next subsection, we will show that coverability in $1$-VASS is in $\AC^1$. 
Let us first show that this implies that $\vasreach(S)$ is in $\AC^1$ membership if $S$ is uniforly quasi-upwards closed (part~\ref{itm:vas-p-new} of~\cref{main-result-vass}).
\begin{restatable}{lemma}{vassReduction}\label{vass-reduction}
	If $S\subseteq\ZZ^{p}\times\N$ is uniformly quasi-upward
	closed, then $\vasreach(S)$ logspace-reduces to $1$-VASS coverability.
\end{restatable}
\begin{proof}[Proof sketch]
	If $S$  is $\tuple{\delta,M}$-upward closed, then for every
	$\vec{t}\in\ZZ^p$, the set $S[\vec{t}]$ is $\delta$-upward closed, if
	we add some $m$-many points with $m\le M$. We argue that a smallest
	such set of $m$ additional points is uniquely determined, and that that
	there is a Presburger definable $\gamma\colon\ZZ^p\to\NN^m$
	that yields these. This function is logspace-computable by
	\cref{compute-presburger}.  If we are given $\vec{t}\in\ZZ^p$ and
	an automaton, we compute the numbers
	$\gamma(\vec{t})=\{f_1,\ldots,f_m\}$ and know that
	$S[\vec{t}]\cup\{f_1,\ldots,f_m\}$ is $\delta$-upward closed. It thus
	suffices to decide whether for some residue $r\in[0,\delta-1]$, we can
	reach an $x\in\NN$ that (i)~is congruent $r$ modulo $\delta$,
	(ii)~lies above the minimal number of $S[\vec{t}]$ within $r$'s residue
	class modulo $\delta$, and (iii)~differs from $f_1,\ldots,f_m$. We construct a VASS that, in its state, checks conditions (i) and
	(ii) (recall that $\delta$ is a constant). However, $f_1,\ldots,f_m$
	are not constants (they are computed from the input). To make sure that
	$x\notin\{f_1,\ldots,f_m\}$, we guess primes $p_1,\ldots,p_m$
	with logarithmically many bits and then check in the state that $x
	\not\equiv f_i\pmod{p_i}$, for each $i\in[1,m]$. 
\end{proof}

\subsection{Coverability in $\AC^1$: The algorithm}
We now show that coverability in $1$-VASS is in $\AC^1$. This improves on the
$\NC^2$ upper bound shown by Almagor, Cohen, P\'{e}rez, Shirmohammadi, and
Worrell~\cite[Theorem 23]{AlmagorCPSW20}. Their algorithm, roughly speaking,
maintains a set of paths, and expands this set $\lceil\log n\rceil$ times (for
$n$ states). After $k$ steps, the current set $P_k$ of paths is large enough so
that anything that can be covered in $\le 2^k$ steps, can also be covered using
$P_k$. To prevent exponential growth, the set is pruned each time. This
involves operations on paths in $\NC^1$ in each step, hence $\NC^2$ overall.
Our algorithm can be viewed as decomposing the paths in $P_k$ so that it
suffices to store a single number for each piece. Then each step is in
$\AC^0$, yielding $\AC^1$ overall.

We first describe our algorithm as a simple combinatorial procedure. Its correctness is not obvious, but will be proven afterwards via an algebraic description.
A \emph{coverability table} $\covertable{p}{q}$, for each pair of states $p, q$, is a list of pairs of natural numbers $\langle u_1,v_1 \rangle, \ldots, \langle u_m,v_m \rangle$ such that $\config{p}{u}$ can cover $\config{q}{v}$ if and only if there is an $i$ with $u\ge u_i$ and $v_i+(u-u_i)\ge v$. 
We will show that in $\AC^1$, given an \emph{acyclic} \hbox{$1$-VASS}, we can compute a coverability table $\covertable{p}{q}$ in $\AC^1$.

Computing coverability tables for acyclic \hbox{$1$-VASS} in $\AC^1$ suffices for deciding coverability in (not necessarily acyclic) \hbox{$1$-VASS} in $\AC^1$ for the following reason. 
Let $\Aa$ be a $1$-VASS with $n$ states $Q$.
Consider the acyclic \hbox{$1$-VASS} $\cA'$ that is obtained from $\cA$ by simulating runs of length at most $n$.
More precisely, let $Q' \sset \set{ q_\ell : q \in Q \text{ and } \ell \in [0, n]}$ be the states of $\cA'$; the idea is that there are $n+1$ copies of each state $q$ in $\Aa$ and, for $k, \ell \in [0,n]$, the state $q_\ell$ in $\cA'$ can be reached from $p_k$ if there is a run from $p$ to $q$ of length exactly $\ell - k$ in $\Aa$.
Note that $\config{p}{u}$ can cover $\config{q}{v}$ in $\cA$ if and only if  
\begin{enumerate}[(i)]
	\item for some $\ell \in [0, n]$, $\config{p_0}{u}$ can cover $\config{q_\ell}{v}$ in $\cA'$, or
	\item there is a state $r \in Q$, $\ell \in [0, n]$, and an entry $\langle u_i, v_i \rangle$ in the coverability table $\covertable{r_0}{r_\ell}$ of $\cA'$ such that $\config{p_0}{u}$ can cover $\config{r_\ell}{u_i}$ in $\cA'$ and $v_i > u_i$, and $q$ is reachable from $r$ in the underlying finite automaton.
\end{enumerate}

Let us now see how to compute $\covertable{p}{q}$ in $\AC^1$. 
Suppose we are given an acyclic $1$-VASS $\cA$ with $n$ states $Q$. 
Our algorithm performs a simple translation, $\log(n)$ times. 
The translation is the following. 
For a path $\pi$ whose transitions have weights $w_1, \ldots, w_m$, we define its \emph{amplitude} as follows: 
(i)~if all partial sums $w_1+\cdots+w_i$ are non-negative, then the amplitude is the sum $w_1+\cdots+w_m$, or 
(ii)~if some partial sum $w_1+\cdots+w_i$ is negative, then the amplitude is the smallest such partial sum.
The $1$-VASS $\cA^\dagger$ is obtained from $\cA$ by removing all transitions and then for all states $p, q$ introducing the transition $p\xrightarrow{a} q$ where $a$ is the maximal amplitude of a path of length $\le 3$ from $p$ to $q$.
Since there are only polynomially many paths of length $\leq 3$, we can compute the maximal amplitude in $\AC^0$.
Accordingly, we can compute $\cA^\dagger$ from $\cA$ in $\AC^0$. 

Our algorithm builds $\cA_0$ by taking $\cA$ and adding, for each state $p$, the self-loop $\tuple{p, 0, p}$.
Subsequently, the algorithm computes $\cA_{k+1} \coloneqq \cA_k^{\dagger}$ for every $k \in [0, \lceil\log(n)\rceil]$. 
We will see that for the final 1-VASS $\cA_k$, it is true that $\config{p}{u}$ can cover $\config{q}{v}$ in $\cA$ if and only if $\config{p}{u}$ can cover $\config{q}{v}$ in $\cA_k$ \emph{in at most two steps}. 
Hence from~$\cA_k$, one can compute the coverability table $\covertable{p}{q}$ for $\cA$ in logspace by considering all (quadratically many) paths of length from $p$ to $q$ of length \hbox{at most $2$} in $\cA_k$.

\subsection{Coverability in $\AC^1$: The proof}

Our correctness proof for the algorithm above views the set of \emph{coverability tables as a semiring}.
Let $P$ be a set of paths from a state $p$ to $q$.
We call $f\colon \NN\cup\set{-\infty} \to \NN\cup\set{-\infty}$ the \emph{coverability function for $P$} if $f(i)$ is the maximal counter value one can obtain in $q$ when starting from $p$ with counter value $i$ (and if there is no run to $q$ when starting from $p$ with counter value $i$, then $f(i) = -\infty$).
Then $f$ is strictly monotone.

The coverability table $\covertable{p}{q}$ is represented by the coverability function for the set of all paths from $p$ to $q$.
More precisely, when $\covertable{p}{q} = \langle u_1, v_1\rangle,\ldots,\langle u_m, v_m \rangle$, we would get the coverability function $f$ with
\begin{equation*}
	f(i) = 
	\begin{cases}
		-\infty & \text{if $i < u_1$;}\\
		v_j + k & \text{if there is $j \in [1, m], k \in [0, u_{j+1} - u_j) $} \\ 
				& \text{such that $i = u_j+k$; and }\\
		v_m + k & \text{if there is $k \geq 0$ such that $i = u_m + k$.}
	\end{cases}
\end{equation*}
Note that the coverability table entries $\tuple{u_1, v_1}, \ldots, \tuple{u_m, v_m}$ can be recovered from $f$ as its ``discontinuities'': the values $i$ for which \hbox{$f(i) > f(i-1) + 1$}. 
See~\Cref{fig:coverability-function} for an example.

\begin{figure}
	\centering
	\begin{tikzpicture}
	\draw[line width = 0.1mm] (-0.05, -0.5) -- (0,-0.5);
	\node at (-0.4, -0.5) {\small$-\infty$};
	\draw[line width = 0.1mm] (-0.05, 0) -- (0, 0);
	\node at (-0.2, 0) {\small$0$};
	\draw[line width = 0.1mm] (-0.05, 0.5) -- (0, 0.5);
	\node at (-0.2, 0.5) {\small$2$};
	\draw[line width = 0.1mm] (-0.05, 2.25) -- (0, 2.25);
	\node at (-0.2, 2.25) {\small$9$};

	\draw[line width = 0.1mm] (2, 0) -- (2, -0.05);
	\node at (2, -0.25) {\small$4$};
	\draw[line width = 0.1mm] (4.5, 0) -- (4.5, -0.05);
	\node at (4.5, -0.25) {\small$9$};

	\node (d0) at (0, -0.495) {}; %
	\node (e0) at (1.5, -0.495) {}; %
	\draw[line width = 0.1mm, blue] (d0.center) -- (e0.center);	
	\node (r0) at (0, -0.505) {}; %
	\node (s0) at (1.5, -0.505) {}; %
	\draw[line width = 0.1mm, red] (r0.center) -- (s0.center);
	\node[circle, fill=black, inner sep = 0.6pt] at (0.5, -0.5) {}; %
	\node[circle, fill=black, inner sep = 0.6pt] at (1, -0.5) {}; %
	\node[circle, fill=black, inner sep = 0.6pt] at (1.5, -0.5) {}; %

	\node[circle, fill=blue, inner sep = 0.6pt] (d1) at (2, 0.5) {}; %
	\node[circle, fill=blue, inner sep = 0.6pt] at (2.5, 0.75) {}; %
	\node[circle, fill=blue, inner sep = 0.6pt] at (3, 1) {}; %
	\node[circle, fill=blue, inner sep = 0.6pt] at (3.5, 1.25) {}; %
	\node[circle, fill=blue, inner sep = 0.6pt] (e1) at (4, 1.5) {}; %
	\draw[line width = 0.2mm, blue] (d1.center) -- (e1.center);

	\node[circle, fill=blue, inner sep = 0.6pt] (d2) at (4.5, 2.25) {}; %
	\node[circle, fill=blue, inner sep = 0.6pt] at (5, 2.5) {}; %
	\node[circle, fill=blue, inner sep = 0.6pt] at (5.5, 2.75) {}; %
	\node[rotate = 22.5, blue] (e2) at (6, 3) {\tiny$...$}; %
	\draw[line width = 0.2mm, blue] (d2.center) -- (e2);

	\draw[line width = 0.2mm, dashed, blue, opacity = 0.4] (d1) -- (e0);
	\node[gray, blue, opacity = 0.4] at (1.2, 0.25) {\tiny first discontinuity};
	\draw[line width = 0.2mm, dashed, blue, opacity = 0.4] (d2) -- (e1);
	\node[blue, opacity = 0.4] at (3.4, 1.875) {\tiny second discontinuity};

	\node[blue] at (5.25, 2.9) {\small$f$};

	\node (s1) at (2, 0) {}; %
	\node[circle, fill=red, inner sep = 0.6pt] at (2.5, 0.25) {}; %
	\node[circle, fill=red, inner sep = 0.6pt] at (3, 0.5) {}; %
	\node[circle, fill=red, inner sep = 0.6pt] at (3.5, 0.75) {}; %
	\node[circle, fill=red, inner sep = 0.6pt] at (4, 1) {}; %
	\node[circle, fill=red, inner sep = 0.6pt] at (4.5, 1.25) {}; %
	\node[circle, fill=red, inner sep = 0.6pt] at (5, 1.5) {}; %
	\node[circle, fill=red, inner sep = 0.6pt] at (5.5, 1.75) {}; %
	\node[rotate = 22.5, red] (s2) at (6, 2) {\tiny$...$}; %
	\draw[line width = 0.2mm, red] (s1.center) -- (s2);

	\node[red] at (5.25, 1.25) {\small$\sigma(f)$};

	\draw[-{Stealth[width=1.5mm, length=2mm]}, line width = 0.3mm] (0,0) -- (6.5,0);
	\draw[-{Stealth[width=1.5mm, length=2mm]}, line width = 0.3mm] (0,-0.015) -- (0,3.2);
	\draw[line width = 0.3mm, dashed, dash pattern={on 1.5pt off 1pt}] (0,0) -- (0,-0.45);
	\draw[line width = 0.3mm] (0,-0.45) -- (0,-0.55);
\end{tikzpicture}
	\vspace{-0.1in}
	\caption{
		Drawn in \textcolor{blue}{blue} is the coverability function $f$ (labelled with its two discontinuities) for the $\tuple{p,q}$ coverability table $\tuple{4, 2}, \tuple{9, 9}$.
		Drawn in \textcolor{red}{red} is $\sigma(f) = \bar{X}^4$.
		For the definition of $\bar{X}$, see the text above~\cref{discontinuities}; and for the definition of $\sigma$, please turn to~\cpageref{def:sigma}. 
		Note that the $y$-axis has been scale by a factor of $0.5$.
	}
	\label{fig:coverability-function}
\end{figure}
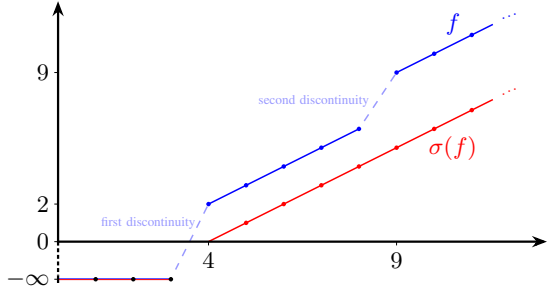

Now, consider the set $\calF$ of all functions $f\colon\NN\cup\{-\infty\}\to\NN\cup\{-\infty\}$ that are \emph{strictly monotone} (i.e. $i < j$ implies \hbox{$f(i) < f(j)$}), that have $f(-\infty) = -\infty$, and that have only finitely many \emph{discontinuities}, meaning numbers $i \in \N$ where \hbox{$f(i) > f(i-1) + 1$}. Then $\calF$ is a semiring, where 
\begin{enumerate}[(1)]
	\item the multiplication operation is function composition: $(fg)(i)\coloneqq g(f(i))$; 
	\item the addition operation is the pointwise maximum:\linebreak $(f+g)(i)\coloneqq\max\set{f(i),g(i)}$; and 
	\item the identity function is the multiplicative neutral and the constant $-\infty$ is the additive neutral. 
\end{enumerate}

We write $X$ for the element $f\in\calF$ with $f(i)=i+1$ for all $i \in \NN$ and $f(-\infty) = -\infty$.
We write $\bar{X}$ for the element $f\in\calF$ with $f(i) = i-1$ for all $i > 0$ and \hbox{$f(-\infty), f(0) = -\infty$}.
Note that, if $D = \set{\langle i,f(i)\rangle \mid f(i) > f(i-1)+1}$ is the set of discontinuities of $f$, then 
\begin{equation} \label{discontinuities} 
	f=\sum_{\langle i,j\rangle\in D} \bar{X}^iX^j. 
\end{equation}

Let us now see what addition and multiplication of $\calF$ mean in terms of coverability. A \emph{path} in a 1-VASS is a sequence \hbox{$\pi = e_1 \cdots e_m$} of transitions $e_1,\ldots, e_m$ where, for each \hbox{$i \in [1, m-1]$}, the end state of $e_i$ is the start state of $e_{i+1}$. 
To each transition $e = \tuple{p,w,q}$, we associate an element $f_e$ as follows. If $w\ge 0$, then $f_e=X^w$. If $w<0$, then $f_e=\bar{X}^{-w}$. For a path $\pi=e_1\cdots e_m$, we define $f_\pi\coloneqq f_{e_1}\cdots f_{e_m}$. 
Furthermore, for a pair of states $p,q$, let $P$ be a set of paths from $p$ to $q$; we define $f_P \coloneqq \sum_{\pi\in P} f_\pi$. 
See~\cref{app:functions-are-tables} for the following:

\begin{restatable}{claim}{functionsAreTables}\label{clm:functions-are-tables}
	Let $p,q$ be a pair of states and let $P$ be a set of paths from $p$ to $q$; $f_P$ is the coverability function for $P$. 
\end{restatable}

In our algorithm, all occurring 1-VASS have the same set of states; let us assume that the set of states is fixed to \hbox{$Q = \set{1,\ldots,n}$}.
Given an acyclic 1-VASS $\cV$, its \emph{adjacency matrix} $\Adj(\cV)$ is the following
$n\times n$ matrix $\Adj(\cV)=A$ over $\calF$. 
For a pair of states $p,q \in Q$ such there is an transition $e$ from $p$ to $q$, we set $(A)_{p,q} = f_e$; if there is no transition from $p$ to $q$, then $(A)_{p,q}=0$.
From now on, we write $A\coloneqq \Adj(\cA)$ and $A_k\coloneqq\Adj(\cA_k)$, where $\cA$ and $\cA_k$ are the 1-VASS occurring in our algorithm. Recall that $\cA_0$ is obtained from $\cA$ by adding the self-loop $\tuple{p,0,p}$ at every state $p \in Q$. 
Thus, we have $A_0=A+I_n$, where $I_n$ is the $n\times n$ identity matrix over $\calF$.

Observe that $(A^k)_{p,q}=\sum_{\pi\in P_{p,q}^{=k}} f_\pi$, where $P_{p,q}^{=k}$ is the set of paths from $p$ to $q$ in $\cA$ of length $k$. 
Since $A_0^k=\sum_{i=0}^k A^i$, this implies that $(A_0^k)_{p,q}=\sum_{\pi\in P_{p,q}^{\le k}} f_\pi$, where $P_{p,q}^{\le k}$ is the set of all paths in $\cA$ from $p$ to $q$ of length $\le k$. 
By our observation above, this means $(A_0^n)_{p,q} = f_{P_{p,q}}$, where $P_{p,q}$ is the set of all paths from $p$ to $q$. Hence, \cref{clm:functions-are-tables} implies:

\begin{lemma}\label{matrix-power-coverability}
	$(A_0^n)_{p,q}$ is $\covertable{p}{q}$.
\end{lemma}

\begin{remark}\label{remark-vass-matrix-product}
	Note that by \cref{matrix-power-coverability}, we could also
	use repeated squaring of $A_0$ to decide coverability. Unfortunately,
	computing the product of two matrices over $\calF$ seems to require
	$\TC^0$, so this would only yield a $\TC^1$ upper bound overall.
\end{remark}

\cref{matrix-power-coverability} tells us that correctness of our algorithm is established, if we can prove the following proposition.

\begin{proposition}\label{correctness-coverability}
	For $k\ge\lceil\log n\rceil$, we have $A_kA_k=A_0^n$.
\end{proposition}

Indeed, this would mean that the coverability function of all paths from $p$ to $q$ of length at most two in $\cA_k$ is the same as the coverability table $\covertable{p}{q}$ for $\cA$. 
Thus, the remainder of this subsection is devoted to proving \cref{correctness-coverability}.

We want to describe $A_{k+1}$ algebraically in terms of $A_k$.
In a semiring $S$, for $s,t\in S$, we write $s\le t$ if there is \hbox{$r\in S$} such that $t=s+r$.
This is a quasi-ordering, and also multiplicative: if $s\le t$ and $s'\le t'$ with $t=s+r$ and $t'=s'+r'$, then $tt'=(s+r)(s'+r')=ss'+sr'+rs'+rr'$, hence $ss'\le tt'$. Moreover, in the semiring of $n\times n$ matrices over $S$, we have $M\le N$ if and only if $(M)_{i,j}\le (N)_{i,j}$ for every $i,j \in [1,n]$.
For $f,g\in\calF$, we have $f\le g$ iff $f(i)\le g(i)$ for every $i\in\N$.

An element of $\calF$ is \emph{simple} if it is $0$ or of the form $X^i$ or $\bar{X}^i$ for some $i\in\N$. 
We denote the set of simple elements by $\hat{\calF}$. A matrix over $\calF$ is \emph{simple} if all its entries are simple. 
Consider the map $\sigma\colon\calF\to\hat{\calF}$, where $\sigma(f)$ is the largest simple $g\in\calF$ for which $g\le f$. \label{def:sigma}
This is well-defined: 
Each $f\in\calF$ can be written as in~\cref{discontinuities}. If $D=\emptyset$, then $\sigma(f)=0$. 
If $D$ contains a pair $\langle 0,j\rangle$, then $\sigma(f)=X^j$ for the largest such $j$. 
If $D$ only contains pairs $\langle i,j\rangle$ with $i>0$, then $\sigma(f)=\bar{X}^i$ for the smallest such $i$. 
See~\Cref{fig:coverability-function} for an example where $D$ only contains pairs $\langle i,j\rangle$ with $i>0$.
The following is straightforward (see~\cref{app:sigma-properties}):
\begin{restatable}{lemma}{sigmaProperties}\label{sigma-properties}
	The map $\sigma\colon\calF\to\calF$ is a monoid homomorphism of $\tuple{\calF,+}$. In particular, $f\le g$ implies \hbox{$\sigma(f)\le\sigma(g)$}. Moreover, $\sigma(f)\le f$ for every $f\in\calF$.
\end{restatable}

Moreover, $\sigma$ allows us to decribe the translation
\hbox{$\cA_k\leadsto\cA_{k+1}$} algebraically. 
It is easy to observe the following lemma (a detailed proof can be found in~\cref{app:update-algebraically}).
\begin{restatable}{lemma}{updateAlgebraically}\label{update-algebraically}
	For every $k\ge 0$, $A_{k+1}=\sigma(A_kA_kA_k)$.
\end{restatable}
This will allow us to prove for every $k\ge 0$:
\begin{equation}\label{sandwich}
	A_0^{2^k} \le A_kA_k \le A_0^{3^k}.
\end{equation} 
Since for $k\ge \lceil \log n\rceil$ we have $A_0^{2^k}=A_0^n=A_0^{3^k}$ (recall $A_0^{n+m}=A_0^n$ for $m\ge 0$ because of acyclicity of $\cA$) and thus $A_0^n\le A_kA_k\le A_0^n$, which yields~\cref{correctness-coverability}. 
Hence, it remains to prove~\eqref{sandwich}, which we do inductively.
Here, $k=0$ is trivial, and for $k>0$, we use a simple, but key observation:
\begin{lemma}\label{minimal-decomposition}
	If $s_1,\ldots,s_m\in\hat{\calF}$, then
	\begin{equation*}
		s_1\cdots s_m=\sum_{1\le i<m} \sigma(s_1\cdots s_i)\sigma(s_{i+1}\cdots s_m).
	\end{equation*}
\end{lemma}
\begin{proof}
	Clearly, since $\sigma(fg)\le fg$ for every $f,g\in\calF$ (by definition of $\sigma$), the right-hand side is $\le$ the left-hand
	side. Conversely, consider the \emph{effect} $\delta(s)\in\Z$ of a
	simple $s\in\hat{\calF}$.
	For $s=\bar{X}^u$, let $\delta(s) \coloneqq -u$; and for $s=X^v$, let $\delta(s) \coloneqq v$. 
	Pick $i \in [1,m-1]$ so that $\delta(s_1) + \cdots + \delta(s_i)$ is minimal. 
	Then clearly, $s_1 \cdots s_i$ and $s_{i+1} \cdots s_m$ are both simple, and thus
	$\sigma(s_1\cdots s_i)=s_1\cdots s_i$ and $\sigma(s_{i+1}\cdots
	s_m)=s_{i+1}\cdots s_m$. Hence, $\sigma(s_1\cdots
	s_i)\sigma(s_{i+1}\cdots s_m)$ is $\ge$ the left-hand side.
\end{proof}

\cref{minimal-decomposition} yields the first inequality in:
\begin{equation}\label{sandwich-step}
	\begin{split}
	A_k^4\le\sigma(A_k)\sigma(A_k^3)+\sigma(A_k^2)\sigma(A_k^2)+\sigma(A_k^3)\sigma(A_k) \\
	\le\underbrace{\sigma(A_k^3)\sigma(A_k^3)}_{=A_{k+1}A_{k+1}} \le A_k^6\,. 
	\end{split}
\end{equation}
The second inequality follows from multiplicativity of $\le$, monotonicity of $\sigma$, and from $I_n\le A_k$, which implies \hbox{$A_k^i\le A_k^3$} for $i\le 3$. 
The third inequality follows from $\sigma(A_k^3)\le A_k^3$ and again multiplicativity of $\le$. 
Now the induction step of \eqref{sandwich} can be carried out.
If~\eqref{sandwich} holds for $k$, then~\eqref{sandwich-step} implies $A_0^{2^{k+1}}\le A_k^4\le A_{k+1}^2\le A_k^6\le A_0^{3^{k+1}}$, and thus~\eqref{sandwich} for $k+1$.

\subsection{$\NP$-hardness in dichotomy and decidability}
We briefly sketch (i)~$\NP$-hardness in \cref{main-result-vass} and (ii)~\cref{main-result-vass-decidable} (details in \cref{app:vass-hardness-decidability}). We show that with the constant $B\ge 1$ from \cref{lem:modulo-free}, the $S\subseteq\ZZ^{p}\times\NN$ is not uniformly quasi-upward closed if and only if for some \hbox{$\vec{b}\in[0,B-1]^{p+1}$}, the set $[S]_{B,\vec{b}}$ has \emph{unbounded gaps}. This means, for each $g\ge 1$, there is a $\vec{t}\in\ZZ^p$ such that $S[\vec{t}]$ contains a gap of size $\ge g$ above some $x\in S[\vec{t}]$.
If $S$ itself has unbounded gaps, then subset sum easily reduces to $\vasreach(S)$, and a slight modification works when some $[S]_{B,\vec{b}}$ has unbounded gaps. Since ``unbounded gaps'' is expressible in Presburger arithmetic, this also implies \cref{main-result-vass-decidable}.

\section{Conclusion}
We have provided dichotomies for the reachability problem in systems with one counter  and integer, natural, or VASS semantics. 
This raises several further questions.

First, which complexities arise in higher dimension? For fixed-dimensional
VASS, even ordinary reachability is not
understood~\cite{czerwinski_et_al:LIPIcs.FSTTCS.2023.35}, meaning a full dichotomy
for semilinear target sets appears to be a longterm goal. For natural
semantics, there are concrete target sets $S\subseteq \NN^{0+2}$ (i.e., 
with $p=0$) where it is not clear whether
$\natreach(S)$ is in polynomial time, such as $S=\{\langle
x_1,x_2\rangle\in\NN^2 \mid x_2\in [x_1,2x_1]\}$. Thus, \hbox{$\NP$ vs.\ $\compP$}
dichotomies in higher dimension for integer and natural semantics seem
challenging as well.

This work reveals a variety of decision problems about systems with one counter
that can be solved in $\AC^1$.  
Here, $\AC^1$ seems to be a difficult barrier to break.
Since at least 1990, it is open whether the $\AC^1$ upper bound for
the shortest-path problem in automata weighted with naturals can be improved
(see \cite[p.~13]{cook1985taxonomy} and \cite[p.~32]{chandra1990complexity});
and this is a special case of coverability in both integer and VASS semantics.

\section*{Acknowledgments}
\noindent\raisebox{-9pt}[0pt][0pt]{\includegraphics[height=.8cm]{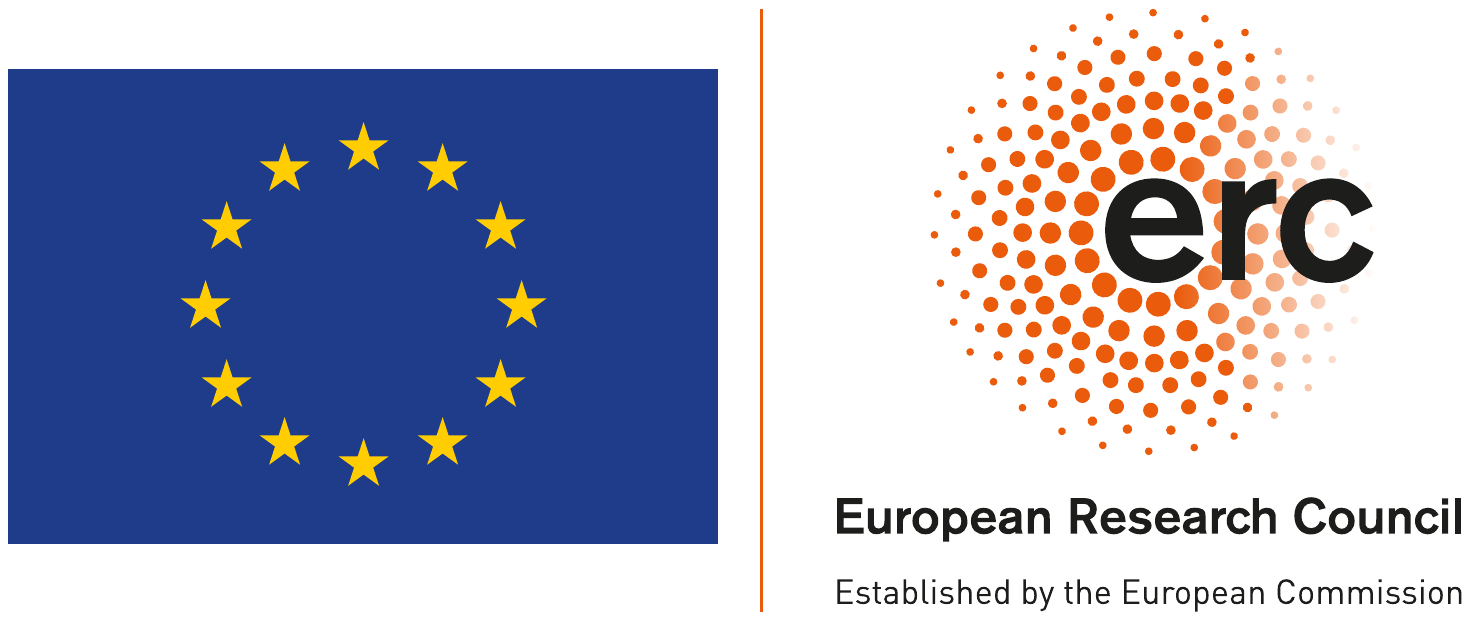}}
Funded by the European Union (ERC, FINABIS, 101077902).
Views and opinions expressed are however those of the authors only and do not necessarily reflect those of the European Union or the European Research Council Executive Agency. 
Neither the European Union nor the granting authority can be held responsible for them.

Henry Sinclair-Banks is supported by the ERC grant \hbox{INFSYS}, agreement \hbox{no.~950398}.

\label{beforebibliography}
\newoutputstream{pages}
\openoutputfile{main.pages.ctr}{pages}
\addtostream{pages}{\getpagerefnumber{beforebibliography}}
\closeoutputstream{pages}

\printbibliography

\pagebreak

\appendices
\crefalias{section}{appsec}
\crefalias{subsection}{appsec}

\gdef\thesubsection{\Roman{subsection}.}%
\gdef\thesubsectiondis{\Roman{subsection}.}%
\renewcommand\thesubsection{\thesection.\Roman{subsection}}

	\section{Additional material on Section~\ref{sec:main-results}}
        \label{app:main-results}
	\subsection{The case of one-counter automata}
\label{app:one-counter-automata}

Let us briefly mention that the case of one-counter automata (OCA) with binary
updates~\cite{DBLP:conf/concur/HaaseKOW09}---which work like $1$-VASS, but
allow a zero-test instruction---is not interesting in terms of complexity
dichotomies with target sets $S\subseteq\ZZ^{p+1}$.

Indeed, for any semilinear $S\subseteq\ZZ^{p}\times\N$, the problem for OCA
would be trivial (the answer to every instance would be ``no'') if
$S[\vec{t}]=\emptyset$ for all $\vec{t}$. Moreover, if $S[\vec{t}]\ne\emptyset$
for any $\vec{t}\in\ZZ^p$, then the reachability problem for $S$ would be
$\NP$-complete. This is just because using the zero tests, we can always reduce
subset sum; and after such a zero test, we can add some fixed target number
$x\in S[\vec{t}]$. Then, we can reach $S[\vec{t}]$ if and only if the subset
sum instance is positive.

\subsection{Additional integer semantics density properties}
\label{app:z-density-properties}

\begin{restatable}{fact}{densityMultiply}\label{density-multiply}
	For every set $A\subseteq\ZZ$ and $k,\ell\ge 1$ and $x\in A$, we have
	$\density_k(A,x)\ge \tfrac{1}{2\ell}\density_{k\ell}(A,x)$.
\end{restatable}
\begin{proof}
	Take any $n\in\NN$ and write $n=m\ell+r$ for some $r\in[0,\ell-1]$. Then 
	\begin{align*}
		| A \cap (&x + k\cdot[-n,n]) | \\
		& = | A \cap [x-kn, x+kn] \cap (x + k\ZZ) | \\
		& \geq | A \cap [x-kn,x+kn] \cap (x+k\ell\ZZ) | \\
		& = |A\cap [x-k(m\ell+r),x+k(m\ell+r)]\cap (x+k\ell\ZZ)| \\
		& \geq |A\cap [x-k\ell m,x+k\ell m]\cap (x+k\ell\ZZ)| \\
		& = | A \cap (x + k\ell\cdot[-m,m]) |.
	\end{align*}
	Therefore,
	\begin{align*} 
		\frac{| A \cap (x + k\cdot[-n,n]) |}{2n+1} 
		& \geq \frac{| A \cap (x + k\ell\cdot[-m,m]) |}{2(m\ell+r)+1} \\
		& \geq \frac{1}{2\ell} \cdot \frac{| A \cap (x + k\ell\cdot[-m,m]) |}{2m+1} \\
		& \geq \tfrac{1}{2\ell}\cdot \density_{k\ell}(A,x),
	\end{align*}
	where the second inequality follows from 
	\[ 2\ell\cdot(2m+1)=4\ell m+2\ell\ge 2m\ell+2r+1=2(m\ell+r)+1. \]
	This shows that $\density_k(A,x)\ge\tfrac{1}{2\ell}\cdot \density_{k\ell}(A,x)$.
\end{proof}

\begin{claim}\label{clm:z1}
	Let $S_1 \coloneqq -2\cdot\NN \cup (2\cdot\NN+1)$; $\intreach(S_1)$ is in $\AC^1$.
\end{claim}
\begin{proof}
	Suppose $k$ is even. 
	If $x\in S_1$, then $S_1 \cap (x + k\cdot[-n,n])$ contains at least $\tfrac{n}{2}$ integers.
	If $x$ is positive, then $x$ must be odd, and thus $S_1 \cap (x + k\cdot[-n,n])$ contains all odd numbers in $[x, x + kn]$. 
	If $x$ is non-positive, then $x$ is even and thus $S_1 \cap (x + k\cdot[-n,n])$ contains all even numbers in $[x - kn, x]$. 
	Hence, we have $\density_k(S_1,x) \geq \tfrac{1}{2}$ for all even $k$.

	Moreover, if $k$ is odd, then by~\cref{density-multiply}, we have $\density_{k}(S_1,x)\ge \tfrac{1}{4}\cdot \density_{2k}(S_1,x)$, which is at least $\tfrac{1}{8}$ by the previous case. 
	Hence, we have $\density_k(S_1,x)\ge\tfrac{1}{8}$ for any $k$.
	This implies that $\density(S_1) \geq \tfrac{1}{8}$ and so, by~\cref{main-result-integer}, this implies that $\intreach(S_1)$ is in $\AC^1$.
\end{proof}

\begin{example}\label{exa:z6}
	At first glance, it might seem that if $S\subseteq\ZZ^{p+1}$ has positive local density and we have $S' \supseteq S$, then $S'$ would also have positive local density.
	However, note that we measure the density $\density_k(S[\vec{t}],x)$ only for $x\in
	S[\vec{t}]$. 
	This means, each additional point $x$ in the set incurs the additional requirement that the density at $x$ must be positive. 

	For example, let $S_6[t] \coloneqq (-\infty,0]$ and $S'_6[t] \coloneqq (-\infty,0] \cup \set{t}$. 
	We have $\density(S_6) = \frac{1}{2}$, so $S_6$ has positive local density.
	However, $\density_1(S'_6[t], t) = \tfrac{1}{2t+1}$ for every $t\ge 0$, so $S'_6$ has zero density.
	It is therefore the case that even though $S' \supseteq S$ and $\density(S) > 0$, $\intreach(S_6)$ is in $\AC^1$, whereas $\intreach(S'_6)$ is $\NP$-complete.
\end{example}

\begin{example}\label{exa:z7}
	Despite~\cref{exa:z6}, the union of two sets with positive local density will again have positive local density. 
	However, intersecting sets does not preserve this property.
	Consider the sets $S_7, S'_7 \subseteq \ZZ^{p+1}$ for $p=1$ with $S_7[t] \coloneqq (-\infty,t]$ and $S'_7[t] \coloneqq [t,\infty)$.
	Both sets clearly have positive local density but their intersection $\set{t}$ does not.
\end{example}

\subsection{Additional natural semantics density properties}
\label{app:n-density-properties}

\begin{restatable}{fact}{positiveDensityMultiply}\label{positive-density-multiply}
	For every set $A\subseteq\NN$ and $k,\ell\ge 1$ and $x\in A$, we have
	$\densityplus_k(A,x)\ge \tfrac{1}{2\ell}\densityplus_{k\ell}(A,x)$.
\end{restatable}
\begin{proof}
	Take any $n\in\NN$ with $x-kn\ge 0$ and write $n=m\ell+r$ for some $r\in[0,\ell-1]$. Then 
	\begin{align*}
		| A \cap (&x + k\cdot[-n,n]) | \\
		& = |A\cap [x-kn,x+kn]\cap (x+k\ZZ)| \\
		& \ge |A\cap [x-kn,x+kn]\cap (x+k\ell\ZZ)| \\
		& = |A\cap [x-k(m\ell+r),x+k(m\ell+r)]\cap (x+k\ell\ZZ)| \\
		& \ge |A\cap [x-k\ell m,x+k\ell m]\cap (x+k\ell\ZZ)| \\
		& = | A \cap (x + k\ell\cdot[-m,m]) |.
	\end{align*}
	Therefore,
	\begin{align*} 
		\frac{| A \cap (x + k\cdot[-n,n]) |}{2n+1}
		& \geq \frac{| A \cap (x + k\ell\cdot[-m,m]) |}{2(m\ell+r)+1} \\
		& \geq \frac{1}{2\ell} \cdot \frac{| A \cap (x + k\ell\cdot[-m,m]) |}{2 m+1} \\
		& \geq \tfrac{1}{2\ell} \cdot d^+_{k\ell}(A,x),
	\end{align*}
	where the second inequality follows from 
	\[ 2\ell\cdot(2m+1)=4\ell m+2\ell\ge 2m\ell+2r+1=2(m\ell+r)+1. \]
	Finally, notice that $x-k\ell m\ge 0$, because $0\le x-kn=x-k(m\ell+r)=x-k\ell m-kr$.
	This shows that $\densityplus_k(A,x)\ge\tfrac{1}{2\ell}\cdot \densityplus_{k\ell}(A,x)$.
\end{proof}

\noParametersNat*
\begin{proof}
	First, observe that union of two sets with positive local${}^+$ density will again have positive local${}^+$ density. 
	Moreover, every semilinear $S\subseteq\NN$ is a finite union of arithmetic progressions $a+b\cdot\NN$, where $a,b\ge 0$. 
	Thus, it suffices to argue that sets $a+b\cdot\NN$ have positive local${}^+$ density. 

	If $b\ge 1$, note that $\densityplus_{bk}(a+b\cdot\NN,x)\ge \tfrac{1}{2}$ for every $x\in a+b\cdot\NN$ and every $k\ge 1$. 
	Moreover, by \cref{positive-density-multiply}, we have that, for any $k\ge 1$,
	\begin{equation*}
		\densityplus_k(a+b\cdot\NN,x)\ge \tfrac{1}{2b} \cdot d^+_{bk}(a+b\cdot\NN)\ge \tfrac{1}{4b}.
	\end{equation*}
	On the other hand, if $b=0$, then $\densityplus_k(\{a\},a)\ge\tfrac{1}{2a+1}$. 

	Thus, for every $k\ge 1$, $\densityplus_k(S) \geq \min\set{\tfrac{1}{4b+1}, \tfrac{1}{2a+1}} > 0$.
	We therefore conclude that $\natreach(S)$ is in $\AC^1$ by~\cref{main-result-natural}.
\end{proof}

	\section{Additional material on Section~\ref{sec:general-techniques}}
        \label{app:general-techniques}
	\subsection{Proof of Theorem \ref{np-upper-bound}}\label{app:np-upper-bound}
\npUpperBound*
\begin{proof}
	The result for $\natreach(S)$ is a special case of the result for
	$\intreach(S)$, so we need not consider $\natreach(S)$. Given a
	$\Z$-weighted automaton and two states $p,q$, we can construct, both
	for VASS semantics and for integer semantics, a Presburger formula
	$\psi_{p,q}(x,y)$ such that from a configuration $p(u)$, we can reach
	$q(v)$ if and only if $\psi_{p,q}(u,v)$ holds. For VASS semantics, this
	follows from the fact that this can even be done for one-counter
	automata~\cite{LiCWX20}. For integer semantics, a more general fact has
	been shown in \cite[Corollary~10]{DBLP:conf/rp/HaaseH14}.

	Therefore, suppose $S$ is defined by a formula $\varphi$ with $p+1$
	free variables. Given a $\Z$-weighted automaton, states $p,q$, and
	$\vec{t}\in\ZZ^p$, we have reachability of $S[\vec{t}]$ if and only if
	$\exists x\colon \psi_{p,q}(0,x)\wedge \varphi(\vec{t},x)$, which is an
	existential Presburger formula. 
	Since determining the truth of existential Presburger formulas is $\NP$-complete, this completes the proof.
\end{proof}

\subsection{Proof of Lemma \ref{lem:modulo-free}}\label{app:modulo-free}
\moduloFree*
\begin{proof}
	Let $\varphi(\vec{x})$ be a quantifier-free Presburger formula defining
	$S$. Then $\varphi$ may contain modulo constraints $x\equiv r\pmod m$
	for some $r,m\in\NN$. Let $B\ge 1$ be the least common multiple of all
	$m$ in such modulo constraints in $\varphi$. 

	Consider some $\vec{b}\in[0,B-1]^{p+1}$. Notice that whether an
	assignment $\vec{x}=B\vec{u}+\vec{b}$ to the variables $\vec{x}$
	satisfies a modulo constraint $x\equiv r\pmod{m}$ does not depend on
	$\vec{u}$, because the residues of $B\vec{u}+\vec{b}$ modulo $m$ are
	determined by $\vec{b}$. Thus, we can build a formula $\varphi'$ by
	replacing each modulo constraint $x\equiv r\pmod{m}$ by $\top$ or
	$\bot$, so that $\varphi'(B\vec{u}+\vec{b})$ if and only if
	$\varphi(B\vec{u}+\vec{b})$ for every $\vec{u}\in\ZZ^{p+1}$. Notice
	that $\varphi'$ does not contain modulo constraints anymore, and is
	thus modulo-free. Therefore, the formula
	\[ \psi(\vec{x}):=\varphi'(B\vec{x}+\vec{b}) \]
	is modulo-free and defines $[S]_{B,\vec{b}}$.
\end{proof}

\subsection{Computing interval endpoints}\label{app:intervals}
In this subsection, we will talk about functions $\gamma\colon \ZZ^k\to \ZZ\cup\{-\infty,\infty\}$ (or similar) being Presburger definable. This will be meant w.r.t.\ a simple encoding. We encode such a function as the function
\[ \hat{\gamma}\colon\ZZ^k\to \ZZ\times\ZZ, \]
where $\hat{\gamma}(\vec{t})=\tuple{a,0}$ if $\gamma(\vec{t})=a\in\ZZ$.
Moreover, $\hat{\gamma}(\vec{t})=\tuple{0,1}$ if $\gamma(\vec{t})=\infty$, and
$\hat{\gamma}(\vec{t})=\tuple{0,-1}$ if $\gamma(\vec{t})=-\infty$. With this,
we say that $\gamma\colon\ZZ^k\to\ZZ\cup\{-\infty,\infty\}$ is
\emph{Presburger-definable} if $\hat{\gamma}$ is Presburger-definable.

\begin{restatable}{lemma}{intervals}\label{lem:intervals}
	Let $T\subseteq\ZZ^{p+1}$ be defined by a modulo-free Presburger formula.
	Then there exist Presburger-definable functions $\alpha_1, \beta_1, \ldots, \alpha_m, \beta_m : \ZZ^p \to \ZZ \cup \set{-\infty, \infty}$ such that, for every $\vec{t} \in \ZZ^p$, $\alpha_1(\vec{t}) \leq \beta_1(\vec{t}) \leq \ldots \leq \alpha_m(\vec{t}) \leq \beta_m(\vec{t})$ and 
	\begin{equation*}
		T[\vec{t}] = (\alpha_1(\vec{t}), \beta_1(\vec{t})) \cup \cdots \cup (\alpha_m(\vec{t}), \beta_m(\vec{t})).
	\end{equation*}
\end{restatable}
We call $(\beta_1(\vec{t}), \alpha_2(\vec{t})), \ldots, (\beta_{m-1}(\vec{t}), \alpha_m(\vec{t}))$ as well as $(-\infty, \alpha_1(\vec{t}))$ if $\alpha_1(\vec{t})> \infty$ and $(\beta_m(\vec{t}),\infty)$ if $\beta_m(\vec{t})<\infty$ the \emph{gaps} of $T[\vec{t}]$.
Our proof of \cref{lem:intervals} begins with a lemma. We call a modulo-free Presburger formula \emph{conjunctive} if all its logical connectives are conjunctions (i.e.\ it contains no negations and no disjunctions).
\begin{lemma}\label{lem:intervals:single}
	Let $T\subseteq\ZZ^{p+1}$ be defined by a conjunctive modulo-free Presburger formula.
	Then there exist Presburger-definable functions $\alpha, \beta\colon \ZZ^p \to
	\ZZ \cup \set{-\infty, \infty}$ such that, for every $\vec{t} \in
	\ZZ^p$, $\alpha(\vec{t}) \leq \beta(\vec{t})$ and $T[\vec{t}] = (\alpha(\vec{t}),
	\beta(\vec{t}))$.
\end{lemma}
\begin{proof}
	Suppose $T$ is defined be the conjunctive modulo-free Presburger
	formula $\varphi(x_1,\ldots,x_p,y)$. Then $\varphi$ is a conjunction of
	inequalities, which we can write 
	\begin{align*}\vec{u}_i^\top\vec{x}&<y,& y<\vec{v}_j^\top\vec{x} \end{align*}
		for $i=1,\ldots,n$, $j=1,\ldots,m$, where
	$\vec{u}_1,\ldots,\vec{u}_n\in\QQ^p$ and
	$\vec{v}_1,\ldots,\vec{v}_m\in\QQ^p$. Thus, for each vector $\vec{t}\in\ZZ^p$, the set of $y\in\Z$ with $\varphi(\vec{t},y)$ forms an interval: The smallest value among all $\lfloor\vec{u}_i^\top\vec{t}\rfloor$ is the left end-point; the largest value among all $\lceil\vec{v}_j^\top\vec{t}\rceil$ is the right end-point. However, the largest lower bound could be larger-or-equal to the smallest upper bound, causing the set of $y$ to be empty. The latter is the case if the following formula holds.
	\[ \psi(\vec{x}):=\bigvee_{i\in[1,n]}\bigvee_{j\in[1,m]} \vec{u}_i^\top\vec{x}\ge\vec{v}_j^\top\vec{x} \]
	Therefore, we define
	\[ \alpha(\vec{t})=\begin{cases} \lfloor \max(\vec{u}_i^\top\vec{t}\mid i\in[1,n])\rfloor& \text{if $\neg\psi(\vec{t})$ and $n\ge 1$}, \\
		0 & \text{if $\psi(\vec{t})$ and $n\ge 1$}, \\
		-\infty & \text{if $n=0$.}
	\end{cases}\]
	Furthermore, we set
	\[ \beta(\vec{t})=\begin{cases}\lceil \min(\vec{v}_j^\top\vec{t}\mid j\in[1,m])\rceil 
		& \text{if $\neg \psi(\vec{t})$ and $m\ge 1$}, \\
		0 & \text{if $\psi(\vec{t})$ and $m\ge 1$},  \\
		\infty & \text{if $m=0$.}
	\end{cases}
		\]
	Note that if $\psi(\vec{t})$ holds, then $\alpha(\vec{t})=\beta(\vec{t})=0$, hence $(\alpha(\vec{t}),\beta(\vec{t}))=(0,0)=\emptyset$. Since $\alpha$ and $\beta$ are clearly Presburger-definable, this proves the \lcnamecref{lem:intervals:single}.
\end{proof}

Let us now show the general statement of \cref{lem:intervals}:
\begin{proof}[Proof of \cref{lem:intervals}]
	Suppose $T$ is defined by the modulo-free Presburger formula $\varphi(x_1, \dots, x_p, y)$. We may assume that $\varphi$ is in disjunctive normal form, meaning that $\varphi$ is a disjunction of $k$-many conjunctive modulo-free formulas, for some $k \geq 1$. We proceed by induction on $k$. (Here, the disjunctive normal form may still contain negations of inequalities, but since $\neg (a<b)$ is equivalent to $b<a+1$, we can eliminate all negations.)

The case $k = 1$ is already shown in Lemma IV.3. For the case $k+1$, we write $\varphi = \psi \lor \kappa$, where $\psi$ is a $k$-fold disjunction of conjunctive modulo-free formulas, and $\kappa$ is itself conjunctive. By induction, there is an $m \geq 1$ and Presburger-definable functions $\alpha_1, \beta_1, \dots, \alpha_m, \beta_m : \mathbb{Z}^p \to \mathbb{Z} \cup \{-\infty, \infty\}$ such that the set of all $y$ with $\psi(\vec{t}, y)$ is 
\[
(\alpha_1(\vec{t}), \beta_1(\vec{t})) \cup \cdots \cup (\alpha_m(\vec{t}), \beta_m(\vec{t})).
\]
These intervals represent the part of $T[\vec{t}]$ corresponding to $\psi$, with no overlap between them.

	Moreover, by \cref{lem:intervals:single}, there are Presburger-definable functions $\alpha_{m+1}, \beta_{m+1} : \mathbb{Z}^p \to \mathbb{Z} \cup \{-\infty, \infty\}$ such that the set of all $y$ with $\kappa(\vec{t}, y)$ is 
\[
(\alpha_{m+1}(\vec{t}), \beta_{m+1}(\vec{t})).
\]

Thus, the set of all $y$ with $\varphi(\vec{t}, y) = \psi(\vec{t}, y) \lor \kappa(\vec{t}, y)$ is the union of $T_\psi[\vec{t}]$ and $T_\kappa[\vec{t}]$:
\begin{multline*}
	T[\vec{t}] = (\alpha_1(\vec{t}), \beta_1(\vec{t})) \cup \cdots \cup (\alpha_m(\vec{t}), \beta_m(\vec{t})) \\ 
	\cup (\alpha_{m+1}(\vec{t}), \beta_{m+1}(\vec{t})).
\end{multline*}

Since these intervals may overlap, we construct new non-overlapping intervals. Let 
\[
\begin{aligned}
S_\alpha(\vec{t}) &= \{\alpha_1(\vec{t}), \dots, \alpha_m(\vec{t}), \alpha_{m+1}(\vec{t})\}, \\
S_\beta(\vec{t}) &= \{\beta_1(\vec{t}), \dots, \beta_m(\vec{t}), \beta_{m+1}(\vec{t})\}.
\end{aligned}
\]
The first interval is defined as 
\[
(\tilde{\alpha}_1(\vec{t}), \tilde{\beta}_1(\vec{t})) = (\min(S_\alpha(\vec{t})), \min(S_\beta(\vec{t}))).
\]
For $i \geq 2$, define 
\[
\tilde{\alpha}_i(\vec{t}) = \max(\tilde{\beta}_{i-1}(\vec{t})-1, \min(S_\alpha(\vec{t}) \setminus \{\tilde{\alpha}_1(\vec{t}), \dots, \tilde{\alpha}_{i-1}(\vec{t})\})),
\]
\[
\tilde{\beta}_i(\vec{t}) = \min(S_\beta(\vec{t}) \setminus \{\tilde{\beta}_1(\vec{t}), \dots, \tilde{\beta}_{i-1}(\vec{t})\}).
\]
We now have:
\[
	T[\vec{t}] = (\tilde{\alpha}_1(\vec{t}), \tilde{\beta}_1(\vec{t})) \cup \cdots \cup (\tilde{\alpha}_{m+1}(\vec{t}), \tilde{\beta}_{m+1}(\vec{t})).
\]
	This procedure ensures all intervals are non-overlapping, ordered, and fully cover $T[\vec{t}]$. Furthermore, since $\min$, $\max$, and set operations preserve Presburger-definability, the functions $\tilde{\alpha}_i(\vec{t})$ and $\tilde{\beta}_i(\vec{t})$ are Presburger-definable. The only requirement the functions $\tilde{\alpha}_1,\ldots,\tilde{\alpha}_{m+1}$, $\tilde{\beta}_1,\ldots,\tilde{\beta}_{m+1}$ might still be missing is $\tilde{\beta}_{i-1}(\vec{t})\le\tilde{\alpha}_i(\vec{t})$. This means, it is possible that some intervals $(\tilde{\alpha}_{i-1}(\vec{t}),\tilde{\beta}_{i-1}(\vec{t}))$ and $(\tilde{\alpha}_{i}(\vec{t}),\tilde{\beta}_{i}(\vec{t}))$ have an empty gap between them. By distinguishing the finitely many cases of which intervals have distance $1$ to which others, we can now easily modify the functions $\tilde{\alpha}_1,\ldots,\tilde{\alpha}_{m+1}$, $\tilde{\beta}_1,\ldots,\tilde{\beta}_{m+1}$ to guarantee non-empty gaps. Note that if for some $\vec{t}$, $T[\vec{t}]$ contains strictly less than $m+1$ intervals, we can set some of the pairs $\tuple{\tilde{\alpha}_i(\vec{t}),\tilde{\beta}_i(\vec{t})}$ to $\tuple{\ell,\ell}$ for some value $\ell$, which yields an empty interval.
\end{proof}

	\section{Additional material on Section~\ref{sec:integer-semantics}}
        \label{app:integer-semantics}
	\subsection{Local density of modulo-free sets}
We begin this section with an auxiliary lemma that tells us that a modulo-free
set has positive local density if and only if it has positive local
$1$-density. Here, $1$-density of a set $S\sset\ZZ^{p+1}$ means that we only
measure the $1$-density at each point. Let us make this precise.  The local
$k$-density of a set $A\subseteq\ZZ$ is defined as
\[ \density_k(A)=\inf_{x\in A} \density_k(A,x). \]
\begin{lemma}\label{lem:interval-local-density}
	If $A$ is a disjoint union of $m$ intervals, then $\density_k(A)\ge \tfrac{1}{5m}\density_1(A)$.
\end{lemma}
\begin{proof}
	First, observe that if $I$ is an interval, then
	\begin{equation*}
		| I \cap [x-kn, x+kn] | \leq k \cdot | I \cap (x + k\cdot[-n,n]) | + k.
	\end{equation*}
	This implies 
	\begin{align*}
		\frac{| I \cap (x + k\cdot[-n,n]) |}{2n+1} \ge \frac{\tfrac{1}{k} \cdot | I \cap [x-kn, x+kn]| - 1}{2n+1}
	\end{align*}
	Therefore, if $A = I_1 \cup \cdots \cup I_m$ is a disjoint union, then
	\begin{align*}
		\density_{k,n}(A,x)
		& =\frac{|A\cap (x + k\cdot[-n,n])|}{2n+1} \\
		& = \sum_{i=1}^m \frac{| I_i \cap (x + k\cdot[-n,n]) |}{2n+1} \\
		& \ge \sum_{i=1}^m \frac{\tfrac{1}{k} \cdot | I_i \cap [x-kn, x+kn]| - 1}{2n+1} \\
		& = - \frac{m}{2n+1} + \sum_{i=1}^m \frac{|I_i \cap [x-kn, x+kn]|}{2kn+k} \\
		& \ge - \frac{m}{2n+1} + \tfrac{1}{2}\sum_{i=1}^m \frac{|I_i \cap [x-kn, x+kn]|}{2kn+1} \\
		&=\tfrac{1}{2}\density_{1,kn}(A,x) - \frac{m}{2n+1}
	\end{align*}
	We now distinguish two cases to show that $\density_{k,n}(A)\ge
	\tfrac{1}{5m} \density_1(A)$ for every $n\ge 1$. 
	Let $\density_1(A)=\varepsilon$. 
	\begin{itemize}
		\item If $n>\tfrac{2m}{\varepsilon}$, then $2n+1>\tfrac{4m}{\varepsilon}$  and thus $\tfrac{m}{2n+1}<\tfrac{\varepsilon}{4}$. 
		This implies
		\begin{multline*} 
			\density_{k,n}(A,x) \ge \tfrac{1}{2}\density_{1,kn}(A,x)-\frac{m}{2n+1} \\
			> \frac{\varepsilon}{2}-\frac{\varepsilon}{4} = \frac{\varepsilon}{4} \ge \tfrac{1}{5m}D_1(A).
		\end{multline*}
		\item If $n\le\tfrac{2m}{\varepsilon}$, then $2n+1\le \tfrac{4m}{\varepsilon}+1\le \tfrac{5m}{\varepsilon}$ and thus
		\[ \density_{k,n}(A,x)\ge \frac{1}{2n+1}\ge \frac{\varepsilon}{5m}=\tfrac{1}{5m} \density_1(A). \]
	\end{itemize}
	Thus, for every $n\ge 1$, we have $\density_{k,n}(A)\ge \tfrac{1}{5m}\density_1(A)$. This implies $\density_k(A)\ge \tfrac{1}{5m}\density_1(A)$.
\end{proof}

\subsection{Proof of Lemma \ref{dense-residue-classes}}
\denseResidueClasses*
\begin{proof}
	Suppose $S$ has positive local density, i.e.\ there is an $\varepsilon>0$ such that $d_k(S[\vec{t}],x)\ge \varepsilon$ for every $k\ge 1$, $\vec{t}\in\ZZ^p$, and $x\in\ZZ$. 
	Moreover, let $\vec{b}\in [0,B]^{p+1}$. 
	Write $\vec{b} = \tuple{\vec{s}, y}$, with $\vec{s}\in[0,B-1]^p$ and $y\in[0,B-1]$. 
	Now for every $\vec{t}\in\ZZ^p$ and $x\in\ZZ$, we have
	\begin{align*}
		d_k([S]_{B,\vec{b}}[\vec{t}],x)&=d_{kB}(S[B\cdot\vec{t}+\vec{s}],\vec{b}+B\cdot x)\ge\varepsilon,
	\end{align*}
	thus, $[S]_{B,\vec{b}}$ has positive local density with the same $\varepsilon$.

	Conversely, suppose $[S]_{B,\vec{b}}$ has positive local density for every $\vec{b}\in[0,B]^{p+1}$. 
	Since there are only finitely many $\vec{b}\in[0,B]^{p+1}$, there is a single $\varepsilon>0$ such that $d_{k}([S]_{B,\vec{b}}[\vec{t}],x)\ge \varepsilon$ for every $k\ge 1$, $\vec{t}\in\ZZ^p$, $x\in\ZZ$, and $\vec{b}\in[0,B-1]^{p+1}$. 
	Now consider $k\ge 1$, $\vec{t}\in\ZZ^p$, and $x\in\ZZ$. 
	Then there exists a $\vec{b} = \tuple{\vec{s}, y} \in[0,B]^{p+1}$ such that $\vec{t}=\vec{s}+B\vec{r}$ with $\vec{r}\in \ZZ^p$, $x=y+B\cdot z$ for some $z\in\ZZ$. By \cref{density-multiply}, we have
	\begin{align*}
		d_k(S[\vec{t}],x)&=\tfrac{1}{2B}d_{kB}(S[\vec{t}],x)=\tfrac{1}{2B} d_k([S]_{B,\vec{b}}[\vec{r}],z)\ge \tfrac{\varepsilon}{2B}.
	\end{align*}
	Thus, $S$ has positive local density.
\end{proof}

\subsection{Proof of Lemma \ref{complexity-transfer}}
\complexityTransfer*
We prove \cref{complexity-transfer} using the following lemma.
Here, we use the notation
\begin{equation*}
	\Eff(\cA)=\{w\in\ZZ \mid \tuple{q_0,0}\stepsint\tuple{q_1,w} \},
\end{equation*}
where $\cA = \tuple{Q, T q_0, q_1}$, i.e.\ $q_0$ and $q_1$ are the initial and
final state of $\cA$, respectively.

\begin{lemma}\label{unwrap-modulo}
	Let $B\ge 1$ and $b\in[0,B-1]$ be constants. Given an automaton $\cA$,
	we can compute in logspace an automaton $\cA'$ such that
	$\Eff(\cA')=\{x\in\ZZ \mid Bx+b\in \Eff(\cA)\}$. Moreover, if $\cA$ has
	updates in $\NN$, then so does $\cA'$.
\end{lemma}

\begin{proof}
	The idea is to maintain in the control state the residue modulo $B$ of
	the value added so far. 
	If $Q$ is the set of states of $\cA$, then $\cA'$ has state set $Q\times[0,B-1]$. 
	For each transition $p\xrightarrow{x}q$, we write $x = By+r$ for some $r\in[0,B-1]$, and then add transitions:
	\begin{equation*}
		\tuple{p,i} \xrightarrow{y} \tuple{q,i+r}
	\end{equation*}
	for every $i\in[0,B-1]$ with $i+r\in[0,B-1]$. 
	For those $i\in[0,B-1]$ where $i+r\in[B,2B-2]$ (and this is the only other option), say $i+r = B+s$, we add a transition:
	\begin{equation*}
		\tuple{p,i} \xrightarrow{y+1} \tuple{q,s}.
	\end{equation*}
	Moreover, a state $\tuple{q,i}$ is final if and only if $q$ is final in $\cA$
	and $i=b$. 
	Then $\cA'$ is clearly as desired.
\end{proof}

We are now ready to prove \cref{complexity-transfer}.
\begin{proof}[Proof of \cref{complexity-transfer}]
	For the first statement, suppose we are given an automaton $\cA$ and
	$\vec{t}\in\ZZ^p$. 
	We write \hbox{$\vec{t}=B\cdot\vec{s}+\vec{r}$} for some
	$\vec{s}\in\ZZ^p$ and $\vec{r}\in[0,B-1]^p$. 
	Then, for each possible
	residue $b\in[0,B-1]$ of the final counter value, we use
	\cref{unwrap-modulo} to build $\cA'_b$ with 
	\begin{equation*}
		\Eff(\cA'_b) = \set{x\in\ZZ\mid Bx+b\in\Eff(\cA)}.
	\end{equation*}
	Then clearly, $\cA$ can reach $S$ if and only if for some \hbox{$b\in[0,B-1]$}, the pair $\tuple{\cA'_b, \vec{s}}$ is a positive instance of $\intreach([S]_{B,\vec{b}})$, where $\vec{b} = \tuple{\vec{r}, b}$.
	Thus, if $\intreach([S]_{B,\vec{b}})$ is in $\AC^1$ for each $\vec{b} \in [0,B-1]^{p+1}$, then so is $\intreach(S)$. 
	Moreover, by the guarantee of~\cref{unwrap-modulo}, our reduction preserves the property that all updates are natural numbers. Hence, we have the same reduction for $\natreach(\cdot)$. 

	For the second statement, it suffices to prove that $\intreach([S]_{B,\vec{b}})$ reduces to $\intreach(S)$. 
	This is easy: We write $\vec{b} = \tuple{\vec{s}, b}$ with $\vec{s}\in [0,B-1]^{p}$, and $b\in [0,B-1]$. 
	Then given an automaton $\cA$ and $\vec{t}\in \ZZ^p$, we modify the automaton by multiplying $B$ to every counter update, and attaching to each final state a transition that adds $b$ and leads to a fresh final state. 
	Let $\cA'$ be the resulting automaton. 
	The parameter vector for the new instance is $\vec{t}' = B\vec{t}+\vec{s}$. 
	Then clearly, the instance $\tuple{\cA, \vec{t}}$ is positive for $\intreach([S]_{B,\vec{b}})$ if and only if $\tuple{\cA', \vec{t}'}$ is positive for $\intreach(S)$. Note also that the reduction preserves the property of all updates being non-negative, yielding the same reduction for $\natreach(\cdot)$.
\end{proof}

\subsection{Proof of Theorem \ref{make-acyclic}}\label{app:make-acyclic}
In this subsection, we prove the following:
\makeAcyclic*
We will prove \cref{make-acyclic} with the help of the following proposition.
Let us introduce some notation for an automaton $\cA = \tuple{Q, T, q_0, q_1}$. 
For a $\ZZ$-weighted automaton $\cA$, a \emph{run} is a sequence of transitions.  
By $\Runs(\cA)$, we denote the set of all runs, in $\cA$, from $\config{q_0}{0}$ to a configuration $\config{q_1}{w}$, for some $w \in \ZZ$.
Moreover, $\Runs_{\ell}(\cA)$ is the set of all such runs that have length $\le \ell$.
The \emph{effect} of a run $\config{q_0}{0} \xrightarrow{*} \config{q_1}{w}$ is $w$ and for a set of runs $R$, we denote 
\begin{equation*}
	\Eff(R) = \set{ w \in \ZZ : \config{q_0}{0} \xrightarrow{*} \config{q_1}{w} \text{ is in } R}.
\end{equation*}
Equivalently to the definition of $\Eff(\cA)$ above, we have $\Eff(\cA) = \Eff(\Runs(\cA))$.
\begin{proposition}\label{simulate-acyclic}
Given a $\Z^k$-labeled automaton $\cA$ and a bound $\ell$ in binary, we can
compute in logspace an acyclic automaton $\cA'$ with
	$\Eff(\Runs_{\le\ell}(\cA))\subseteq\Eff(\cA')\subseteq\Eff(\cA)$.  Moreover, if $\cA$
is $\N^k$-labeled, then $\cA'$ is $\N^k$-labeled as well.
\end{proposition}

Before we prove \cref{simulate-acyclic}, let us see how we can prove \cref{make-acyclic} using \cref{simulate-acyclic}.
\begin{proof}[Proof of \cref{make-acyclic}]
	According to \cref{simulate-acyclic}, it suffices to show that there is
	an exponential bound $\ell\in\NN$ such that if there is a run reaching
	$S$, then there is one of length at most $\ell$. Indeed, given $\ell$,
	we can invoke \cref{simulate-acyclic} and construct $\cA'$. Then we can
	reach $S$ in $\cA'$ if and only if we can reach $S$ in $\cA$.

	To find the bound $\ell$, first note that since $S$ is
	Presburger-definable, it is also definable by some (fixed)
	quantifier-free Presburger formula. Let $m$ the number of transitions
	of $\cA$. Since for a finite automaton, there is a polynomial-sized
	existential Presburger formula for its Parikh image~\cite[Theorem
	1]{DBLP:conf/icalp/SeidlSMH04}, there is also a polynomial-sized
	existential Presburger formula $\varphi(x_1,\ldots,x_m)$ with a free
	variable for each transition of $\cA$, such that $\varphi$ is satisfied
	if and only if there is a run with those transition multiplicities that
	reaches a vector in $S$.  Since every satisfiable existential
	Presburger formula has a solution that is at most exponential in its
	size (this is, in turn a consequence of solution sizes of systems of
	integer linear inequalities~\cite[Theorem]{vonZurGathenSieveking1978}),
	if $\cA$ can reach $S$, then $\varphi$ must have a solution, and this
	solution must be at most exponential in size. This solution, in turn,
	yields an at-most exponentially long run of $\cA$. We thus obtain our
	bound $\ell$.
\end{proof}

The remainder of this subsection is devoted to proving \cref{simulate-acyclic}. 
Let us introduce some notation. For an automaton $\cA$ and a state $q$ of
$\cA$, we denote by $\Cycles_q(\cA)$ the set of $q$-cycles. Here, a
\emph{$q$-cycle} is a sequence of transitions starting and ending in $q$.
Moreover, for a number $\ell\in\N$, $\Cycles_{q,\ell}(\cA)$ denotes the set of $q$-cycles of length at most
$\ell$. Furthermore, for a set $K\subseteq E^*$ of transitions, $K^*$ denotes the
usual Kleene closure, i.e.\ the set of concatenations of sequences in $K$.
Similarly, $K^{\le \ell}$ denotes the set of concatenations of at most $\ell$
sequences from $K$. In particular, $\Cycles_{q,n}(\cA)^{\le \ell}$ is the set of transition sequences that can be written as a concatenation of at most $\ell$-many $q$-cycles of length $n$ of $\cA$.
\begin{lemma}\label{moving-cycles}
	Let $\cA$ be a counter automaton with $n$ states. For every path $\pi$
	in $\cA$ there is a path $\pi'$ of the same effect such that $\pi'$ has
	the form $v_0 u_1 v_1 \cdots u_m v_m$, where 
	\begin{enumerate}
		\item $|v_0\cdots v_m|\le n^2$ and
		\item for each $i\in[1,m]$, there is a state $q_i$ such that $u_i$ belongs to $\Cycles_{q_i,n}(\cA)^*$.
	\end{enumerate}
\end{lemma}
\begin{proof}
	Let $\pi$ be a path of $\cA$. For each state that appears in $\pi$, we
	pick some occurrence of a transition that visits it and color it red;
	all the transition occurrences are colored black. Note that at most $n$
	positions in $\pi$ are red. If there are $n$ consecutive black
	transitions, then these must contain a cycle. While we can still find
	$n$ consecutive black transitions, we remove such cycles. After this
	removal, we end up with a run where between any two red transitions
	(and before the first; and after the last), there are less than $n$
	transitions. Thus, this run has length $\le (n+1)(n-1)=n^2-1$.
	Moreover, this run still visits all states that $\pi$ visits. We can
	thus re-insert all removed cycles and an appropriate place, and obtain
	a run as desired.
\end{proof}

\begin{lemma}\label{simulate-cycles}
	Given a $k$-counter automaton $\cA$ with $n$ states, a state $q$ of $\cA$, and a number $\ell\in\N$ in binary, 
	one can construct in logspace an acyclic automaton $\cA'$ such that 
	\[ \Eff(\Cycles_{q,n}(\cA)^{\le\ell})\subseteq\Eff(\cA')\subseteq\Eff(\Cycles_{q,n}(\cA)^*). \]
	Moreover, if $\cA$ is $\NN^k$-labeled, then $\cA'$ is as well.
\end{lemma}
\begin{proof}
	We begin with some notation. For any set $A\subseteq \ZZ^k$, let
	$A^*\subseteq\ZZ^k$ be the submonoid generated by $A$. Moreover, for
	$F\subseteq\ZZ$, define $F\cdot A:=\{f\cdot\vec{a}\mid f\in
	F,~\vec{a}\in A\}$. Moreover, $A^{r}$ is the $r$-fold Minkowski sum of $A$, meaning the set of vectors that
	can be written as a sum of $r$ vectors from $A$. In other
	words, we have
	\[ A^{r}=\{\vec{a}_1+\cdots+\vec{a}_r \mid \vec{a}_1,\ldots,\vec{a}_r\in A\}. \]
	Furthermore, we define $A^{\le r}=\bigcup_{i\in[0,r]} A^r$.
	For example, \hbox{$(\NN\cdot A)^{\le r}$} is the set of all vectors that can be written as a positive integer linear combination of at most $r$ distinct vectors from $A$. 

	Now let $C=\Eff(\Cycles_{q,n}(\cA))\subseteq\ZZ^k$ be the set of
	effects of cycles in $\cA$ of length $\le n$.  With this notation,
	our goal is to construct in logspace an acyclic automaton $\cA'$ such
	that
	\[ C^{\le\ell}\subseteq\Eff(\cA')\subseteq C^*. \]

	A result of Eisenbrand \& Shmonin~\cite[Theorem 1]{EisenbrandS06} shows that if
	$A$ is a finite set and $M$ is the maximal absolute value appearing in
	vectors in $A$, then $A^*=(\N\cdot A)^{\le r}$, where \hbox{$r=2k\log(4kM)$}.
	Since all the absolute values appearing in $C$ are at most exponential,
	this implies that Eisenbrand and Shmonin's result yields a polynomial
	bound $r$ with \hbox{$C^* = (\NN\cdot C)^{\le r}$}. 

	Note that each vector in $C^{\le\ell}$ has polynomial bit length. As
	observed above, we have in particular $C^{\le\ell}\subseteq(\NN\cdot
	C)^{\le r}$. This means, for each $\vec{c}\in C^{\le\ell}$, there are vectors
	$\vec{c}_1,\ldots,\vec{c}_t\in C$ with $t\le r$ and coefficients
	$\lambda_1,\ldots,\lambda_t\in\N$ with
	$\vec{c}=\lambda_1\vec{c}_1+\cdots+\lambda_t\vec{c}_t$. By the standard
	solution size bound for integer linear
	programming (e.g.\ \cite[Theorem]{vonZurGathenSieveking1978}), since the
	vectors $\vec{c},\vec{c}_1,\ldots,\vec{c}_t$ have polynomial bit
	length, and $t\le r$ is polynomial, the equation
	$\vec{c}=\lambda_1\vec{c}_1+\cdots+\lambda_t\vec{c}_t$ even has a
	solution where all the $\lambda_i$ have polynomial bit-length. This
	implies that there is an exponential bound $m\in\NN$ such that every
	vector from $C^{\le\ell}$ belongs to $([0,m]\cdot C)^{\le r}$. In
	summary, we observe
	\[ C^{\le\ell}\subseteq ([0,m]\cdot C)^{\le r}\subseteq C^*. \]
	Now the difficulty is that it is not clear how to construct a small
	automaton whose effects are exactly those in $[0,m]\cdot C$: Such an
	automaton would have to pick a factor $f\in[0,m]$ (for which there are exponentially many choices), and pick a cycle in
	$C$, and then multiply $f$ with each transition weight on this cycle. 
	Here, our key observation is that we can generate a \emph{larger set}
	consisting of multiples of cycles just by powers of two, where we allow
	to use different cycles for each power of two. This is still correct,
	since this larger effect set is still included in $C^*$. To make this precise,
	observe that
	\begin{align*}
		C^{\le\ell}\subseteq ([0,m]\cdot C)^{\le r}\subseteq \underbrace{\left(\sum_{j=0}^{\lceil \log m\rceil} 2^j\cdot (C\cup\{0\})\right)^r}_{=:B}\subseteq C^*.
	\end{align*}
	Therefore, it suffices to construct an acyclic automaton $\cA'$ with $\Eff(\cA')=B$. This is easy to do in deterministic logspace. Indeed, by definition of $C$, it is easy to construct an acyclic automaton for $C\cup\{0\}$, and thus also for $2^j\cdot (C\cup\{0\})$, by multiplying $2^j$ to each transition effect. For the sum term, we then chain together $\lceil\log m\rceil$-many (which is a polynomial quantity) of such automata together. Finally, for the $r$-fold Minkowski sum, we just chain $r$ copies of this automaton together.
\end{proof}

We are now ready to prove \cref{simulate-acyclic}.
\begin{proof}[Proof of \cref{simulate-acyclic}]
	We first construct an acyclic automaton $\cA''$ that simulates all runs of
	$\cA$ of length $\le n^2$: Its states are of the form $\tuple{q, i}$, where
	$q$ is a state of $\cA$ and $i\in[0,n^2]$. 
	Moreover, each transition
	increases the second component of the state, making the automaton
	acyclic. 

	In the second step, for each state $q$ of $\cA$, we use \cref{simulate-cycles} to construct an acyclic automaton $\cA_q$ so that $\Eff(\Cycles_{q,n}(\cA)^{\le\ell})\subseteq\Eff(\cA_q)\subseteq\Eff(\Cycles_{q,n}(\cA)^*)$. 

	We now combine $\cA''$ with the automata $\cA_q$ to obtain $\cA'$: We
	replace each state $\tuple{q, i}$ by a copy of $\cA_q$: Every transition between
	$\tuple{p, j}$ and $\tuple{q,i}$ will induce one from each final state of $\tuple{p,j}$'s
	copy of $\cA_p$ to each initial state of $\tuple{q,i}$'s copy of $\cA_q$.

	Then clearly, $\cA'$ is acyclic and constructible in logspace.
	Moreover, by \cref{moving-cycles} and the construction of the $\cA_q$,
	we have $\Eff(\Runs_{\le
	\ell}(\cA))\subseteq\Eff(\cA')\subseteq\Eff(\cA)$.
\end{proof}

\subsection{Proof of Lemma \ref{intreach-hardness}}\label{app:intreach-hardness}

\intreachHardness*

Before proving~\cref{intreach-hardness}, we first provide a short sketch of the proof.
Based on \cref{compute-presburger}, there are constants \hbox{$k,\ell\in\N$} where $k < \ell$ such that for given $\delta\ge 1$, we can compute in logspace a number
$x\in\ZZ$ and parameters $\vec{t}\in\ZZ^p$ such that $x + k \in S[\vec{t}]$ , $S[\vec{t}]\cap [x-\delta,x-1]=\emptyset$, and \hbox{$S[\vec{t}]\cap [x+\ell,x+\ell+\delta]=\emptyset$}. 
Given a subset sum instance $\tuple{a_1, \ldots, a_n, b}$, we compute
such $x,\vec{t}$ for $\delta \coloneqq |\ell a_1|+\cdots+|\ell a_n|$.  
Then the modified instance $\tuple{\ell a_1,\ldots,\ell a_n,\ell b}$
is equivalent to the first. Moreover, any sum that is not exactly $\ell
b$ will be off by at least $\ell$, and by at most $\delta$. Thus, we
construct an automaton that produces the sum of a subset of $\ell
a_1,\ldots,\ell a_n$, and then adds $x+k-\ell b$. Then, the automaton
will either reach $x+k$ (if the instance was positive) or will be off
by at least $\ell$, and by at most $\delta$ (and thus not reach
$S[\vec{t}]$). %

\begin{proof}[Proof of~\cref{intreach-hardness}]
	We know that for each $\delta\ge 1$, there exist $\vec{t}\in\ZZ^p$ and
	$x\in\ZZ$ such that 
	\begin{enumerate}[(a)]
		\item $S[\vec{t}]\cap[x,x+\ell-1]\ne\emptyset$,
		\item $S[\vec{t}]\cap[x-\delta,x-1]=\emptyset$, and
		\item $S[\vec{t}]\cap[x+\ell,x+\ell+\delta]=\emptyset$. 
	\end{enumerate}
	This implies that
	there exists a $k\in[0,\ell-1]$ such that for every $\delta\ge 1$,
	there are $\vec{t}\in\ZZ^p$ and $x\in\ZZ$ such that
	\begin{enumerate}[(a')]
		\item $x+k\in S[\vec{t}]$,
		\item $S[\vec{t}]\cap[x-\delta,x-1]=\emptyset$, and
		\item $S[\vec{t}]\cap[x+\ell,x+\ell+\delta]=\emptyset$.
	\end{enumerate}
	We call such a pair $\tuple{\vec{t}, x}$ a \emph{hard pair} for $\delta$.

	Consider a Presburger-definable well-order $\preceq$ on $\ZZ^{p+1}$.
	For example, within each orthant, we use the lexigraphic ordering
	w.r.t.\ absolute values of the entries; and then we pick an arbitrary linear
	ordering on the $2^{p+1}$ orthants. Since we know that for each
	$\delta\ge 1$, there exists a hard pair, we can define a function
	$f\colon\NN\to\ZZ^p\times\ZZ$ in Presburger arithmetic such that for
	each $\delta\ge 1$, the pair $f(\delta)$ is a hard pair for $\delta$:
	The function picks the smallest hard pair w.r.t.\ the ordering
	$\preceq$. Since $f$ is Presburger-definable, it is computable in
	logarithmic space by \cref{compute-presburger}.

	Consider a subset sum instance $\tuple{a_1,\ldots,a_n,b}$, set $\delta = |\ell a_1|+\cdots+|\ell a_n|$, and compute a hard pair $f(\delta) = \tuple{\vec{t}, x}$ for $\delta$.  Consider the automaton in \cref{fig:vas-np-hard},
	and use it with $x_i:=\ell a_i$ for $i\in[1,n]$, $y:=\ell b$, and
	$v:=x+k$. Moreover, use the vector $\vec{t} \in \ZZ^p$ as the parameter vector
	for our instance of $\intreach(S)$. We claim that the automaton can
	reach $S$ if and only if the instance $\tuple{a_1,\ldots,a_n,b}$ is positive.

	If $\tuple{a_1,\ldots,a_n,b}$ is a positive instance, then this automaton can
	clearly reach $x+k$. 
	Conversely, if $\tuple{a_1,\ldots,a_n,b}$ is a negative
	instance of subset sum, then any sum of a subset of $\{\ell
	a_1,\ldots,\ell a_n\}$ will be either $\le \ell b-\ell$ or $\ge \ell
	b+\ell$. Therefore, any number reached in the final state of our
	automaton will belong to $[x+k-\delta,x+k-\ell]$ or to
	$[x+k+\ell,x+k+\delta]$. However, both of these sets have an empty
	intersection with $S[\vec{t}]$, by the choice of $\vec{t}$ and $x$.
	Hence, the instance of $\intreach(S)$ is also negative.
\end{proof}

\subsection{Proofs of \ref{equivalence-positive-local-density} to \ref{equivalence-no-unbounded-isolation} and \ref{equivalence-finite-union} to \ref{equivalence-positive-local-density} of Prop.~\ref{integer-dichotomy-equivalence}}
\label{app:equivalence-finite-union-implications}
\begin{claim}[\impl{equivalence-positive-local-density}{equivalence-no-unbounded-isolation} of~\cref{integer-dichotomy-equivalence}]\label{clm:equivalence-1implies2}
	If $S$ has positive local density, then $S$ does \emph{not} have unbounded isolation.
\end{claim}
\begin{proof}
	Let us assume that $S$ has positive local density $\density(S) > 0$.
	For the sake of contradiction, suppose that $S$ has unbounded isolation.
	This means that there exists $\ell\ge 1$ such that for every $\delta\ge 1$, there exists a $\vec{t}\in\ZZ^{p}$ and an $x\in\ZZ$ such that:
	\begin{enumerate}
		\item $S[\vec{t}]\cap [x,x+\ell-1]\ne\emptyset$,
		\item $S[\vec{t}]\cap [x-\delta,x-1]=\emptyset$, and
		\item $S[\vec{t}]\cap [x+\ell,x+\ell+\delta]=\emptyset$.
	\end{enumerate}
	
	Let $\delta = \frac{1}{2}\left(\frac{2\ell}{\density(S)} - 1\right)$.
	Consider the $\vec{t} \in \ZZ^p$ and $x \in \ZZ$ for which the unbounded isolation conditions hold (for this $\delta$).
	The first condition implies the existence of $y \in [x, x+\ell]$ such that $y \in S[\vec{t}]$.
	Given the fact that $y \geq x$, the second condition implies 
	\begin{equation*}
		S[\vec{t}] \cap [y-\delta,x-1] \sset S[\vec{t}] \cap [x-\delta,x-1] = \emptyset.
	\end{equation*}
	Similarly, given the fact that $y \leq x+\ell$, the third condition implies
	\begin{equation*}
		S[\vec{t}] \cap [x+\ell,y+\delta] \sset S[\vec{t}] \cap [x+\ell,x+\ell+\delta] = \emptyset.
	\end{equation*}
	We therefore know that $S[\vec{t}] \cap [y-\delta,y+\delta] \sset S[\vec{t}] \cap [x, x+\ell-1] \sset [x, x+\ell-1]$, so 
	\begin{equation}\label{eq:sparse-interval}
		| S[\vec{t}] \cap [y-\delta,y+\delta] | \leq \ell.
	\end{equation}

	Now, by definition of local density, we know that $\density(S[\vec{t}]) \geq D(S)$.
	Precisely, we know that 
	\begin{equation*}
		\density(S[\vec{t}]) = \inf_{k\in\NN} \inf_{x' \in S[\vec{t}]} \density_k(S[\vec{t}],x').
	\end{equation*}
	By considering $k = 1$ and $x' = x$, we deduce that $\density_1(S[\vec{t}],x) \geq D(S)$.
	However,
	\begin{align*}
		\density_1(S[\vec{t}],x) 
			& = \inf_{n\ge 1} \frac{|S[\vec{t}] \cap [x-n,x+n] \cap (x+\ZZ)|}{2n+1} \\
			& = \inf_{n\ge 1} \frac{|S[\vec{t}] \cap [x-n,x+n]|}{2n+1} \\
			& \leq \frac{|S[\vec{t}] \cap [x-\delta,x+\delta]|}{2\delta+1} \\
			& \leq \frac{\ell}{2\delta+1} && \text{(\cref{eq:sparse-interval})} \\
			& = \frac{\ell}{2 \cdot \frac{1}{2}\left(\frac{2\ell}{\density(S)} - 1\right) +1} 
			  = \frac{\density(s)}{2}.
	\end{align*}	
	Since $D(S) > 0$, this contradicts the fact that $\density_1(S[\vec{t}],x) \geq D(S)$.
	We therefore conclude that $S$ indeed must not have unbounded isolation.
\end{proof}

\begin{claim}[\impl{equivalence-finite-union}{equivalence-positive-local-density} of~\cref{integer-dichotomy-equivalence}]\label{clm:equivalence-3implies1}
	If $S$ is a finite union of transformed building blocks, then $S$ has positive local denisty.
\end{claim}
\begin{proof}
	Indeed, if $S$ is the empty set, then it is a finite union of transformed building blocks (the empty set is a valid tranformed building block) and $S$ has positive local density.
	Let us assume that $S \neq \emptyset$ and $S$ is a finite union of transformed building blocks.
	First, observe that for an arbitrary sets $A_1, A_2 \sset \ZZ$, $\density(A_1 \cup A_2) \geq \min\set{\density(A_1), \density(A_2)}$.
	Thus, given that $S$ is a finite union of transformed building blocks, it suffices to prove that an arbitrary transformed building block has positive local density.

	Let $S'$ be a non-empty transformed building block; there exists $\rho > 0$, $m \geq 1$, and Presburger definable functions $\tau : \ZZ^p \to \ZZ^{2m+1}$ and $\sigma : \ZZ^p \to \set{0, 1}$ such that for every $\vec{t} \in \ZZ^p$, we have 
	\begin{equation*}
		S'[\vec{t}] = (-1)^{\sigma(\vec{t})} \cdot S'^{(\rho, m)}[\tau(\vec{t})].
	\end{equation*}
	Recall that
	\begin{equation*}
		\density(S') = \inf_{\vec{t} \in \ZZ^p} \density(S[\vec{t}]).
	\end{equation*}
	Given than $S' \neq \emptyset$, we know that there exists at least one $\vec{t} \in \ZZ^p$ such that $S'[\vec{t}] \neq \emptyset$ and recall that $\density(\emptyset) = 1$, thus
	\begin{equation}\label{eq:nonempty-parameters}
		\density(S') = \inf_{\vec{t} \in \ZZ^p : S'[\vec{t}] \neq \emptyset} \density(S'[\vec{t}]).
	\end{equation}
	Furthermore, given that $\sigma(\vec{t})$ is only used to (potentially) flip the sign, for every $\vec{t} \in \ZZ^p$, 
	\begin{equation*}
		\density(S'[\vec{t}]) = \density(S'^{(\rho, m)}[\tau(\vec{t})]).
	\end{equation*}
	Given~\cref{eq:nonempty-parameters}, if there exists some constant $\varepsilon > 0$ such that for all $\vec{t} \in \ZZ^p,\linebreak S'[\vec{t}] \neq \emptyset \implies \density(S'^{(\rho, m)}[\tau(\vec{t})]) \geq \varepsilon$, then $\density(S') \geq \varepsilon > 0$.

	Accordingly, let $\vec{t} \in \ZZ^p$ be an instantiation of the parameters such that $S[\vec{t}] \neq \emptyset$.
	We will prove that $\density(S'^{(\rho, m)}[\tau(\vec{t})]) \geq \frac{1}{15(m+1)(\rho+1)}$ (and clearly $\frac{1}{15(m+1)(\rho+1)} > 0)$.
	Since $S'^{(\rho, m)}[\tau(\vec{t})] \neq \emptyset$, we know that $\tau(\vec{t})$ is $\rho$-admissible.
	Suppose $\tau(\vec{t}) = (s_1, t_1, \ldots, s_m, t_m, s_{m+1})$.
	Let $I_i = [s_i, t_i]$ for all $i \in [1,m]$ and $I_{m+1} = [s_{m+1}, \infty)$.
	We know that $I_1, \ldots, I_{m+1}$ satisfy Conditions~\ref{itm:disjoint}--\ref{itm:december} of~\cref{def:rho-admissible}.

	We will now prove that $\density_1(S'^{(\rho, m)}[\tau(\vec{t})]) \geq \frac{1}{3(\rho+1)}$ by arguing that in fact, for every $x \in S'^{(\rho, m)}[\tau(\vec{t})]$, $\density_1(S'^{(\rho, m)}[\tau(\vec{t})], x) \geq \frac{1}{3(\rho+1)}$.
	Recall that 
	\begin{multline*}
		\density_1(S'^{(\rho, m)}[\tau(\vec{t})], x) = \inf_{n \in \NN} \density_{1,n}(S'^{(\rho, m)}[\tau(\vec{t})], x) \\ 
		= \inf_{n \in \NN} \frac{|S'^{(\rho, m)}[\tau(\vec{t})] \cap [x-n, x+n]|}{2n+1}.
	\end{multline*}

	Before analysing the value of $\density_{1,n}(S'^{(\rho, m)}[\tau(\vec{t})], x)$, we will introduce some helpful notation.
	For every $j \in [i, m]$, let $d_j$ be the largest distance between $x$ and any point in $I_i \cup \cdots \cup I_j$.
	Precisely, 
	\begin{equation*}
		d_j \coloneqq \max\set{d(x, y) : y \in I_i \cup \cdots \cup I_j}.
	\end{equation*}
	We will also define $d_{m+1} \coloneqq \infty$.
	Note that the point $y \in I_i \cup \cdots \cup I_j$ such that $d(x, y)$ is maximised as an endpoint of one of the intervals $I_i, \ldots, I_j$.
	For all $j \in [1, m+1]$, we will also define $e_j$ to be the maximum distance between $x$ and any one of the intervals $I_i, \ldots, I_j$.
	Precisely, 
	\begin{equation*}
		e_j \coloneqq \max\set{d(x, I_i), \ldots, d(x, I_j)}.
	\end{equation*}
	Similar to before, note that the point $x + e_j$ is an endpoint of one of the intervals $I_i, \ldots, I_j$.
	The following two basic facts will be useful later in this proof.
	Firstly, simply by definition of $d_j$,
	\begin{equation}\label{eq:contains}
		[x-d_j, x+d_j] \supseteq I_i \cup \cdots \cup I_j.
	\end{equation}
	Secondly, given Condition~\ref{itm:distances} of~\cref{def:rho-admissible},
	\begin{align}\label{eq:interval-distance}
		e_j 
		& \leq |I_i| + d(I_i, I_{i+1}) + \ldots + |I_{j-1}| + d(I_{j-1}, I_j)  \nonumber \\
		& \leq |I_i| + \rho\cdot|I_i| + \ldots + |I_{j-1}| + \rho\cdot|I_{j-1}| \\
		& = (\rho+1)\cdot(|I_i| + \ldots + |I_{j-1}|). \nonumber
	\end{align}

	Now, we shall conside the value of $\density_{1,n}(S'^{(\rho, m)}[\tau(\vec{t})], x)$ for various values of $n$.
	First, suppose that $n \leq d_i$.
	Observe that $[x-n, x]$ or $[x, x+n]$ is contained in $I_i$.
	Without loss of generality, suppose that $[x-n, x] \sset I_i$; we conclude that
	\begin{align*}
		\density_{1,n}(S'^{(\rho, m)}[\tau(\vec{t})], x) 
		& \geq \frac{|I_i \cap [x-n, x+n]|}{2n+1} \\
		& \geq \frac{|[x-n, x]|}{2n+1} = \frac{n+1}{2n+1} \geq \tfrac{1}{2}.
	\end{align*}

	Next, for any $j \in [i+1, m+1]$ suppose $n > d_{j-1}$ and $n \leq e_j$.
	Since $n > d_{j-1}$ and given~\cref{eq:contains}, we know that $[x-n, x+n] \supseteq I_i \cup \cdots \cup I_{j-1}$.
	Therefore, 
	\begin{align*}
		\density_{1,n}&(S'^{(\rho, m)}[\tau(\vec{t})], x) \\
		& \geq \frac{|(I_i \cup \cdots \cup I_{j-1}) \cap [x-n, x+n]|}{2n+1} \\
		& \geq \frac{|(I_i \cup \cdots \cup I_{j-1}) \cap [x-n, x+n]|}{2(\rho+1)\cdot(|I_i| + \ldots + |I_{j-1}|) + 1} \hspace{0.2in}\text{(By~\cref{eq:interval-distance})} \\
		& \geq \frac{|I_i| + \ldots + |I_{j-1}|}{2(\rho+1)\cdot(|I_i| + \ldots + |I_{j-1}|) + 1} \hspace{0.2in}\text{(By~\ref{itm:disjoint})} \\
		& \geq \frac{|I_i| + \ldots + |I_{j-1}|}{3(\rho+1)\cdot(|I_i| + \ldots + |I_{j-1}|)} = \frac{1}{3(\rho+1)}.
	\end{align*}

	Now, for any $j \in [i+1, m+1]$ suppose $n > e_j$ and $n \leq d_j$.
	Given that $n \leq d_j$ and by considering~\cref{eq:interval-distance}, we know that $n \leq (\rho+1)\cdot(I_i + \ldots I_{j-1}) + |I_j|$.
	Given that $n > e_j$, we at least know that $[x-n, x+n] \supseteq I_i \cup \cdots \cup I_{j-1}$.
	It could also be the case that (i) $I_j \sset [x-n, x+n]$, otherwise if $I_j \not\sset [x-n, x+n]$, then it must be true that (ii) $[x-n, x-e_j] \sset I_j$ or $[x+e_j, x+n] \sset I_j$.
	In case (i), we know that 
	\begin{align*}
		\density_{1,n}&(S'^{(\rho, m)}[\tau(\vec{t})], x) \\
		& \geq \frac{|(I_i \cup \cdots \cup I_{j-1} \cup I_j) \cap [x-n, x+n]|}{2n+1} \\
		& \geq \frac{|I_i| + \ldots + |I_{j-1}| + |I_j|}{2((\rho+1)\cdot(I_i + \ldots I_{j-1}) + |I_j|)+1} \hspace{0.2in}\text{(By~\ref{itm:disjoint})} \\
		& \geq \frac{|I_i| + \ldots + |I_{j-1}| + |I_j|}{3((\rho+1)\cdot(I_i + \ldots I_{j-1}) + |I_j|)} \\
		& \geq \frac{|I_i| + \ldots + |I_{j-1}| + |I_j|}{3((\rho+1)\cdot(I_i + \ldots I_{j-1}) + (\rho+1)\cdot|I_j|)} \\
		& = \frac{1}{3(\rho+1)}.
	\end{align*}
	In case (ii), without loss of generality, let us assume that $[x-n, x-e_j] \sset I_j$.
	We therefore conclude that,
	\begin{align*}
		\density_{1,n}&(S'^{(\rho, m)}[\tau(\vec{t})], x) \\
		& \geq \frac{|(I_i \cup \cdots \cup I_{j-1} \cup I_j) \cap [x-n, x+n]|}{2n+1} \\
		& = \frac{|(I_i \cup \cdots \cup I_{j-1} \cup I_j) \cap [x-n, x+n]|}{2(e_j + (n-e_j))+1} \\
		& \geq \frac{|I_i| + \ldots + |I_{j-1}| + |[x-n, x-e_j]|}{2(e_j + (n-e_j))+1} \hspace{0.2in}\text{(By~\ref{itm:disjoint})} \\
		& = \frac{|I_i| + \ldots + |I_{j-1}| + n - e_j + 1}{2(e_j + (n-e_j))+1} \\
		& \geq \frac{|I_i| + \ldots + |I_{j-1}| + n - e_j + 1}{2(\rho+1)\cdot(|I_i| + \ldots + |I_{j-1}|) + 2(n-e_j) + 1} \\
		& \geq \frac{|I_i| + \ldots + |I_{j-1}| + n - e_j}{2(\rho+1)\cdot(|I_i| + \ldots + |I_{j-1}|) + 2(n-e_j)} \\
		& \geq \frac{|I_i| + \ldots + |I_{j-1}| + n - e_j}{2(\rho+1)\cdot(|I_i| + \ldots + |I_{j-1}|) + 2(\rho+1)\cdot(n-e_j)} \\
		& = \frac{1}{2(\rho+1)}.
	\end{align*}

	Given that $x \in S'^{(\rho, m)}[\tau(\vec{t})]$ was selected arbitrarily and since $\NN = [0, d_i] \cup (d_i, e_{i+1}] \cup (e_{i+1}, d_{i+1}] \cup \cdots \cup (e_{m+1}, \infty)$, 
	\begin{align*}
		D_1(S'^{(\rho,m)}[\tau(\vec{t})]) 
		& = \inf_{x \in S'^{(\rho, m)}[\tau(\vec{t})]}\inf_{n \in \NN} D_{1, n}(S'^{(\rho,m)}[\tau(\vec{t})], x) \\
		& \geq \frac{1}{3(\rho+1)}.
	\end{align*}

	Therefore, given that $S'^{(\rho, m)}[\tau(\vec{t})]$ is a union of $m+1$ intervals, by~\cref{lem:interval-local-density} we know that for all $k \geq 2$, 
	\begin{align*}
	\density_k(S'^{(\rho, m)}[\tau(\vec{t})]) 
	& \geq \frac{1}{5(m+1)}\density_1(S'^{(\rho, m)}[\tau(\vec{t})]) \\
	& \geq  \frac{1}{15(m+1)(\rho+1)}.
	\end{align*}

	Finally, given~\cref{eq:nonempty-parameters} and since
	\begin{align*}
		\density(S'[t]) = \density(S'^{(\rho, m)}[\tau(\vec{t})]) 
		& = \inf_{k \in \NN} \density_k(S'^{(\rho, m)}[\tau(\vec{t})]) \\
		& \geq \frac{1}{15(m+1)(\rho+1)},
	\end{align*}
	we conclude that $D(S') \geq \frac{1}{15(m+1)(\rho+1)} > 0$.
\end{proof}

\subsection{Proof of Proposition \ref{pro:interval-uniform-decomposition}}\label{app:interval-uniform-decomposition}

\intervalUniformDecomposition*
\begin{proof}
	By \cref{lem:intervals}, there are functions
$\alpha_1,\beta_1,\ldots,\alpha_m,\beta_m\colon\ZZ^p\to\ZZ\cup\{-\infty,\infty\}$
such that for every $\vec{t}\in\ZZ^p$, we have
$\alpha_1(\vec{t})\le \beta_1(\vec{t})\le\alpha_2(\vec{t})\le\beta_2(\vec{t})\le\cdots\le\alpha_m(\vec{t})\le\beta_m(\vec{t})$
and
$S[\vec{t}]=(\alpha_1(\vec{t}),\beta_1(\vec{t}))\cup\cdots\cup(\alpha_m(\vec{t}),\beta_m(\vec{t}))$.
This implies that for every $\vec{t}\in\ZZ^p$, the set $S[\vec{t}]$ has a canonical decomposition of at most $m$ intervals. Moreover, a set with a canonical decomposition with $k\le $ intervals has one of at most $4^k$ different types. Let $T=\{\tau \mid \text{$\tau$ is a type of length $\le m$}\}$. Then $T$ is finite, and for each $\tau\in T$, the set
\[ W_\tau =\{\vec{t}\in\ZZ^p \mid \text{$S[\vec{t}]$ has type $\tau$}\} \]
is definable in Presburger arithmetic. This can easily be done by employing
Presburger formulas for the functions
$\alpha_1,\beta_1,\ldots,\alpha_m,\beta_m$. Note that a formula for $W_\tau$
will have to consider all possible cases of intervals
$(\alpha_i(\vec{t}),\beta_i(\vec{t}))$ being empty.
But those intervals that are non-empty must have pairwise distance $\ge 2$ and
thus form separate constituents of the canonical decomposition.  Hence, we only
need to distinguish which intervals are empty (i.e. $2^m$ cases), and in each
case, one can verify that the type of $S[\vec{t}]$ is $\tau$.

Now for each $\tau\in T$, consider the set $S_\tau\subseteq\ZZ^{p+1}$ with
\[ S_\tau[\vec{t}]=\begin{cases} S[\vec{t}] & \text{if $t\in W_{\tau}$} \\ \emptyset & \text{otherwise}\end{cases} \]
Then $S_\tau$ is Presburger-definable, since $W_\tau$ is. Observe that each
$S_\tau$ is interval-uniform: For every $\vec{t}\in\ZZ^p$, the set
$S_\tau[\vec{t}]$ is either empty or of type $\tau$.  Then, we clearly have
$S=\bigoplus_{\tau\in T} S_\tau$, as desired.
\end{proof}

\subsection{Proof of Proposition \ref{no-isolation-no-constellation}}\label{app:no-isolation-no-constellation}

In this subsection, we prove the following:
\noIsolationNoConstellation*

We begin with a simple lemma about semilinear sets.
\begin{lemma}\label{ratio-lemma}
	Let $A\subseteq\NN^2$ be a semilinear set where $\vec{u}[2]>0$ for every $\vec{u}\in A$. The following are equivalent.
	\begin{enumerate}
		\item\label{cond:unbounded-ratio} For every $k\ge 1$, there is a $\vec{u}\in A$ with $\vec{u}[1]\ge k\cdot\vec{u}[2]$.
		\item\label{cond:unbounded-fixed} There is $\ell\ge 1$ such that, for every $\delta\ge 1$, there is a $\vec{u}\in A$ with $\vec{u}[2]\le \ell$ and $\vec{u}[1]\ge \delta$.
	\end{enumerate}
\end{lemma}
\begin{proof}
	The implication \impl{cond:unbounded-fixed}{cond:unbounded-ratio} is
	trivial, so suppose we have \cref{cond:unbounded-ratio} and consider a
	semilinear representation for $A$. 
	Clearly, if~\cref{cond:unbounded-ratio} holds for $A$, it must hold for some linear
	set $L \sset A$ in the semilinear representation of $A$. 
	Let $\vec{b}\in\NN^2$ be the base vector of $L$.
	Since all vectors in $A$ have non-zero second
	component, we know that $\vec{b}[2]>0$. 
	\cref{cond:unbounded-ratio} implies that there must be a period vector $\vec{p}\in\NN^2$ such that $\vec{p}[1]>0$ and $\vec{p}[2]=0$. 
	Indeed, if every period vector had $\vec{p}[2]>0$, then in all elements $\vec{u} \in L$,
	would have $\vec{u}[1] \le r \cdot \vec{u}[2]$, where
	\begin{equation*}
		r = \max\left\{\frac{\vec{b}[1]}{\vec{b}[2]}, \frac{\vec{p}[1]}{\vec{p}[2]} : \vec{p} \text{ is period vector of } L \right\},
	\end{equation*}
	in contradiction to~\cref{cond:unbounded-ratio}. 
	Now, with $\ell=\vec{b}[2]$, we satisfy~\cref{cond:unbounded-fixed} by picking
	$\vec{u}=\vec{b}+\delta\cdot\vec{p}$.
\end{proof}

For interval-uniform sets, we have a variant of \cref{lem:intervals} that
guarantees each of the intervals $(\alpha_i(\vec{t}),\beta_i(\vec{t}))$ to be
non-empty. 
\begin{restatable}{lemma}{intervalsIntervalUniform}\label{lem:intervals-interval-uniform}
	If $T\subseteq\ZZ^{p+1}$ is interval-uniform, then 
	there exist Presburger-definable functions $\alpha_1, \beta_1, \ldots, \alpha_m, \beta_m : \ZZ^p \to \ZZ \cup \set{-\infty, \infty}$ such that, for every $\vec{t} \in \ZZ^p$, 
	$\alpha_i(\vec{t})+2 \le \beta_i(\vec{t})$ for $i\in[1,m]$ and also $\beta_{i}(\vec{t})\le\alpha_i(\vec{t})$ for $i\in[1,m-1]$ and either (i)~$T[\vec{t}]=\emptyset$ or (ii)~%
		$T[\vec{t}] = (\alpha_1(\vec{t}), \beta_1(\vec{t})) \cup \cdots \cup (\alpha_m(\vec{t}), \beta_m(\vec{t}))$.
\end{restatable}
\begin{proof}
	By interval-uniformity, we know that there is a type $\tau\in
\{1,-\infty,\infty,2\infty\}^m$ such that for every $\vec{t}\in\ZZ^p$, the set
$S[\vec{t}]$ is either empty or of type $\tau$.  Moreover, by
\cref{lem:intervals}, there are Presburger-definable functions
$\alpha'_1,\beta'_1\ldots,\alpha'_k,\beta'_k\colon\ZZ^p\to\ZZ\cup\{-\infty,\infty\}$
such that for each $\vec{t}$, we have
$\alpha'_1(\vec{t})\le\beta'_1(\vec{t})\le\cdots\le\alpha'_k(\vec{t})\le\beta'_k(\vec{t})$
with
$S[\vec{t}]=(\alpha'_1(\vec{t}),\beta'_1(\vec{t}))\cup\cdots\cup(\alpha'_k(\vec{t}),\beta'_k(\vec{t}))$.
In the latter decomposition, not every interval
$(\alpha'_i(\vec{t}),\beta'_i(\vec{t}))$ is non-empty. However, since $S$ is
interval-uniform, the canonical decomposition of $S[\vec{t}]$ has the same
number $m$ of intervals, provided $S[\vec{t}]\ne\emptyset$. We can therefore
consider, for each subset $J\subseteq[1,k]$, the set 
\begin{align*}
 W_J=\{\vec{t}\in\ZZ^p \mid S[\vec{t}]\ne\emptyset\text{ and for all $j\in[1,k]$, we have} \\
\text{$(\alpha'_j(\vec{t}),\beta'_j(\vec{t}))=\emptyset$~if and only if $j\in J$}\}. \end{align*}
Then $W_J$ is clearly Presburger-definable. We can therefore define
$\alpha_1,\beta_1,\ldots,\alpha_m,\beta_m\colon\ZZ^p\to\Z$ as follows. First,
they all check for which subset $J\subseteq[1,k]$, we have $\vec{t}\in W_J$.
This tells us which of the intervals $(\alpha'_i(\vec{t}),\beta'_i(\vec{t}))$
are empty. Note that all the non-empty intervals
$(\alpha'_i(\vec{t}),\beta'_i(\vec{t}))$ must pairwise have distance $\ge 2$,
so they form the canonical decomposition. Therefore, there must be exactly $m$
intervals $(\alpha'_i(\vec{t}),\beta'_i(\vec{t}))$ that are non-empty. The
functions $\alpha_1,\beta_1,\ldots,\alpha_m,\beta_m$ can thus point to those
non-empty intervals.
\end{proof}
\begin{proof}[Proof of~\cref{no-isolation-no-constellation}]
	Suppose in the ideal decomposition $D_1\cup\cdots\cup D_z$, there is an
	$\omega$-constellation. 
	Without loss of generality, suppose that $D_1$ has an $\omega$-constellation. 
	For each $\vec{t}\in\ZZ^p$, let
	$S[\vec{t}]=I_1[\vec{t}]\cup\cdots\cup I_m[\vec{t}]$ be the canonical
	decomposition of $S[\vec{t}]$, and let $\ZZ\setminus
	S[\vec{t}]=K_1[\vec{t}]\cup\cdots K_r[\vec{t}]$ be the canonical
	decomposition of the complement $\ZZ\setminus S[\vec{t}]$. 
	Since $S$ is interval-uniform, the values $m$ and $r$ do not depend on $\vec{t}$. 
	Moreover, by~\cref{lem:intervals-interval-uniform}, the
	interval endpoints of the $I_j[\vec{t}]$ and the $K_j[\vec{t}]$ can be
	computed by Presburger functions.

	Observe that the set $W\coloneqq\{\vec{t}\in\ZZ^p \mid
	M_{S[\vec{t}]}\le D_1\}$ is Presburger-definable, despite the fact that
	$M_{S[\vec{t}]}$ measures ratios. The entries in $D_1$ are either
	$\omega$ or constants, and so, for each $i$ and $j$, the set of $\vec{t} \in \ZZ^p$ such that $\left\lceil\tfrac{|K_i[\vec{t}]|}{|I_j[\vec{t}]|}\right\rceil\le (D_1)_{i,j}$ is Presburger-definable.
	Since $D_1$ has an $\omega$-constellation, we know
	that there are indices $i_1,i_2\in[1,m]$ such that for the indices
	$J=\{j\in[1,r] \mid K_{i_1}<I_j<K_{i_2}\}$, we have
	$(D_1)_{i_1,j} = (D_1)_{i_2,j} = \omega$ for each $j\in J$. 
	Moreover, by choosing
	$i_1,i_2$ with minimal difference, we can ensure that whenever
	$i_1 < i < i_2$, there exists $j\in J$ such that  $(D_1)_{i,j}\in\NN$. 

	Our goal is to apply \cref{ratio-lemma} to the set $A\subseteq\N^2$ of all pairs
	\begin{equation*}
		\left\langle \min\set{|K_{i_1}[\vec{t}]|, |K_{i_2}[\vec{t}]|},~\sum_{j\in J} |I_j[\vec{t}]|+\sum_{i=i_1+1}^{i_2-1} |K_{i}[\vec{t}]| \right\rangle
	\end{equation*}
	for $\vec{t}\in W$. 
	Thus, we claim that for every $k\in\NN$, there is a $\vec{t}\in W$ with 
	\begin{equation}\label{eq-unbounded-ratio} 
		\min\set{|K_{i_1}[\vec{t}]|,|K_{i_2}[\vec{t}]|} \ge k\cdot \left(\sum_{j\in J} |I_j[\vec{t}]|+\sum_{i=i_1+1}^{i_2-1} |K_i[\vec{t}]|\right). 
	\end{equation}
	Indeed, let $R\in\NN$ be the largest number occurring in $D_1$.
	Since $(D_1)_{i_1,j} = (D_1)_{i_2,j} = \omega$ for every $j\in J$, there must be a $\vec{t}\in W$ with $|K_{i_1}[\vec{t}]|, |K_{i_2}[\vec{t}]| \ge 2km^2R\cdot |I_j[\vec{t}]|$ for every $j\in J$. 
	Also, by the pigeonhole principle, there must be a specific $j'\in J$ with $\sum_{j\in J}|I_j[\vec{t}]|\le m \cdot |I_{j'}[\vec{t}]|$ (since $|J|\le m$). 
	With this, we have 
	\begin{equation}\label{eq-unbounded-ratio-a} 
		|K_{i_1}[\vec{t}]|, |K_{i_2}[\vec{t}]| 
		\ge 2km^2R \cdot |I_{j'}[\vec{t}]|
		\ge 2kmR \cdot \sum_{j\in J}|I_j[\vec{t}]|. 
	\end{equation}
	Furthermore, we know that $|K_i[\vec{t}]|\le R\cdot \sum_{j\in J}|I_j[\vec{t}]|$ for every $i\in[i_1+1,i_2-1]$ and thus 
	\begin{equation}\label{eq-unbounded-ratio-b}
		kmR\cdot \sum_{j\in J} |I_j[\vec{t}]|\ge k\cdot\sum_{i=i_1+1}^{i_2-1} |K_i[\vec{t}]| 
	\end{equation}
	Together,~\cref{eq-unbounded-ratio-a} and~\cref{eq-unbounded-ratio-b} imply \cref{eq-unbounded-ratio}, which proves our claim. 
	The claim means that the direction \impl{cond:unbounded-ratio}{cond:unbounded-fixed} of~\cref{ratio-lemma} is applicable to the semilinear set $A$ and yields an $\ell\in\NN$ such that for every $\delta\in\NN$, there is a $\vec{t}\in W$ with
	\begin{align*} 
		&|K_{i_1}[\vec{t}]|, |K_{i_2}[\vec{t}]| \ge \delta, 
		&&\sum_{j\in J}|I_j[\vec{t}]|+\sum_{i=i_1+1}^{i_2-1} |K_{i}[\vec{t}]|\le \ell, 
	\end{align*}
	which proves that $S$ has unbounded isolation.
\end{proof}

\subsection{Proof of Lemma \ref{harbor-existence}}\label{app:harbor-existence}
\harborExistence*
\begin{proof}
	Suppose $I_i[\vec{t}]$ is finite.
	We know, therefore that there exists an infinite interval $I$ (that is either $I_1[\vec{t}]$ or $I_m[\vec{t}]$) such that $I < I_i[\vec{t}]$ or $I_i[\vec{t}] > I$.
	We therefore know that there exists an interval that is larger than $I_i[\vec{t}]$ (which need not be infinite) that is to the left or to the right of $I_i[\vec{t}]$.
	We split the proof into two cases: the first is when there is a larger interval to the left and to the right; the second is when there is only a larger interval to the left and there does not exist a larger interval to the right.
	Note that the (remaining) scenario when there a larger interval to the right but no larger interval to the left is symmetric to the second case. 

	\paragraph*{Case 1} there exists two indices $j_1, j_2 \in [1, m]$ such that $j_1 < i < j_2$ and $|I_{j_1}[\vec{t}]|, |I_{j_2}[\vec{t}]| > |I_i[\vec{t}]|$.
	In fact, let $j_1$ be the greatest index such that $j_1 < i$ and $|I_{j_1}[\vec{t}]| > |I_i[\vec{t}]|$ and let $j_2$ be the least index such that $j_2 > i$ and $|I_{j_2}[\vec{t}]| > |I_i[\vec{t}]|$.
	In other words, $I_{j_1}[\vec{t}]$ is the closest interval to the left of $I_i[\vec{t}]$ that is larger than $I_i[\vec{t}]$ and $I_{j_2}[\vec{t}]$ is the closest interval to the right of $I_i[\vec{t}]$ that is larger than $I_i[\vec{t}]$.
	We claim that 
	\begin{equation*}
		d(I_i[\vec{t}], I_{j_1}[\vec{t}]) \leq \rho \cdot |I_i[\vec{t}]|
		\text{ or }
		d(I_i[\vec{t}], I_{j_2}[\vec{t}]) \leq \rho \cdot |I_i[\vec{t}]|.
	\end{equation*}

	To see why this is true, we shall assume for the sake of contradiction, that $d(I_i[\vec{t}], I_{j_1}[\vec{t}]), d(I_i[\vec{t}], I_{j_2}[\vec{t}]) > \rho \cdot |I_i[\vec{t}]|$.
	By definition of $j_1$, we know that for all $k \in (j_1, i), |I_k[\vec{t}]| \le |I_i[\vec{t}]|$. 
	Given that $d(I_i[\vec{t}], I_{j_1}[\vec{t}]) > \rho \cdot |I_i[\vec{t}]|$, we deduce that there exists $i_1 \in [1, r]$ such that $I_{j_1}[\vec{t}] < K_{i_1}[\vec{t}] < I_i[\vec{t}]$ and 
	\begin{align*}
		|K_{i_1}[\vec{t}]| 
		& > \frac{\rho \cdot |I_i[\vec{t}]| - (i - j_1 - 1) \cdot |I_i[\vec{t}]|}{i - j_1} \\
		& > \frac{\rho \cdot |I_i[\vec{t}]| - m \cdot |I_i[\vec{t}]|}{m} \\
		& = \frac{2R(m+1) \cdot |I_i[\vec{t}]| - m \cdot |I_i[\vec{t}]|}{m} \\
		& > R \cdot |I_i[\vec{t}]|.
	\end{align*}
	Symmetrically, by definition of $j_2$, we know that for all $k \in (i, j_2), |I_k[\vec{t}]| < |I_i[\vec{t}]|$.
	Given that $d(I_i[\vec{t}], I_{j_2}[\vec{t}]) > \rho \cdot |I_i[\vec{t}]|$, we can also deduce that there exists $i_2 \in [1, r]$ such that $I_i[\vec{t}] < K_{i_2}[\vec{t}] < I_{j_2}[\vec{t}]$ such that $|K_{i_2}[\vec{t}]| > R \cdot |I_i[\vec{t}]|$.
	
	Let $J = \set{ j \in [1, m] : K_{i_1}[\vec{t}] < I_j[\vec{t}] < K_{i_2}[\vec{t}] }$; observe that $J \sset (j_1, j_2)$.
	Furthermore, by definition of $j_1$ and $j_2$, we know that for all $j \in J$, $|I_j[\vec{t}]| \leq |I_i[\vec{t}]|$; so, with the above, we know that that
	\begin{equation}\label{eq:big-gaps-1}
		|K_{i_1}[\vec{t}]|, |K_{i_2}[\vec{t}]| > R\cdot |I_j[\vec{t}]| \text{ for all } j \in J.
	\end{equation}

	Now, given that $M_{S[\vec{t}]} \in \MM$, we know that there exists a matrix $D$ in the ideal decomposition of $\MM$ such that $M_{S[\vec{t}]} \in \downclose{D}$.
	Thus, for any $i$ and $j$ such that $(D)_{i,j} \in \NN$, it is true that $R \geq (D)_{i,j} \geq (M_{S[\vec{t}]})_{i,j}$.
	Thus, by~\cref{eq:big-gaps-1}, we know that $(M_{S[\vec{t}]})_{i_1, j}, (M_{S[\vec{t}]})_{i_2, j} > R$ for all $j \in J$.
	This implies that $(D)_{i_1, j} \not\in \NN$ and $(D)_{i_2, j} \not\in \NN$ for all $j \in J$; accordingly such entries must be $\omega$.
	This would imply that $S[\vec{t}]$ has an $\omega$-constellation, however this contradicts the assumption of this proposition.
	We therefore conclude that
	\begin{equation*}
		d(I_i[\vec{t}], I_{j_1}[\vec{t}]) \leq \rho \cdot |I_i[\vec{t}]|
		\text{ or }
		d(I_i[\vec{t}], I_{j_2}[\vec{t}]) \leq \rho \cdot |I_i[\vec{t}]|.
	\end{equation*}
	Lastly, in either scenario, we note that properties~\ref{itm:harbor-size} and~\ref{itm:harbor-nearest} of a harbor immediately follow from the definition of $j_1$ (and $j_2$).

	\paragraph*{Case 2} there exists an index $j_1 \in [1, m]$ such that $j_1 < i$ and $|I_{j_1}[\vec{t}]| > |I_i[\vec{t}]|$ and there does not exist an index $j_2 \in [1, m]$ such that $j_2 > i$ and $|I_{j_2}[\vec{t}]| > |I_i[\vec{t}]|$.
	This case is very similar to Case 1; instead of finding two unbounded gaps between $I_{j_1}[\vec{t}]$ and $I_{j_2}[\vec{t}]$, we will find one unbounded gap between $I_{j_1}[\vec{t}]$ and $I_i[\vec{t}]$ and use the non-existence of $j_2$ to deduce that the rightmost gap $K_r[\vec{t}]$ is infinite.

	Let $j_1 \in [1, m]$ be the greatest index such that $j_1 < i$ and $|I_{j_1}[\vec{t}]| > |I_i[\vec{t}]|$.
	Assume, for the sake of contradiction, that $d(I_i[\vec{t}], I_{j_1}[\vec{t}]) > \rho \cdot |I_i[\vec{t}]|$.
	By definition of $j_1$, we know that for all $k \in (j_1, i), |I_k[\vec{t}]| \le |I_i[\vec{t}]|$.
	We can therefore find an index $i_1 \in [1, r]$ such that $I_{j_1}[\vec{t}] < K_{i_1}[\vec{t}] < I_i[\vec{t}]$ and
	\begin{equation}\label{eq:large-left-interval}
		|K_{i_1}[\vec{t}]| 
		> \frac{\rho \cdot |I_i[\vec{t}]| - (i_1 - j - 1) \cdot |I_i[\vec{t}]|}{i-j_1} 
		> R \cdot |I_i[\vec{t}]|.
	\end{equation}

	Now, since there does not exist $j_2 \in [1, m]$ such that $j_2 > i$ and $|I_{j_2}[\vec{t}]| > |I_i[\vec{t}]|$, we know that $K_r[\vec{t}]$ is infinite.
	We now proceed as before in Case 1, but now with $i_2 = r$.
	This time, let $J = \set{ j \in [1, m] : K_{i_1}[\vec{t}] < I_j[\vec{t}] < K_r[\vec{t}]}$; observe that $J \sset (j_1, m]$.
	Given the definition of $j_1$ and the non-existence of $j_2$, we know that for all $j \in J, |I_j[\vec{t}]| \leq |I_i[\vec{t}]|$.
	By~\cref{eq:large-left-interval} and since $|K_r[\vec{t}]| = \infty$, it follows that
	\begin{equation}\label{eq:big-gaps-2}
		|K_{i_1}[\vec{t}]|, |K_r[\vec{t}]| > R \cdot |I_j[\vec{t}]| \text{ for all } j \in J.
	\end{equation}

	Since $M_{S[\vec{t}]} \in \MM$, we know that there exists a matrix $D$ in the ideal decomposition of $\MM$ such that $M_{S[\vec{t}]} \in \downclose{D}$.
	Thus, for all $i$ and $j$ such that $(D)_{i,j} \in \NN$, it is true that $R \geq (D)_{i,j} \geq (M_{S[\vec{t}]})_{i,j}$.
	By~\cref{eq:big-gaps-2}, we know that $(M_{S[\vec{t}]})_{i_1,j}, (M_{S[\vec{t}]})_{r,j} > R$ for all $j \in J$. 
	This implies that $(D)_{i_1, j} \notin \NN$ and $(D)_{r, j} \notin \NN$ for all $j \in J$; accordingly, such entries must be $\omega$.
	This would imply that $S[\vec{t}]$ has an $\omega$-constellation, however this contradicts the assumption of this proposition.
	We therefore conclude that 
	\begin{equation*}
		d(I_i[\vec{t}], I_{j_1}[\vec{t}]) \leq \rho \cdot |I_i[\vec{t}]|.
	\end{equation*}
	Lastly, we note that properties~\ref{itm:harbor-size} and~\ref{itm:harbor-nearest} of a harbor immediately follow from the definition of $j_1$.
\end{proof}

\subsection{Decomposition into transformed building blocks}\label{app:decomposition-transformed-building-blocks}
Here, we show how the equality
\[ S=T_\infty\cup T_{-\infty}\cup\bigcup_{\chi} T_\chi \]
on \cpageref{claim-decomposition-transformed-building-blocks} can be deduced
from \cref{harbor-existence}.  

The sets $T_\chi$ cover all the intervals that
are finite or one-sided infinite; and $T_\infty$ and $T_{-\infty}$ cover all
the intervals that are two-sided infinite (i.e.\ comprise all of $\ZZ$).
Finally, it is not hard to check that each set $T_\chi$, $T_{\infty}$, and
$T_{-\infty}$ are transformed building blocks. This is obvious for $T_{\infty}$
and $T_{-\infty}$. 

For $T_{\chi}$, in case $I_{i_\ell}[\vec{t}]$ is an interval of the form
$[s,\infty)$ (note that this only depends on the index $i_\ell$, due to
interval-uniformity), then this is clear as well: The interval endpoints of
$I_{i_1}[\vec{t}],\ldots,I_{i_\ell}[\vec{t}]$ can be computed using the
Presburger-definable functions $\alpha_1,\beta_1,\ldots,\alpha_m,\beta_m$ from
\cref{lem:intervals-interval-uniform}.  Hence, we can choose
$\gamma(\vec{t})\coloneqq\langle\alpha_{i_1}(\vec{t}),\beta_{i_1}(\vec{t}),\ldots,\alpha_{i_\ell}(\vec{t})\rangle$
and $\sigma\colon\ZZ^p\to \{0,1\}$ as the constant $0$, provided that $\chi$ is
indeed a harbor chain for $\vec{t}$. If $\chi$ is not a harbor chain for
$\vec{t}$, then we choose for $\gamma(\vec{t})$ some tuple that is not
$\rho$-admissible. Note that since the set of $\vec{t}$ for which $\chi$ is a
harbor chain is clearly Presburger-definable, and thus $\gamma$ can depend on
this.

However, in case $I_{i_\ell}[\vec{t}]$ is of the form $(-\infty,s]$, then we
can use
$\gamma(\vec{t})\coloneqq\langle -\alpha_{i_1}(\vec{t}),-\beta_{i_1}(\vec{t}),\ldots,-\alpha_{i_\ell}(\vec{t})\rangle$
and $\sigma\colon\ZZ^p\to \{0,1\}$ as the constant $1$. As above, we only do
this is $\chi$ is a harbor chain for $\vec{t}$, and otherwise $\gamma$ yields
an arbitrary non-$\rho$-admissible tuple. 

\subsection{Making instances growing}\label{appendix-make-growing}
The following lemma shows that we
can always assume that our input instance is growing. This is because if it is
not growing, then for some $I_i$, its length is in particular bounded by a
constant factor in $I_{i-1}$. And in this case, we can consider (i)~the
instance without $I_i$ (which, because of the constant factor, will still be
$\rho'$-admissible for a larger $\rho'$) and (ii)~the instance consisting just
of $I_i,I_{i+1},\ldots,I_{m+1}$. Since in both instances, the number of
intervals is strictly smaller, this reduction will (after at most $m$
repetitions) yield a growing instance.
\begin{restatable}{lemma}{makeGrowing}\label{make-growing}
	Let $\rho,m\ge 1$ be constants. There are $\rho',k\ge 1$ such that
	given an instance of $\intreach(S^{(\rho,m)})$, we can compute in
	logspace at most $k$ instances of $\intreach(S^{(\rho',r)})$ for some
	$r\in[1,m]$, such that the original instance is positive if and only if
	one of the new instances is positive.
\end{restatable}
\begin{proof}
	We proceed by induction on $m$. For $m=1$, the statement is trivial:
	Since $|I_2|=|I_{m+1}|=\infty$, every instance for $m=1$ is growing.

	Now suppose $m>1$ and that the intervals $I_1,\ldots,I_{m+1}$ for the
	parameter vector $\vec{t}\in\ZZ^p$ are not growing, say we have $|I_i|<
	2d(I_j,I_i)$ for some $i\in[2,m]$ and $j\in[1,i]$.  We build two new
	instances: One where we remove $I_i$, and one where we remove
	$I_1,\ldots,I_{i-1}$.

	First, let us consider $I_1,\ldots,I_{i-1},I_{i+1},\ldots,I_m$. Then
	for any $k\notin\{i-1,i\}$, we clearly have $d(I_k,I_{k+1})\le
	\rho\cdot |I_k|$. Moreover, we have $|I_i|<2d(I_{i-1},I_i)\le
	2\rho\cdot |I_{i-1}|$ and thus
	\begin{align*}
		d(I_{i-1},I_{i+1})&\le d(I_{i-1},I_i)+|I_i|+d(I_i,I_{i+1}) \\
		&\le \rho\cdot |I_{i-1}|+2\rho|I_{i-1}|+\rho|I_i| \\
		&\le \rho\cdot |I_{i-1}|+2\rho|I_{i-1}|+2\rho^2 |I_{i-1}| \\
		&\le 5\rho^2|I_{i-1}|.
	\end{align*}
	Thus, the instance $I_1,\ldots,I_{i-1},I_{i+1},\ldots,I_{m+1}$ is
	$5\rho^2$-admissible.

	Second, the instance $I_i,I_{i+1},\ldots,I_{m+1}$ is clearly $\rho$-admissible.
	Moreover, the union
	\[ (I_1\cup\cdots\cup I_{i-1}\cup I_{i+1}\cup\cdots\cup I_{m+1})\cup (I_i\cup\cdots\cup I_{m+1}) \]
	is the same as $I_1\cup\cdots\cup I_{m+1}$, so that the original
	instance is positive if and only if one of the two new instances is
	positive.

	Since both new instances have strictly fewer intervals, by induction,
	they can in turn be reduced to growing instances.
\end{proof}

\subsection{Proof of Lemma \ref{correctness-equations}}\label{appendix-correctness-equations}
\correctnessEquations*
\newcommand{\eqstep}{\Rightarrow}
\newcommand{\eqlr}{\Leftrightarrow}
\newcommand{\eqlrs}{\xLeftrightarrow{*}}
	In order to control the additional equalities incurred by taking the
	smallest congruence, let us define a step relation whose reflexive
	transitive closure will be exactly $\equiv_{\rho,\vec{t}}$.

	For any $i,j,k,h$ and $g\in\B[X,X^{-1}]$, we have 
	\begin{multline}\label{eq-step-left}
		g+X^i+X^j+X^k ~\eqstep~ g+X^{[i,j]}+X^k~\text{if $j-i\in[0,v_\ell]$, }\\
		\text{$k-j\in[u_\ell,v_\ell]$, and }
		\text{($I_1<I_\ell$ or $\ell=1$); }
	\end{multline}
	\begin{multline}\label{eq-step-right}
		g+X^h+X^i+X^j ~\eqstep~ g+X^h+X^{[i,j]}~\text{if $j-i\in[0,v_\ell]$, }\\		
		\text{$i-h\in[u_\ell,v_\ell]$, and }
		\text{$I_\ell<I_1$}.
	\end{multline}

Let $\eqlr$ be defined as $\eqstep\cup\eqstep^{-1}$, i.e.\ the symmetric
closure of $\eqstep$. Moreover, let $\eqlrs$ be the reflexive transitive closure of
$\eqlr$. We now show that $\eqlrs$ is indeed identical to $\equiv_{\rho,\vec{t}}$.
\begin{lemma}\label{congruence-step-characterization}
	The congruence $\equiv_{\rho,\vec{t}}$ is identical to $\eqlrs$.
\end{lemma}
\begin{proof}
	The relation $\eqstep$ is clearly included in $\equiv_{\rho,\vec{t}}$, and
	thus $\eqlrs$ is included in $\equiv_{\rho,\vec{t}}$. Moreover, the pairs
	listed in \cref{eq-left,eq-right} are included in $\eqstep$. Thus, it
	remains to be shown that $\eqlrs$ is a congruence relation.  By
	definition, $\eqlrs$ is symmetric, reflexive, and transitive. Hence, we
	only need to show that $f_1\eqlrs f'_1$ and $f_2\eqlrs f'_2$ imply
	$f_1+f_2\eqlrs f'_1+f'_2$ and $f_1f_2\eqlrs f'_1f'_2$. For this, in turn, it suffices to show that
	\begin{equation} f_1\eqlrs f'_1~\text{implies}~f_1+f_2\eqlrs f'_1+f_2~\text{and}~f_1f_2\eqlrs f'_1f_2\label{one-unchanged}\end{equation}
		for every $f_2\in\B[X,X^{-1}]$.
		Indeed, if $f_1\eqlrs f'_1$ and $f_2\eqlrs f'_2$, then \eqref{one-unchanged} yields
	\[ f_1+f_2\eqlrs f'_1+f_2\eqlrs f'_1+f'_2~\text{and}~f_1f_2\eqlrs f'_1f_2\eqlrs f'_1f'_2, \]
	since addition and multiplication are commutative.

	To prove \eqref{one-unchanged}, in turn, it suffices to prove
	\begin{equation} f_1\eqlrs f'_1~\text{implies}~f_1+X^i\eqlrs f'_1+X^i~\text{and}~f_1X^i\eqlrs f'_1X^i\label{monomial-unchanged}\end{equation}
		for every $i\in\Z$. Indeed, if $f_1\eqlrs f'_1$ and $f_2=X^{i_1}+\cdots+X^{i_k}$, then 
	applying $k$-times the implication \eqref{monomial-unchanged} yields $f_1+f_2\eqlrs f'_1+f_2$. Moreover, we have 
	\begin{align*}
		f_1f_2&=f_1X^{i_1}+\cdots+f_1X^{i_k}\eqlrs f'_1X^{i_1}+\cdots+f'_1X^{i_k} \\
		&=f'_1(X^{i_1}+\cdots+X^{i_k})=f'_1f_2,
	\end{align*}
	where the first $\eqlrs$ relation follows by $k$ applications of \eqref{monomial-unchanged},
and the argument above for adding up equivalent terms.

	Finally, \eqref{monomial-unchanged} follows by observing that
	$f_1\eqstep f'_1$ implies $f_1+X^i\eqstep f'_1+X^i$ and also
	$f_1X^i\eqstep f'_1X^i$.
\end{proof}

We are now in a position to prove \cref{correctness-equations}.
\begin{proof}[Proof of \cref{correctness-equations}]
	Because of \cref{congruence-step-characterization}, it suffices to show
	that if $f\eqstep f'$ for some $f,f'\in\B[X,X^{-1}]$, then $f'$
	contains a point in $S^{(\rho,m)}[\vec{t}]=I_1\cup\cdots\cup I_{m+1}$
	if and only if $f$ does. Since $f\eqstep f'$ means that $f'$ has a
	superset of the terms of $f$, the ``if'' direction is trivial. Thus, we
	only need to show: If $f$ contains no point in $I_1\cup\cdots\cup
	I_{m+1}$, then $f'$ does not either.

	Suppose $f\eqstep f'$ via rule \eqref{eq-step-left}, i.e.
	\begin{align*}
		f&=g+X^i+X^j+X^k, \\
		f'&=g+X^{[i,j]}+X^k
	\end{align*}
	for $j-i\in[0,v_\ell]$ and $k-j\in[u_\ell,v_\ell]$ and we have
	$I_1<I_\ell$ or $\ell=1$. The case of \eqref{eq-step-right} (i.e.\ $I_\ell<I_1$) is
	symmetric, so we don't treat it explicitly.
	
	Suppose $f'$ contains a point of $I_1\cup\cdots\cup I_{m+1}$. Note that
	$f'$ only introduces points in $X^{[i+1,j-1]}$, and since $f$ contains no point
	in $I_1\cup\cdots\cup I_{m+1}$, any violating point in $f'$ would need to be in
	$[i+1,j-1]$. Suppose $f'$ intersects $I_r$. Then $I_r$ would need to lie in
	$[i+1,j-1]$, because $f$ does not intersect $I_r$. But since
	$j-i\in[0,v_\ell]$, this means $|I_r|<v_\ell=|I_\ell|$, meaning $r<\ell$. This
	already rules out the case $\ell=1$, so we are in the case $I_1<I_\ell$. By the
	conditions of $\rho$-chains of intervals, if $I_1$ is to the left of $I_\ell$,
	then all intervals $I_1,\ldots,I_{\ell-1}$ are to the left of $I_\ell$. Thus,
	the points $i,j$ and the intervals are arranged as follows:
	\[ i,~I_r,~j,~I_\ell \]
	across $\Z$. Indeed, $j$ has to lie between $I_r$ and $I_\ell$, because
	$j-i\in[0,v_\ell]=[0,|I_\ell|]$, and if $j$ were beyond $I_\ell$, the
	difference $j-i$ would have to be larger than $|I_\ell|$.

	Finally, note that $k-j\in[u_\ell,v_\ell]$ implies that $k\in I_\ell$:
	We have $k\ge j+u_i\ge j+d(I_r,I_\ell)$, so $k$ is at least the left
	end-point of $I_\ell$. Furthermore, we have $k\le j+v_\ell$, meaning
	$k$ cannot exceed the right end-point of $I_\ell$. However, since $X^k$
	occurs in $f$, this would mean that $f$ already intersects $I_\ell$,
	which is impossible.
\end{proof}

\subsection{Proof of Lemma \ref{small-representation}}\label{app:small-representation}
Let us now describe how we exploit the equations in $\sr{\rho}{\vec{t}}$ to obtain small representations of elements of $\sr{\rho}{\vec{t}}$.
Specifically, we will se that every polynomial $f\in\B[X,X^{-1}]$ has an equivalent sum
of at most $(m\rho+m)^{m+1}$ (i.e.\ constantly many) intervals.
\smallRepresentation*

We first need a small observation about the numbers $u_i$ and $v_i$.
\begin{restatable}{lemma}{uvProperties}\label{uv-properties}
For $i\in[0,m]$, we have $u_{i+1}/v_i\le m(\rho+1)$ and $v_i/u_i\ge 2$.
\end{restatable}
\begin{proof}
First, note that
\begin{align*}
	d(I_j,I_i)&\le d(I_j,I_{j+1})+\sum_{k=j+1}^{i-1} (|I_k|+d(I_k,I_{k+1})) \\
	&\le \rho\cdot|I_j|+\sum_{k=j+1}^{i-1} (|I_k|+\rho\cdot |I_k|) \\
	&=m(\rho+1)|I_{i-1}|,
\end{align*}
where the first inequality is due to the triangle inequality for intervals. Thus
\[ \frac{u_{i+1}}{v_i}=\frac{d(I_j,I_{i+1})}{|I_i|}\le \frac{m(\rho+1)|I_i|}{|I_{i}|}=m(\rho+1) \]
where $j\in[1,i+1]$ maximizes $d(I_j,I_{i+1})$. 

Furthermore, the fact that $v_i/u_i\ge 2$ follows from the fact that our
	instance is growing (recall that we assume that our instance is growing, see \cref{make-growing}).
\end{proof}

For $\ell\in[1,m+1]$, we call a polynomial an \emph{$\ell$-cluster} if between
any two terms $X^i$ and $X^j$, there is a path of terms in $f$ from $X^i$ to
$X^j$ such that in each step, the difference in exponents is at most $v_\ell$.
For \cref{small-representation}, we first prove the following lemma about clusters.

\begin{lemma}\label{two-points-good-distance}
	If $f_1,\ldots,f_{r}$ are distinct (hence disjoint) $(\ell-1)$-clusters
	with $r\ge m(\rho+1)$, then $f_1+\cdots+f_{r}$ contains two points with
	distance in $[u_{\ell},v_{\ell}]$.
\end{lemma}
\begin{proof}
Note that a $1$-cluster is equivalent to a single interval $X^{[i,j]}$ for some
$i$ and $j$.  Now consider an $\ell$-cluster $f$ for $\ell>1$. It can be
written as a sum $f_1+\cdots+f_r$, where each $f_j$ is an $(\ell-1)$-cluster
for $j\in[1,r]$, and $f_j$ is to the left of $f_{j+1}$, for each $j\in[1,r-1]$.
Since $f$ is an $\ell$-cluster and the $f_j$ are distinct $(\ell-1)$-clusters,
we know that the distance between $f_j$ and $f_{j+1}$ is at least
$v_{\ell-1}+1$, but at most $v_{\ell}$. 

Let us consider the case that $I_{\ell}<I_1$ (the case $I_1<I_{\ell}$ is
analogous). We claim that within the first $m(\rho+1)$ many $(\ell-1)$-clusters
of $f_1,\ldots,f_r$, there must be two points with distance in
$[u_{\ell},v_{\ell}]$. This will make \eqref{eq-right} applicable, which
lets us interpolate all other $(\ell-1)$-clusters into a single interval.

Since $d(f_j,f_{j+1})>v_{\ell-1}$ and $u_{\ell}/v_{\ell-1}\le m(\rho+1)$, there must
clearly be two points within $f_1,\ldots,f_{m(\rho+1)}$ of distance at least
$u_{\ell}$.  It remains to argue that there must be two such points where the
distance is also at most $v_{\ell}$. Consider the left-most point $X^h$ in
$f_1+\cdots+f_r$. Moreover, let $X^i$ be the left-most point $X^i$ among those that have distance
$\ge u_{\ell}$ from $X^h$. Now if $i-h=u_{\ell}$, then we are done because
$u_{\ell}\le 2u_{\ell}\le v_{\ell}$. If $i-h>u_{\ell}$, then to the
left of $X^i$ the interval sum $f_1+\cdots+f_r$ must have a gap directly to the
left of $X^i$. Let $X^s$ be the term to the left of this gap. 

Observe that if $i-s\ge u_{\ell}$, then we can choose $X^s,X^i$ as our two
points, because they have difference in $[u_{\ell},v_{\ell}]$ (recall that
all gaps in $f_1+\cdots+f_r$ are at most $v_{\ell}$, since it is an
$\ell$-cluster). If, on the other hand $i-s<u_{\ell}$, then we know that
$s-h<u_{\ell}$ (by choice of $i,h$), and also $i-s<u_{\ell}$. Together,
this means $u_{\ell}\le i-h=(i-s)+(s-h)<2u_{\ell}\le v_{\ell}$.
\end{proof}

Hence, within the first $m(\rho+1)$ many $(\ell-1)$-clusters, we have found two
points with distance in $[u_{\ell},v_{\ell}]$. We can therefore interpolate
all the other $(\ell-1)$-clusters into a single interval using \eqref{eq-right}.
Since a single interval is in particular an $(\ell-1)$-cluster, we have written our
$\ell$-cluster as a sum of $m(\rho+1)+1\le m(\rho+2)$ many $(\ell-1)$-clusters:
\begin{lemma}
	Let $\rho,m\ge 1$, $\vec{t}\in\ZZ^{2m+1}$ and let $\ell\in[1,m+1]$. Each
	$\ell$-cluster $f\in\B[X,X^{-1}]$ is
	$\equiv_{\rho,\vec{t}}$-equivalent to a sum of at most $m(\rho+2)$ many
	$(\ell-1)$-clusters.
\end{lemma}

Finally, \cref{small-representation} follows by induction. First,
the previous lemma implies that each $\ell$-cluster can be written as a sum of at most
$(m\rho+2m)^{\ell-1}$ intervals.  Since every polynomial is an $(m+1)$-cluster, any
polynomial can be written as a sum of $(m\rho+2m)^{m}$ intervals.

\subsection{Proof of Proposition \ref{building-block-ac1}}\label{app:building-block-ac1}
\buildingBlockACOne*
Motivated by \cref{small-representation}, we say that a polynomial
$f'\in\B[X,X^{-1}]$ is a \emph{$(\rho,\vec{t})$-representation} of
$f\in\B[X,X^{-1}]$ if $f'\equiv_{\rho,\vec{t}} f$ and also $f$ is a sum of at
most $(m\rho+2m)^{m}$ intervals.

Let $C\in\Mat(n\times n,\sr{\rho}{\vec{t}})$ be the matrix of equivalence
classes of $B$, i.e.\ $C_{i,j}=[B_{i,j}]$. Then since $\equiv_{\rho,\vec{t}}$
is a congruence, we have
\[ (C^n)_{i,j}=[(B^n)_{i,j}], \]
in other words, the $(i,j)$ entry of $B^n$ is the equivalence class of the
$(i,j)$ entry of $B^n$. Thus, to decide in $\AC^1$ whether we can reach a
number in $I_1\cup\cdots\cup I_{m+1}$, it suffices to compute the matrix $C^n$.

For this, we employ repeated squaring. We define $C^{(1)}:=C$ and
$C^{(r+1)}:=(C^{(r)})^2$.  Then clearly, $C^{(\lceil \log n\rceil)}=C^n$ (note
that, as every path has length $\le n$, $C^{n+1}=C^n$). 
This means, it
suffices to show that given a matrix \hbox{$E\in\Mat(n\times n,\sr{\rho}{\vec{t}})$},
we can compute $E^2$ in $\AC^0$. Then, applying this circuit $\lceil\log
n\rceil$ times, and starting with $C$, will yield $C^{(\lceil \log n\rceil)}$ in a circuit of logarithmic height, hence in $\AC^1$.
Thus, it remains to show:
\begin{proposition}\label{square-matrix-ac0}
	Let $\rho,m\ge 1$ be constants. Given $\vec{t}\in\ZZ^p$ and a matrix
	$E\in\Mat(n\times n,\sr{\rho}{\vec{t}})$, we can compute $E^2$ in
	$\AC^0$.
\end{proposition}
\Cref{square-matrix-ac0} easily follows from the following:
\begin{lemma}\label{compute-representation}
	Let $\rho,m\ge 1$ be constants. Given $\vec{t}\in\ZZ^p$ and a
	polynomial $f\in\B[X,X^{-1}]$, we can compute a
	$\tuple{\rho, \vec{t}}$-representation of $[f]$ in $\AC^0$.
\end{lemma}
Let us see how~\cref{square-matrix-ac0} follows from~\cref{compute-representation}. 
We are given $E=(e_{i,j})_{i,j}$, where each
$E_{i,j}$ is given as a $\tuple{\rho, \vec{t}}$-representation
$f_{i,j}\in\B[X,X^{-1}]$, i.e.\ $E_{i,j}=[f_{i,j}]$. And we want to compute
$E^2$. The $(i,j)$ entry in $E^2$ is
\[ \sum_{k=1}^n e_{i,k}e_{k,j}=[f_{i,k}f_{k,j}].\]
For each $k$, we can easily compute the polynomial $f_{i,k}f_{k,j}$. This is
because for polynomials $f,f'\in\B[X,X^{-1}]$ given as sums of intervals
$X^{[h,\ell]}$, their product is the sum of all products
$X^{[h,\ell]}X^{[h',\ell']}=X^{[h+h',\ell+\ell']}$, where $X^{[h,\ell]}$ is an
interval occurring in $f$, and $X^{[h',\ell']}$ is an interval occurring in
$f'$. Since the sum of two numbers can be computed in $\AC^0$, we can compute
each $f_{i,k}f_{k,j}$. Moreover, we can then compute the sum $g:=\sum_k
f_{i,k}f_{k,j}$ in $\B[X,X^{-1}]$. Then, \cref{compute-representation} allows
us to compute a $(\rho,\vec{t})$-representation of $g$. Doing this for all
entries of $E^2$ in parallel yields the $\tuple{\rho, \vec{t}}$-representation of all
entries in $E^2$ in $\AC^0$.

\begin{lemma}\label{first-cluster-ac0}
	Let $\rho,m\ge 1$ be constants and $\ell\in[0,m]$. Given a polynomial
	$f\in\B[X,X^{-1}]$, we can decide in $\AC^0$ whether a given interval in $f$
	belongs to the left-most (resp., right-most) $\ell$-cluster in $f$.
\end{lemma}

\begin{proof}
	The polynomial $f$ is represented as a sequence of pairs $\tuple{x_i,y_i}$ that represent terms $X^{[x_i, y_i]}$, which in turn is the sum $\sum_{k\in[x_i,y_i]} X^k$. If the interval $[x_j, y_j]$ is to the right 
of $[x_i, y_i]$, meaning $x_j > y_i$, we check whether the gap $x_j - y_i$ exceeds the threshold $v_\ell$ and no other interval $[x_k, y_k]$ satisfies $x_k \leq x_j$ and $y_k \geq y_i$.
 This means that no interval exists in the gap between $[x_i, y_i]$ and $[x_j, y_j]$. If these conditions hold, $[x_j, y_j]$ marks the start of a new $\ell$-cluster. The verification of these
  conditions involves evaluating the conjunction of the conditions for all intervals $[x_k, y_k]$, which can be done in $\AC^0$ using unbounded fan-in operations.

To determine whether a given interval belongs to the left-most $\ell$-cluster (resp., right-most), we first find the left-most pair $[x_i, y_i], [x_j, y_j]$ that satisfies the earlier conditions,
 which can also be done in $\AC^0$. Then, without loss of generality, we assume that $[x_j, y_j]$ is to the right of $[x_i, y_i]$. We verify that the given interval lies between the smallest 
 term of $f$ and $y_i$ (resp., between $x_j$ and the largest term of $f$).

Thus, we can decide in $\AC^0$ whether a given interval belongs to the left-most (resp., right-most) $\ell$-cluster in $f$.
\end{proof}

With \cref{first-cluster-ac0} in hand, we are now prepared to prove
\cref{compute-representation}:
\begin{proof}[Proof of \cref{compute-representation}]
	We show by induction on $\ell\in[0,m+1]$ that given a polynomial
	$f\in\B[X,X^{-1}]$ that is an $\ell$-cluster, we can compute a
	$\tuple{\rho, \vec{t}}$-representation of $f$. To make the induction work, we
	construct something slightly more specific: We want a
	\emph{$\tuple{\rho,\vec{t},\ell}$-representation}, which is a sum of at most
	$(m\rho+2m)^\ell$ intervals.

	For $\ell=1$, this is trivial, since as a $1$-cluster, $f$ is
	equivalent to $X^{[i,j]}$, where $i$ is the lowest number occurring in
	$f$, and $j$ is the highest number occurring in $f$. This can clearly
	be computed in $\AC^0$.

	Let us now suppose we have constructed the $\AC^0$-circuits for $\ell$
	and want to construct them for $\ell+1$. Let us assume $I_{\ell+1}<I_1$
	(the case $I_1<I_{\ell+1}$ is symmetric). Using
	\cref{first-cluster-ac0}, we first mark all those intervals in $f$ that
	belong to the left-most $\ell$-cluster. Then we repeat this to mark all
	intervals that belong to the second (from the left) $\ell$-cluster,
	etc., until we have distinct markings for the first $m(\rho+1)$ many
	$\ell$-clusters. Note that each invocation of \cref{first-cluster-ac0}
	incurs some circuit depth $h$, so in total, we arrive at circuit depth
	$h\cdot m(\rho+1)$, which is still a constant. Hence, we have
	$f=f_1+\cdots+f_{m(\rho+1)}+g$, where $f_1,\ldots,f_{m(\rho+1)}$ are
	the marked $\ell$-clusters (or empty), and $g\in\B[X,X^{-1}]$ contains
	the remaining intervals in $f$. 

	We now apply our $\AC^0$ algorithm for $\ell$ (which exists by
	induction) to each interval subset $f_j$, to obtain a
	$\tuple{\rho,\vec{t},\ell}$-representation of $f_j$. Then, if $g\ne 0$, we
	apply \cref{two-points-good-distance} to conclude that
	$f_1+\cdots+f_{m(\rho+1)}$ contain two points of distance
	$[u_\ell,v_{\ell}]$, which allows us to apply \cref{eq-right} to
	collapose all of $g$ (and perhaps part of $f_1+\cdots+f_{m(\rho+1)}$)
	into a single interval. This yields a
	$\tuple{\rho,\vec{t},\ell+1}$-representation of $f$.
\end{proof}

	\section{Additional material on Section~\ref{sec:natural-semantics}}
        \label{app:natural-semantics}
	\subsection{Proof overview}\label{app:natural-semantics-overview}
Here, we deduce \cref{main-result-natural,main-result-natural-decidable} from
the results in \cref{sec:integer-semantics}. In the following, let
$S\subseteq\ZZ^{p}\times\NN$ be semilinear, and let $S'\coloneqq S\cup
(\ZZ\times\ZZ_{<0})$. In other words, $S'[\vec{t}]$ is obtained from
$S[\vec{t}]$ by putting in all negative numbers. Our proofs will essentially
the results from \cref{sec:integer-semantics} to $S'$.  The first step is to
express local${}^+$ density in terms of local density:
\begin{restatable}{proposition}{denseNaturalVsInteger}\label{dense-natural-vs-integer}
	$\densityplus(S)>0$ if and only if $\density(S')>0$.
\end{restatable}
Here, we show that
$\densityplus(A)\ge\density(A\cup\ZZ_{<0})\ge\tfrac{1}{2}\densityplus(A)$ for
every $A\subseteq\NN$. The first inequality is by definition, and the second
follows by taking an arithmetic progression centering at some $x\in A$ and
separately considering its part within $\ZZ_{<0}$ and its part within $A$.  For
details, see \cref{app:dense-natural-vs-integer}. The $\AC^1$ upper bound follows
from \cref{dense-natural-vs-integer} and
\cref{main-result-integer}:
\begin{restatable}{corollary}{naturalSemanticsUpperBound}\label{natural-semantics-upper-bound}
	If $\densityplus(S)>0$, then
	$\natreach(S)$ is in $\AC^1$.
\end{restatable}
Indeed, we know that $\density(S')>0$ and thus $\intreach(S')$ is in $\AC^1$ by
\cref{main-result-integer}. Moreover, an instance of $\natreach(S)$ can be viewed
as an instance of $\intreach(S')$.  For details, see
\cref{app:natural-semantics-upper-bound}. $\NP$-hardness also follows from
\cref{sec:integer-semantics}:
\begin{restatable}{corollary}{naturalSemanticsLowerBound}\label{natural-semantics-lower-bound}
	$\natreach(S)$ is $\NP$-complete if $\densityplus(S) = 0$.
\end{restatable}

We know from \cref{dense-natural-vs-integer} that $\density(S')=0$. By \cref{dense-residue-classes}, $S'$ must have a residue class with local density zero, and thus with unbounded isolation by \cref{integer-dichotomy-equivalence}. One then observes that in a residue class of $S'$, unbounded isolation must occur within the natural numbers (up to a bounded overlap with the negative integers). This permits a similar reduction as in \cref{intreach-hardness}. For details, see \cref{app:natural-semantics-lower-bound}.

Finally, \cref{main-result-natural-decidable} follows directly from \cref{main-result-integer-decidable} and \cref{dense-natural-vs-integer}.
\subsection{Proof of Proposition \ref{dense-natural-vs-integer}}\label{app:dense-natural-vs-integer}
\denseNaturalVsInteger*
\begin{proof}
	Note that clearly $\densityplus(S)\ge \density(S')$, because every
	expression that participates in the infimum for $\densityplus(S)$ also
	participates in the infimum for $\density(S')$. For the converse, we
	show that $\density(S')\ge \tfrac{1}{2}\densityplus(S)$. For this, it suffices to show that for $A\subseteq\NN$, we have 
	\[ \density(A\cup\ZZ_{<0})\ge \tfrac{1}{2}\densityplus(A). \]
	Consider and $k,n\ge 1$. If $x\in\ZZ_{<0}$, then $\density_{k,n}(A,x)\ge\tfrac{1}{2}$ and in particular $\density_{k,n}(A,x)\ge \tfrac{1}{2}\densityplus(A)$. Now suppose $x\in A$. Let $n^+$ be the largest number in $[0,n]$ such that $[x-kn^+,x+kn^+]\subseteq\NN$. We decompose $(A\cup\ZZ_{<0})\cap [x-kn,x+kn]$ into three sets:
	\begin{align*} 
		(A \cup \ZZ_{<0}) \cap [x-kn,x+&kn] = \underbrace{\ZZ_{<0} \cap [x-kn,x+kn]}_{B_0} \\
		&~\cup~ \underbrace{A\cap [x-kn^+,x+kn^+]}_{B_1} \\
	&~\cup~\underbrace{A\cap[x+kn^++k,x+kn]}_{B_2}. 
	\end{align*}
	Now observe 
	\begin{align*}
		\frac{|B_0\cap (x+k\ZZ)|}{2n+1}&=\frac{n-n^+}{2n+1}, \\
		\frac{|B_1\cap (x+k\ZZ)|}{2n+1}&=\densityplus_{k,n^+}(A,x)\cdot\frac{2n^++1}{2n+1}
	\end{align*}
	Therefore,
	\begin{align*}
		\density_{k,n}(A\cup\ZZ_{<0},x)&=\frac{|B_0\cap (x+k\ZZ)|}{2n+1}+\frac{|B_1\cap (x+k\ZZ)|}{2n+1} \\
		&~~~~~~~+\frac{|B_2\cap (x+k\ZZ)|}{2n+1} \\
		&\ge \frac{n-n^+}{2n+1}+\frac{2n^++1}{2n+1}\cdot \densityplus_{k,n^+}(A,x)\\
		&\ge \frac{n-n^+}{2n+1}+\frac{2n^++1}{2n+1}\densityplus(A).
	\end{align*}
	We claim that $\density_{k,n}(A\cup\ZZ_{<0})\ge \tfrac{1}{2}\densityplus(A)$. We distinguish two cases.
	\begin{itemize}
		\item If $\densityplus(A)\ge\tfrac{1}{2}$.
			\begin{align*}
				\density_{k,n}(A\cup\ZZ_{<0})&\ge \frac{n-n^+}{2n+1}+\frac{2n^++1}{2n+1}\densityplus(A) \\
				&\ge \frac{n-n^+}{2n+1}+\frac{2n^++1}{2n+1}\cdot\frac{1}{2} \\
				&=\frac{n-n^+ + n^+ +\tfrac{1}{2}}{2n+1}\ge \frac{1}{2}\ge \tfrac{1}{2}\densityplus(A).
			\end{align*}
		\item If $\densityplus(A)<\tfrac{1}{2}$. We claim that the expression
			\begin{equation} \frac{n-n^+}{2n+1}+\frac{2n^++1}{2n+1}\densityplus(A)\label{exp-positive} \end{equation}
		is at least $\densityplus(A)$. This follows by considering the possible values for $n^+$: If $n^+=n$, then the expression is exactly $\densityplus(A)$. When we diminish $n^+$ by $1$, then the expression grows by $\tfrac{1}{2n+1}-\tfrac{2}{2n+1}\densityplus(A)$. But the latter can be written as
			\[ \frac{1-2\densityplus(A)}{2n+1} \]
			and since $\densityplus(A)<\tfrac{1}{2}$, that quantity is positive. Thus, for smaller values of $n^+$, the expression \eqref{exp-positive} can only grow. Hence the expression is indeed at least $\densityplus(A)$. Since $\densityplus(A)\ge\tfrac{1}{2}\densityplus(A)$, this completes the case of $\densityplus(A)<\tfrac{1}{2}$.
	\end{itemize}
	This proves the claim, which in turn
	implies $\density(A\cup\ZZ_{<0})\ge\tfrac{1}{2}\densityplus(A)$.
\end{proof}

\subsection{Proof of Corollary \ref{natural-semantics-upper-bound}}\label{app:natural-semantics-upper-bound}
\naturalSemanticsUpperBound*
\begin{proof}
	Consider $S'\coloneqq S\cup (\ZZ^p\times\ZZ_{<0})$. By
	\cref{dense-natural-vs-integer}, we know that $\density(S')>0$ and
	hence \cref{main-result-integer} implies that $\intreach(S')$ is in
	$\AC^1$. However, $\natreach(S)$ logspace-reduces to $\intreach(S')$.
	Indeed, given an automaton with natural updates and $\vec{t}\in\ZZ^p$,
	we can use the same instance for $\intreach(S')$: Since the updates are
	natural, it can reach a number in $S[\vec{t}]$ iff it can reach a
	number in $S'[\vec{t}]=S[\vec{t}]\cup \ZZ_{<0}$. 
\end{proof}
\subsection{Proof of Corollary \ref{natural-semantics-lower-bound}}\label{app:natural-semantics-lower-bound}
\naturalSemanticsLowerBound*
\begin{proof}
	Let $S'=S\cup (\ZZ^{p}\times\Z_{<0})$. Then by
	\cref{dense-natural-vs-integer}, we have $\density(S')=0$.  By
	\cref{lem:modulo-free}, there is a $B\ge 1$ such that for each
	$\vec{b}\in[0,B-1]^{p+1}$, the set $[S']_{B,\vec{b}}$ is defined by a
	modulo-free Presburger formula. Moreover, \cref{dense-residue-classes},
	there must be a $\vec{b}\in[0,B-1]^{p+1}$ such that
	$\density([S']_{B,\vec{b}})=0$.  By \cref{integer-dichotomy-equivalence},
	this means that $[S']_{B,\vec{b}}$ has unbounded isolation. 

	By \cref{complexity-transfer}, we know that $\NP$-hardness of
	$\natreach([S']_{B,\vec{b}})$ implies $\NP$-hardness of
	$\natreach(S')$, which is the same problem as $\natreach(S)$. Hence, it
	remains to show $\NP$-hardness of $\natreach(S'')$, where $S''=[S']_{B,\vec{b}}$.

	Notice that since $S'[\vec{t}]=S[\vec{t}]\cup\Z_{<0}$ contains all
	negative numbers, this means $S''[\vec{t}]=[S']_{B,\vec{b}}[\vec{t}]$
	contains all negative numbers, except perhaps an interval
	$I\subseteq\Z_{<0}$ of size $|I|\le B$.

	Let us now exploit unbounded isolation of $S''$. We know that for each $\delta\ge 1$, there exist $\vec{t}\in\ZZ^p$ and
	$x\in\ZZ$ such that $S''[\vec{t}]\cap[x,x+\ell-1]\ne\emptyset$,
	$S''[\vec{t}]\cap[x-\delta,x-1]=\emptyset$, and
	$S''[\vec{t}]\cap[x+\ell,x+\ell+\delta]=\emptyset$. This implies that
	there exists a $k\in[0,\ell-1]$ such that for every $\delta\ge 1$,
	there are $\vec{t}\in\ZZ^p$ and $x\in\ZZ$ such that
	\begin{enumerate}
		\item $x+k\in S''[\vec{t}]$,
		\item $S''[\vec{t}]\cap[x-\delta,x-1]=\emptyset$, and
		\item $S''[\vec{t}]\cap[x+\ell,x+\ell+\delta]=\emptyset$.
	\end{enumerate}
	We call such a pair $\tuple{\vec{t},x}$ a \emph{hard pair} for $\delta$.

	Consider a Presburger-definable well-order $\preceq$ on $\ZZ^{p+1}$.
	For example, within each orthant, we use the lexigraphic ordering
	w.r.t.\ absolute values of the entries; and then we pick an arbitrary linear
	ordering on the $2^{p+1}$ orthants. Since we know that for each
	$\delta\ge 1$, there exists a hard pair, we can define a function
	$f\colon\NN\to\ZZ^p\times\ZZ$ in Presburger arithmetic such that for
	each $\delta\ge 1$, the pair $f(\delta)$ is a hard pair for $\delta$:
	The function picks the smallest hard pair w.r.t.\ the ordering
	$\preceq$. Since $f$ is Presburger-definable, it is computable in
	logarithmic space by \cref{compute-presburger}.

	Consider a subset sum instance $\tuple{a_1,\ldots,a_n,b}$ with
	$a_1,\ldots,a_n,b\in\NN$ and take $\delta=\ell a_1+\cdots+\ell a_n+\ell
	b+B$ and compute a hard pair $\tuple{\vec{t},x}$ for $\delta$. Now observe
	that since $S''[\vec{t}]$ contains all numbers $\le -B$, and we have
	$S''[\vec{t}]\cap [x-\delta,x-1]=\emptyset$, this implies that
	$x-\delta\ge -B$, and thus $x+k\ge x\ge \ell b$.  Consider the
	automaton in~\cref{fig:vas-np-hard}, and use it with $x_i \coloneqq \ell a_i$
	for $i\in[1,n]$, $y \coloneqq 0$, and $v \coloneqq x+k-\ell b$. 
	Then all transition weights are indeed non-negative, since $x+k\ge \ell b$. 
	Moreover, use the vector
	$\vec{t}$ as the parameter vector for our instance of $\natreach(S'')$.
	We claim that the automaton can reach $S$ if and only if the instance
	$\tuple{a_1,\ldots,a_n,b}$ is positive.

	If $\tuple{a_1,\ldots,a_n,b}$ is a positive instance, then this automaton can
	clearly reach $x+k$. Conversely, if $\tuple{a_1,\ldots,a_n,b}$ is a negative
	instance of subset sum, then any sum of a subset of $\{\ell
	a_1,\ldots,\ell a_n\}$ will be either $\le \ell b-\ell$ or $\ge \ell
	b+\ell$. Therefore, any number reached in the final state of our
	automaton will belong to $[x+k-\delta,x+k-\ell]$ or to
	$[x+k+\ell,x+k+\delta]$. However, both of these sets have an empty
	intersection with $S''[\vec{t}]$, by the choice of $\vec{t}$ and $x$.
	Hence, the instance of $\natreach(S'')$ is also negative.
\end{proof}

        \section{Additional material on Section~\ref{sec:vass-semantics}}
        \label{app:vass-semantics}
        \subsection{Proof of Lemma \ref{vass-reduction}}

In this subsection, we will prove:
\vassReduction*
The first step to proving \cref{vass-reduction} is that, intuitively, removing
a bounded set of points will not increase complexity. 
\begin{lemma}\label{remove-bounded-set}
	Suppose $S,T\subseteq\ZZ^{p+1}$ are Presburger-definable and there is a
	constant $K\ge 0$ so that for every $\vec{t}\in\Z^p$, we have
	$T[\vec{t}]\subseteq S[\vec{t}]$ and $|S[\vec{t}]\setminus
	T[\vec{t}]|\le M$. Then $\vasreach(T)$ logspace-reduces to
	$\vasreach(S)$.
\end{lemma}
\begin{proof}
	Note that for each $m\in[0,M]$, the set $R_m=\{\vec{t}\in\ZZ^p \mid |S[\vec{t}]\setminus T[\vec{t}]|=m\}$ is Presburger-definable. 
	For each $m\in[0,M]$, consider the Presburger-definable function \hbox{$\gamma_m\colon
	R_m\to \Z^m$}, which yields for each $\vec{t}$ the unique tuple $\tuple{f_1,\ldots,f_m}$ such that $S[\vec{t}]=T[\vec{t}]\cup\{f_1,\ldots,f_m\}$, and also
	$f_1<\cdots<f_m$.

	The reduction works as follows. 
	Given $\vec{t} \in \ZZ^p$ and an automaton $\cA$, we first compute $m\in[0,M]$ so that $\vec{t}\in R_m$. 
	This can be done in	logspace by~\cref{compute-presburger}. Second we use
	\cref{compute-presburger} to compute $\gamma_m(\vec{t}) = \tuple{f_1,\ldots,f_m}$. 
	Then, for every tuple $\vec{u} = \tuple{p_1,\ldots,p_m}$ of primes with \hbox{$p_i\in[0,O(\log(f_i))]$}, we
	construct a one-counter automaton $\cA_{\vec{u}}$ as follows. 
	$\cA_{\vec{u}}$	tracks the residue of the counter value modulo $p_i$ for each $i$ in
	its state. This only leads to a polynomial blow-up, because $m$ is a
	constant and $p_i$ is polynomial in the input. Finally, $\cA_{\vec{u}}$
	accepts only if the final counter value is $\not\equiv f_i\pmod{p_i}$. 

	We claim that $\cA$ can reach a value in $T[\vec{t}]$ if and only if
	for some $\vec{u}$, the automaton $\cA_{\vec{u}}$ can reach a value in
	$S[\vec{t}]$. 
	Indeed, if some $\cA_{\vec{u}}$ can reach $S[\vec{t}]$, then the target
	value must be distinct from each $x_i$. Conversely, by the Prime Number
	Theorem, the constant in the expression $O(\log(f_i))$ can be chosen so
	that for each tuple $\tuple{f_1,\ldots,f_m}$ and each value in $T[\vec{t}]$, there
	exists a tuple $\vec{u}$ so that $\cA_{\vec{u}}$ can reach a number in
	$S[\vec{t}]$.
\end{proof}

	\begin{proof}[Proof of \cref{vass-reduction}]
	Suppose $S[\vec{t}]$ is $\langle \delta,M \rangle$-upward closed for every $\vec{t}\in\N^p$. 
		Note that for every $m\in\{0,\ldots,M\}$, there is a Presburger formula $\psi_m$ for the set $T_m$ of all $\vec{t}$ such that $S[\vec{t}]$ is $\langle \delta,m \rangle$-upward closed, but not $\langle \delta,m-1\rangle$-upward closed. 
		Consider the function
		$\varphi_m\colon T_m \to\N^m$ which yields for each $\vec{t}$ the tuple $\langle f_1,\ldots,f_m \rangle$ of $m$ numbers such that $f_1<\cdots<f_m$ and $S'[\vec{t}]\coloneqq S[\vec{t}]\cup\{f_1,\ldots,f_m\}$ is $\delta$-upward closed. Then $\varphi_m$ is also Presburger-definable, and thus also the set $S'\subseteq\ZZ^{p+1}$ with
		\[ S'[\vec{t}]\coloneqq S[\vec{t}]\cup \{f_1,\ldots,f_m\}, \]
		where $\vec{t}\in T_m$ and $\varphi_m(\vec{t}) = \langle f_1,\ldots,f_m \rangle$. Now clearly, \hbox{$S'\supseteq S$} is $\delta$-upward closed and we have $|S'[\vec{t}]\setminus S[\vec{t}]|\le M$ for each $\vec{t}\in\ZZ^p$. Thus, by \cref{remove-bounded-set}, $S$ logspace-reduces to $S'$. Thus, it suffices to show that $\vasreach(S')$ logspace-reduces to coverability in $1$-VASS.

	For each $i\in\{0,\ldots,\delta-1\}$ consider the function $\mu_{i}\colon \ZZ^p\to\NN\cup\{-1\}$ with
		\[ \mu_{i}(\vec{t})=\begin{cases} \min [S'[\vec{t}]]_{\delta,i} & \text{if $[S'[\vec{t}]]_{\delta,i}\ne\emptyset$} \\
		-1 & \text{otherwise}
	\end{cases} \]
		In other words, $\mu_i(\vec{t})$ is the minimal number of the $i$-th residue class modulo $\delta$ that occurs in $S'[\vec{t}]$.

		Observe that since $S'[\vec{t}]$ is $\delta$-upward closed, for know that: For every $x\in\NN$ with $x\equiv i\bmod{\delta}$, we have $x\in S'[\vec{t}]$ if and only if $x\ge\mu_{i}(\vec{t})$. Since each function $\mu_{i}$ is Presburger-definable, it is computable in logspace by \cref{compute-presburger}. Thus, given an automaton and $\vec{t}$, we can consider each $i\in\{0,\ldots,\delta-1\}$ and check whether there is a run that covers $\mu_{i}(\vec{t})$ and reaches $i$'s congruence class modulo $\delta$.
		This reduces to coverability in $1$-VASS, because we can track the residue modulo $\delta$ of the current counter value in the state (recall that $\delta$ is a constant), and make only those states accepting where the residue is $i$.
	\end{proof}

\subsection{Proof of Claim~\ref{clm:functions-are-tables}}
\label{app:functions-are-tables}
\functionsAreTables*
\begin{proof}
	We can observe, by induction, that if $\pi$ is a path in $P$ and $i \in \NN$, then $f_\pi(i)$ is the greatest counter value that one can attain when executing $\pi$starting with counter value $i$; or $-\infty$ if $\pi$ cannot be executed from the configuration $p(i)$. 
	Therefore, for a finite set of paths $P$ from $p$ to $q$, and $i \in \NN$, $f_P(i)$ is therefore the greatest counter value one can attain at $q$ when starting at the configuration $p(i)$ and taking a path in $P$.
	If there does not exist a path in $P$ that can be executed from $i$ counter value, then $f_P(i) = -\infty$.
\end{proof}

\subsection{Proof of Lemma \ref{sigma-properties}}\label{app:sigma-properties}
\sigmaProperties*
\begin{proof}
	First, we know that $\sigma(f)\le f$ for every $f\in\calF$ by definition. 

	Let us now show that $\sigma$ is a homomorphism w.r.t.\ $+$. Recall that for any $f\in\calF$, we have $X^i\le f$ if and only if $f(0)\ge i$. Similarly, we have $\bar{X}^i\le f$ if and only if $f(i)\ge 0$.

	We need to check that $\sigma(f+g)=\sigma(f)+\sigma(g)$. Here, it is clear that
	$\sigma(f)\le \sigma(f+g)$ and $\sigma(g)\le\sigma(f+g)$, since $f\le f+g$. In particular, we have $\sigma(f)+\sigma(g)\le \sigma(f+g)$. It thus remains to show that $\sigma(f+g)\le\sigma(f)+\sigma(g)$. We distinguish three cases.
	\begin{enumerate}
	\item If $\sigma(f+g)=0$. Then $f+g=0$, and thus $f=0$ or $g=0$. In either case, $\sigma(f)+\sigma(g)=0$, hence $\sigma(f+g)\le \sigma(f)+\sigma(g)$.
	\item If $\sigma(f+g)=X^i$ for some $i\ge 0$. Then $f(i)\ge 0$ or $g(i)\ge 0$, and thus $X^i\le f$ or $X^i\le g$. This means $X^i\le\sigma(f)$ or $X^i\le\sigma(g)$, which in turn implies $\sigma(f)(i)\ge 0$ or $\sigma(g)(i)\ge 0$, and thus $(\sigma(f)+\sigma(g))(i)\ge 0$. This means, we have $X^i\le \sigma(f)+\sigma(g)$.
	\item If $\sigma(f+g)=\bar{X}^i$ for some $i\ge 0$. Then $f(0)\ge i$ or $g(0)\ge i$, and thus $\bar{X}^i\le f$ or $\bar{X}^i\le g$. This means $\bar{X}^i\le\sigma(f)$ or $\bar{X}^i\le\sigma(g)$, which in turn implies $\sigma(f)(0)\ge i$ or $\sigma(g)(0)\ge i$, and thus $(\sigma(f)+\sigma(g))(0)\ge i$. This means, we have $\bar{X}^i\le \sigma(f)+\sigma(g)$.\qedhere
	\end{enumerate}
\end{proof}

\subsection{Proof of Lemma \ref{update-algebraically}}\label{app:update-algebraically}
\updateAlgebraically*
\begin{proof}
	Note that $A_kA_kA_k$ contains, in entry $\langle i,j\rangle$, the sum of all $f_\pi$, where $\pi \in P_{i,j}^{=3}$ (where $P_{i,j}^{=3}$ is the set of all paths from state $i$ to state $j$ of length three).
	Since $\sigma$ is a homomorphism by \cref{sigma-properties}, we have
	\begin{equation*}
		(\sigma(A_kA_kA_k))_{i,j}
		= \sigma\left(\sum_{\pi \in P_{i,j}^{=3}} f_\pi\right) 
		= \sum_{\pi \in P_{i,j}^{=3}} \sigma(f_{\pi}).
	\end{equation*}
	Moreover, for any run $\pi$, it is easy to observe that $\sigma(f_\pi)$ is
	essentially the amplitude of $\pi$. More precisely, if $a$ is the amplitude of
	$\pi$, then $\sigma(f_\pi)=X^a$ if $a\ge 0$ and $\sigma(f_\pi)=\bar{X}^{|a|}$
	if $a<0$. Thus, the adjacency matrix of $\cA_{k+1}$, which consists of simple
	elements corresponding to amplitudes of runs of length $3$ in $\cA_k$, is
	exactly $\sigma(A_kA_kA_k)$.
\end{proof}

\subsection{Hardness and decidability for the VASS semantics}\label{app:vass-hardness-decidability}

In this subsection, we prove the second statement of \cref{main-result-vass}, and we prove \cref{main-result-vass-decidable}.

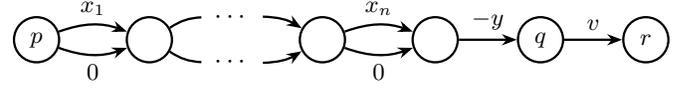
\begin{figure}
	\begin{tikzpicture}
	\node[circle, draw=black, line width = 0.3mm, minimum size = 6mm] (p) at (0,0) {\small$p$};
	\node[circle, draw=black, line width = 0.3mm, minimum size = 6mm] (q1) at (1.5,0) {};
	\node[circle, minimum size = 5mm] (q21) at (2.6, 0.3) {$\cdots$};
	\node[circle, minimum size = 5mm] (q22) at (2.6,-0.3) {$\cdots$};
	\node[circle, draw=black, line width = 0.3mm, minimum size = 6mm] (q3) at (3.8,0) {};
	\node[circle, draw=black, line width = 0.3mm, minimum size = 6mm] (q4) at (5.3,0) {\small$$};
	\node[circle, draw=black, line width = 0.3mm, minimum size = 6mm] (q) at (6.7,0) {\small$q$};
	\node[circle, draw=black, line width = 0.3mm, minimum size = 6mm] (r) at (8.1,0) {\small$r$};

	\path[draw=black, -{Stealth[width=1.5mm, length=2mm]}, line width = 0.3mm] (p) edge[bend left = 20] node[above]{\small$x_1$} (q1);
	\path[draw=black, -{Stealth[width=1.5mm, length=2mm]}, line width = 0.3mm] (p) edge[bend right = 20] node[below]{\small$0$} (q1);

	\path[draw=black, line width = 0.3mm] (q1) edge[bend left = 15] (q21);
	\path[draw=black, line width = 0.3mm] (q1) edge[bend right = 15] (q22);

	\path[draw=black, -{Stealth[width=1.5mm, length=2mm]}, line width = 0.3mm] (q21) edge[bend left = 15] (q3);
	\path[draw=black, -{Stealth[width=1.5mm, length=2mm]}, line width = 0.3mm] (q22) edge[bend right = 15] (q3);

	\path[draw=black, -{Stealth[width=1.5mm, length=2mm]}, line width = 0.3mm] (q3) edge[bend left = 20] node[above]{\small$x_n$} (q4);
	\path[draw=black, -{Stealth[width=1.5mm, length=2mm]}, line width = 0.3mm] (q3) edge[bend right = 20] node[below]{\small$0$} (q4);

	\path[draw=black, -{Stealth[width=1.5mm, length=2mm]}, line width = 0.3mm] (q4) edge node[above]{\small$-y$} (q);
	\path[draw=black, -{Stealth[width=1.5mm, length=2mm]}, line width = 0.3mm] (q) edge node[above]{\small$v$} (r);	
\end{tikzpicture}
	\caption{
		The $1$-VASS $\Vv$; here $v \geq 0$ is a variable.
		The instance of reachability in $\Vv$ from $\config{p}{0}$ to $\config{r}{v}$ is positive if and only if the instance of subset sum $\tuple{x_1, \ldots, x_n, y}$ is positive; see~\cref{clm:subset-sum-vas}.
	}
	\label{fig:vas-np-hard}
\end{figure}

For the hardness proof, we begin with an example.
\begin{example} \label{ex:vas-np-hard}
	We shall argue that the example semilinear target set $S \subseteq \NN^{p+1}$ with $p=1$ and 
	\begin{equation*}
		S[t] \coloneqq \set{0} \cup \set{ x \in \NN : x \geq t }
	\end{equation*}
	yields an \NP-hard reachability problem $\vasreach(S)$. 
	This example highlights the core combinatorial reason why $\vasreach(S)$ is \NP-hard for semilinear sets $S$ for which there exists $\delta \geq 1$, $M \geq 0$, and $\vec{t} \in \NN^p$ such that $S[\vec{t}]$ is not $\tuple{\delta, M}$-upward closed.

	Let $\tuple{x_1, \ldots, x_n, y}$ be an arbitrary instance of subset sum.
	Consider selecting $t > \sum_{i=1}^n x_i$ so that there is a `large' gap between $\set{0}$ and $\set{ x \in \NN : x \geq t }$ in $S[t]$.
	Also, consider the 1-VASS $\Vv$ in~\Cref{fig:vas-np-hard} with $v = 0$.
	With $v = 0$, it is clear that the configuration $\config{r}{x}$ that is reachable from $\config{p}{0}$ in $\Vv$ with the greatest counter values has $x < \sum_{i=1}^n x_i - y + v \leq \sum_{i=1}^n x_i$.
	Since $t$ was chosen so that $t > \sum_{i=1}^n x_i$, we know that the only potential counter value $x \in S[t]$ such that $\config{r}{x}$ is reachable from $\config{p}{0}$ is when $x=0$.
	However, it is clear that $\config{r}{0}$ is reachable from $\config{p}{0}$ in $\Vv$ if and only if there exists $I \sset \set{1, \ldots, n}$ such that $\sum_{i\in I} x_i = y$; for more details, see~\cref{clm:subset-sum-vas}.
	It is therefore \NP-hard to decide $\vasreach(S)$.
\end{example}

\begin{restatable}{claim}{subsetSumVas}\label{clm:subset-sum-vas}
	Let $x_1, \ldots, x_n, y, v \geq 0$.
	Consider 1-VASS $\Vv$ that is defined in~\Cref{fig:vas-np-hard}.
	There exists a run from $\config{p}{0}$ to $\config{r}{v}$ in $\Vv$ if and only if $\tuple{x_1, \ldots, x_n, y}$ is a positive instance of subset sum.
\end{restatable}
\begin{proof}
	First, we shall assume that there is a run from $\config{p}{0}$ to $\config{r}{v}$ in $\Vv$.
	From this run, construct a set $I \sset \set{1, \ldots, n}$ by adding $i$ to $I$ if the transition adding $x_i$ is taken (otherwise, $i$ is not added to $I$).
	Given that, at the end of the run, $\config{r}{v}$ is reached, we know that $\config{q}{0}$ must have been reached.
	Therefore, \hbox{$\sum_{i \in I} x_i - y = 0$}; this implies that the instance of subset sum is positive.

	Conversely, we shall assume that the instance of subset sum is positive.
	There exists $I \sset \set{1, \ldots, n}$ such that \hbox{$\sum_{i \in I} x_i = y$}.
	Now, from this set $I$, construct a run from $p(0)$ to $r(v)$ that takes, for each $i \in I$, the transition that adds $x_i$, and takes, for each $i \notin I$, the transition with zero effect.
\end{proof}

Let us slightly generalize the example. We say that \hbox{$S\subseteq\ZZ^{p+1}$} \emph{has unbounded gaps} if for every gap size $g\in\NN$, there is a parameter vector $\vec{t} \in \ZZ^p$ and a $u\in\NN$ such that $[u,u+g]\cap S[\vec{t}]=\{u\}$. In other words, $S[\vec{t}]$ contains $u$, but above $u$, there is a gap of size $g$.
\begin{lemma}\label{unbounded-gaps-np-hard}
	If $S\subseteq\ZZ^{p+1}$ has unbounded gaps, then $\vasreach(S)$ is $\NP$-hard.
\end{lemma}
\begin{proof}
	Since $S$ has unbounded gaps, there exists a function $\tau \colon\NN\to\ZZ^{p}\times\N$, which yields for each given $g\in\NN$, a vector $\vec{t} \in \ZZ^p$ and a value $u\in \NN$ such that $[u,u+g] \cap S[\vec{t}]=\set{u}$. 
	Moreover, if $\tau$ is the function that always picks the least vector $\tuple{\vec{t}, u}$ (w.r.t.\ some Presburger-definable well-order on $\ZZ^{p+1}$) with the above property, then $\tau$ is clearly 	Presburger-definable. 
	By~\cref{compute-presburger}, the function $\tau$ is computable in logspace. 

	This allows us to reduce subset sum to $\vasreach(S)$.
	Given a subset sum instance $\tuple{x_1, \ldots, x_n, y}$, we compute $\tuple{\vec{t},u}=\tau(x_1+\cdots+x_n)$ in logspace using~\cref{compute-presburger}.  
	Consider the $1$-VASS $\Vv$ that is presented in~\Cref{fig:vas-np-hard}, this time with \hbox{$v:=u$} set to the value given by $\pi$.
	Let \hbox{$R = \set{s : \Run{p(0)}{*}{\Vv}{r(s)}}$} be the set of reachable values.
	Observe that $s \in R$ then both $s \geq u$ and $s < u+x_1+\cdots+x_n$.
	Therefore, $R \cap S[\vec{t}] \sset [u, u+x_1+\cdots+x_n - 1] \cap S[\vec{t}] = \set{u}$.
	Finally, \cref{clm:subset-sum-vas} tells us that $u\in R$ if and only if $\tuple{x_1, \ldots, x_n, y}$ is a positive instance of subset sum.
\end{proof}

Let us now see that if $S$ is not uniformly quasi-upward closed, then it
contains unbounded gaps within some residue class. Note that since
$\vasreach([S]_{B,\vec{b}})$ reduces to $\vasreach(S)$, and \cref{unbounded-gaps-np-hard}
implies $\NP$-hardness of $\vasreach{S}$, meaning the next lemma
completes the \hbox{$\NP$-hardness} in \cref{main-result-vass}.

\begin{restatable}{lemma}{uniformlyQuasiUpwardClosedResidueClasses}\label{uniformly-quasi-upward-closed-residue-classes}
	Given a semilinear $S\sset\ZZ^p\times\NN$, we can compute a constant
	$B\ge 1$ such that $S$ is \emph{not} uniformly quasi-upward closed if
	and only if for some $\vec{b}\in[0,B-1]^{p+1}$, the set
	$[S]_{B,\vec{b}}$ has unbounded gaps.
\end{restatable}
Here, $B$ is the constant from \cref{lem:modulo-free}: One just needs to
observe that $S$ is uniformly quasi-upward closed if and only if each
$[S]_{B,\vec{b}}$ is uniformly quasi-upward closed. Moreover, a modulo-free set
is uniformly quasi-upward closed if and only if it has unbounded gaps. 
\begin{proof}[Proof of \cref{uniformly-quasi-upward-closed-residue-classes}]
	Let $B\ge 1$ be the constant from \cref{lem:modulo-free}. Then for every $\vec{s} \in [0, B-1]^{p+1}$, $[S]_{B, \vec{b}}$ is defined by a modulo-free Presburger formula.
	
	We first prove:
	\begin{claim}
		$S$ is uniformly quasi-upward closed if and only if for every $\vec{b}\in[0,B-1]^p$, the set $[S]_{B,\vec{b}}$ is uniformly quasi-upward closed.
	\end{claim}
	First, suppose $S$ is uniformly quasi-upward closed, say it is $\langle \delta,M\rangle$-upward closed. Then clearly for every $\vec{b}\in[0,B-1]^p$, the set $[S]_{B,\vec{b}}$ is also $\langle\delta,M\rangle$-upward closed.

	Conversely, suppose all residue classes are uniformly quasi-upward closed, then $S$ is uniformly quasi-upward closed. 
	Indeed, if $[S]_{B,\vec{b}}$ is $\tuple{\delta_{\vec{b}},M_{\vec{b}}}$-upward closed for some $\delta_{\vec{b}}\ge 1$ and $M_{\vec{b}}\ge 0$, then it is easy to observe that $S$ is $\tuple{\delta,M}$-upward closed, where $\delta$ is the product of all $\delta_{\vec{b}}$, and $M$ is the sum of all $M_{\vec{b}}$. 

	This proves our claim.

	It now remains to prove:
	\begin{claim}
		Suppose $T$ is modulo-free. Then $T$ is not uniformly
		quasi-upward closed if and only if $T$ has unbounded gaps.
	\end{claim}
	First, suppose $T$ is not uniformly quasi-upward closed. Since $T$ is defined by a modulo-free Presburger formula, let $\alpha_1, \beta_1, \ldots, \alpha_m, \beta_m: \ZZ^p \to \NN \cup \set{\infty}$ be the Presburger-definable functions as per~\cref{lem:intervals}.
	Recall that these functions have the following two properties.
	Firstly, for all $\vec{t} \in \ZZ^p$, 
	\begin{equation*}
		\alpha_1(\vec{t}) \leq \beta_1(\vec{t}) \leq \alpha_2(\vec{t}) \leq \ldots \leq \beta_{m-1}(\vec{t}) \leq \alpha_m(\vec{t}) \leq \beta_m(\vec{t}).
	\end{equation*}
	Secondly,
	\begin{equation*}
		T[\vec{t}] = (\alpha_1(\vec{t}), \beta_1(\vec{t})) \cup \cdots \cup (\alpha_m(\vec{t}), \beta_m(\vec{t})).
	\end{equation*}

	\begin{enumerate}[(1)]
		\item Suppose there is a $\vec{t}\in\ZZ^p$ such that $\beta_m(\vec{t}) \in \NN$. 
		In that case, $T=[S]_{B,\vec{b}}$ clearly has unbounded gaps. 
		Indeed, there is even a single infinite gap; that is $(\beta_m(\vec{t}), \infty)$.

		\item Suppose $\beta_m(\vec{t}) = \infty$ for every $\vec{t} \in \ZZ^p$. 
		Now consider, for any given $\vec{t} \in \ZZ^p$, the ``largest gap size'':
		\begin{equation*}
			M_{\vec{t}} = \max_{i\in[1,m-1]} \set{\alpha_{i+1}(\vec{t})-\beta_i(\vec{t})}.
		\end{equation*}
		Now if the function $\vec{t} \mapsto M_{\vec{t}}$ were bounded by some $M \geq 0$, then $T[\vec{t}]$ would be $\tuple{1, mM}$-upward closed for every $\vec{t} \in \ZZ^p$; this cannot be the case as we assumed that $S$ is not uniformly quasi-upward closed. 
		Therefore, for every $M \geq 0$, there exist $\vec{t} \in \ZZ^p$ and $i \in [1, m-1]$ such that $\alpha_{i+1}(\vec{t}) - \beta_i(\vec{t}) \geq M$.
		This means that $T = [S]_{B,\vec{b}}$ has unbounded gaps. 
	\end{enumerate}

	Conversely, if $T$ has unbounded gaps, then it cannot be $\tuple{\delta,M}$-upward closed for any $\delta\ge 1$, $M\ge 0$. 
	This is true because $T[\vec{t}]$ has a gap of size $g\coloneqq\delta+M+1$, injecting
	any set of $M$-many points into $T$ will still leave a gap of
	$\delta+1$, so that the result cannot be $\delta$-upward closed. This
	completes the proof of the second claim.

	Together the two claims establish the \lcnamecref{uniformly-quasi-upward-closed-residue-classes}.
\end{proof}

Finally, \cref{uniformly-quasi-upward-closed-residue-classes} implies
\cref{main-result-vass-decidable}, since it reduces checking uniformly
quasi-upward closedness to deciding the unbounded gaps property. The latter, in
turn, is directly expressible in Presburger arithmetic and thus decidable.

\label{afterbibliography}
\newoutputstream{pagestotal}
\openoutputfile{main.pagestotal.ctr}{pagestotal}
\addtostream{pagestotal}{\getpagerefnumber{afterbibliography}}
\closeoutputstream{pagestotal}

\newoutputstream{todos}
\openoutputfile{main.todos.ctr}{todos}
\addtostream{todos}{\arabic{@todonotes@numberoftodonotes}}
\closeoutputstream{todos}

\end{document}